\documentclass[review,times,number]{elsarticle}
\journal{Theoretical Computer Science}
\def\template{elsarticle}

\usepackage{common/prelude}
\usepackage{multicol}
\usepackage{subcaption}
\usepackage{dsfont}
\usepackage{bbm}
\usepackage{nicefrac}
\usepackage[normalem]{ulem}
\usepackage{standalone}
\usepackage{circuitikz}
\usepackage{tikz-qtree}
\usepackage{enumitem}
\usepackage{multirow}
\usetikzlibrary{backgrounds, calc, positioning, arrows.meta, fit, decorations.pathreplacing, matrix, shapes.symbols}

\ifcsname problem\endcsname

\renewcommand{\problem}[1]{\textsc{#1}\xspace}
\newenvironment{problemStatement}[1]{\begin{oldproblem}[#1]}{\end{oldproblem}}
\else
\newcommand{\problem}[1]{\textsc{#1}\xspace}
\theoremstyle{remark}
\newtheorem*{problemdef}{Problem}
\newenvironment{problemStatement}[1]{\begin{problemdef}[#1]}{\end{problemdef}}
\fi
\newcommand{\kcore}{$k$-core\xspace}
\newcommand{\kcores}{$k$-cores\xspace}
\newcommand{\klcore}{$(k, l)$-core\xspace}

\newcommand{\ktruss}{$k$-truss\xspace}

\newcommand{\dynTime}[3]{\langle #1,\allowbreak #2,\allowbreak #3 \rangle}
\newcommand{\dynOh}[3]{\dynTime{\Oh(#1)}{\Oh(#2)}{\Oh(#3)}}

\DeclareMathOperator{\polylog}{polylog}
\newcommand{\1}{\mathds{1}}
\newcommand{\eps}{\epsilon}
\newcommand{\Bzo}{\{0,1\}}
\newcommand{\BOR}{{\textsc{\small OR}}\xspace}
\newcommand{\BAND}{{\textsc{\small AND}}\xspace}

\newcommand{\DSP}{\problem{DSP}}
\newcommand{\threeCore}{\problem{3-Core}}
\newcommand{\CoreValue}{\problem{CoreValue}}
\newcommand{\TrussValue}{\problem{TrussValue}}

\newcommand{\twoCore}{\problem{Full-2-Core}}

\newcommand{\fourTruss}{\problem{4-Truss}}
\newcommand{\KcoreMaint}{\problem{FCM}}
\newcommand{\TrussMaint}{\problem{FTM}}

\newcommand{\kSAT}{\problem{$k$-SAT}}
\newcommand{\MCV}{\problem{MCVP}}

\newcommand{\A}{\ensuremath{\mathcal{A}}\xspace}

\newcommand{\complexityClass}[1]{\ensuremath{\mathbf{#1}}\xspace}
\renewcommand{\P}{\complexityClass{P}}
\newcommand{\NP}{\complexityClass{NP}}
\newcommand{\NC}{\complexityClass{NC}}
\newcommand{\IPL}{\complexityClass{IPL}}
\newcommand{\OMv}{\problem{OMv}}
\newcommand{\OuMv}{\problem{OuMv}}

\newcommand{\OMvlink}{\hyperref[conj:omv]{OMv}\xspace}
\newcommand{\OMvConj}{\hyperref[conj:omv]{OMv conjecture}\xspace}
\newcommand{\SETH}{\hyperref[conj:seth]{SETH}\xspace}
\newcommand{\ApproxCore}[1]{\problem{#1-ApproxCoreValue}}

\undef{\todo}
\usepackage[commandnameprefix=ifneeded,todonotes={textsize=footnotesize}]{changes}
\usepackage{soul}
\sethighlightmarkup{\IfIsColored{{\sethlcolor{authorcolor!30}\hl{#1}}}{#1}}

\definechangesauthor[name=Yan, color=orange]{yan}
\definechangesauthor[name=Cristina, color=red]{cris}

\undef\citet 
\usepackage{expl3}
\usepackage{environ}

\ExplSyntaxOn

\cs_new_eq:NN \elsarticle_template_old_author \author
\cs_new_eq:NN \elsarticle_template_old_title \title
\cs_new_eq:NN \elsarticle_template_old_affiliation \affiliation

\seq_new:N \g_elsarticle_template_author_names_seq
\seq_new:N \g_elsarticle_template_author_affiliations_seq
\seq_new:N \g_elsarticle_template_author_emails_seq
\seq_new:N \g_elsarticle_template_author_orcids_seq
\seq_new:N \g_elsarticle_template_author_funding_seq

\tl_new:N \g_elsarticle_template_title_tl
\tl_new:N \g_elsarticle_template_short_title_tl
\tl_new:N \g_elsarticle_template_keywords_tl
\tl_new:N \l_elsarticle_template_affiliation_tl

\cs_set_protected:Npn \author #1#2#3#4#5
{
  \tl_set:Nx \l_elsarticle_template_affiliation_tl {#2}
  \seq_put_right:Nn \g_elsarticle_template_author_names_seq {#1}
  \seq_put_right:NV \g_elsarticle_template_author_affiliations_seq \l_elsarticle_template_affiliation_tl
  \seq_put_right:Nn \g_elsarticle_template_author_emails_seq {#3}
  \seq_put_right:Nn \g_elsarticle_template_author_orcids_seq {#4}
  \seq_put_right:Nn \g_elsarticle_template_author_funding_seq {#5}
}

\RenewDocumentCommand \title { O{} m }
{
  \tl_set:Nn \g_elsarticle_template_short_title_tl {#1}
  \tl_set:Nn \g_elsarticle_template_title_tl {#2}
}

\cs_if_exist:NT \keywords
{
  \cs_undefine:N \keywords
}
\cs_new_protected:Npn \keywords #1
{
  \tl_set:Nn \g_elsarticle_template_keywords_tl {#1}
}

\ExplSyntaxOff

\ExplSyntaxOn
\cs_set:Npn \institution #1 { organization={#1}, }
\cs_set:Npn \city #1 { city={#1}, }
\cs_set:Npn \state #1 { state={#1}, }
\cs_set:Npn \country #1 { country={#1} }
\ExplSyntaxOff
\providecommand{\lastnameauthors}[1]{}
\providecommand{\firstlastnameauthors}[1]{}
\providecommand{\Description}[1]{}
\providecommand{\ccsdesc}[2][]{}
\NewEnviron{CCSXML}{}{}

\ifcsname citet\endcsname
\newcommand{\defineciteauthor}[2]{}
\else
\newcommand{\citet}[1]{%
  \ifcsname citeauthor@#1\endcsname
    \csname citeauthor@#1\endcsname~\cite{#1}%
  \else
    \textbf{[Missing author]}~\cite{#1}%
    \PackageWarning{common-template}{Missing author for #1.}
  \fi
}
\newcommand{\defineciteauthor}[2]{%
  \expandafter\def\csname citeauthor@#1\endcsname{#2}%
}
\fi

\ifdefined\MSCcodes \else
\newcommand{\msccodes}[1]{}%
\fi

\newboolean{isPortuguese}
\newboolean{clearPageAfterAbstract}

\setboolean{isPortuguese}{false}
\setboolean{clearPageAfterAbstract}{false}

\title{Hardness of Dynamic Core and Truss Decompositions}
\ifthenelse{\equal{\template}{acm-conference}}{%
\titlenote{A preliminary version of this work was presented at the Workshop on Approximation and Online Algorithms (WAOA) 2025~\cite{couto_hardness_2025}.}}{}
\date{April 2026}

\ifthenelse{\equal{\template}{lcns}}{%
    \author{Yan S. Couto\inst{1}\orcidID{0009-0005-5850-8696} \and
    Cristina~G.~Fernandes\inst{1}\orcidID{0000-0002-5259-2859}}
    \institute{University of São Paulo, São Paulo, SP, Brazil \\
    \email{\{yancouto,cris\}@ime.usp.br}}
}{%
    \author
    {Yan S. Couto} 
    { 
      \institution{University of São Paulo}%
      \city{São Paulo}%
      \state{SP}%
      \country{Brazil}%
    }{yancouto@ime.usp.br} 
    {0009-0005-5850-8696} 
    {} 

    \author
    {Cristina G. Fernandes}
    {
      \institution{University of São Paulo}%
      \city{São Paulo}%
      \state{SP}%
      \country{Brazil}%
    }{cris@ime.usp.br}
    {0000-0002-5259-2859}
    {}
}
\lastnameauthors{Y.\,S. Couto and C.\,G. Fernandes}
\firstlastnameauthors{Yan S. Couto and Cristina G. Fernandes}

\ifdefined\acmCodeLink
\acmCodeLink{https://github.com/yancouto/phd/tree/main/dynamic_2core}
\fi

\ifthenelse{\equal{\template}{lcns}}{%
}{%
\keywords{community search, dynamic graphs, core decomposition, fine-grained complexity}
}

\ccsdesc[500]{Theory of computation~Dynamic graph algorithms}
\ccsdesc[500]{Theory of computation~Problems, reductions and completeness}
\ccsdesc[300]{Theory of computation~Cell probe models and lower bounds}

\msccodes{
  68Q17, 
  05C85  
}

\makeatletter
\@ifpackageloaded{babel}{}{
\ifthenelse{\boolean{isPortuguese}}{ 
  \usepackage[brazil]{babel}
} { 
  \usepackage[english]{babel}
}}
\makeatother

\begin{document}

\ExplSyntaxOn

\seq_new:N \l_elsarticle_template_unique_affiliations_seq
\seq_new:N \l_elsarticle_template_author_affiliation_idx_seq
\prop_new:N \l_elsarticle_template_affiliation_to_idx_prop

\int_new:N \l_elsarticle_template_loop_int
\tl_new:N \l_elsarticle_template_current_affiliation_tl
\tl_new:N \l_elsarticle_template_current_affiliation_idx_tl
\tl_new:N \l_elsarticle_template_render_affiliation_tl
\tl_new:N \l_elsarticle_template_author_name_tl
\tl_new:N \l_elsarticle_template_author_affiliation_idx_tl
\tl_new:N \l_elsarticle_template_author_email_tl
\int_new:N \l_elsarticle_template_author_count_int
\seq_new:N \l_elsarticle_template_emit_affiliations_seq

\cs_new_protected:Npn \elsarticle_template_emit_author:nn #1#2
{
  \tl_if_blank:nTF {#2}
  {
    \elsarticle_template_old_author{#1}
  }
  {
    \elsarticle_template_old_author[#2]{#1}
  }
}

\cs_new_protected:Npn \elsarticle_template_emit_affiliation:nn #1#2
{
  \elsarticle_template_old_affiliation[#1]{#2}
}

\begin{frontmatter}

\tl_if_blank:VTF \g_elsarticle_template_short_title_tl
{
  \elsarticle_template_old_title{\g_elsarticle_template_title_tl}
}
{
  \elsarticle_template_old_title[\g_elsarticle_template_short_title_tl]{\g_elsarticle_template_title_tl}
}

\int_set:Nn \l_elsarticle_template_author_count_int { \seq_count:N \g_elsarticle_template_author_names_seq }
\int_step_variable:nnNn { 1 } { \l_elsarticle_template_author_count_int } \l_elsarticle_template_loop_int
{
  \seq_pop_left:NN \g_elsarticle_template_author_names_seq \l_elsarticle_template_author_name_tl
  \seq_pop_left:NN \g_elsarticle_template_author_affiliations_seq \l_elsarticle_template_current_affiliation_tl
  \seq_pop_left:NN \g_elsarticle_template_author_emails_seq \l_elsarticle_template_author_email_tl
  \tl_set:Nx \l_elsarticle_template_current_affiliation_tl
  {
    \tl_trim_spaces:V \l_elsarticle_template_current_affiliation_tl
  }
  \tl_if_blank:VTF \l_elsarticle_template_current_affiliation_tl
  {
    \tl_clear:N \l_elsarticle_template_author_affiliation_idx_tl
  }
  {
    \prop_get:NVNTF
      \l_elsarticle_template_affiliation_to_idx_prop
      \l_elsarticle_template_current_affiliation_tl
      \l_elsarticle_template_current_affiliation_idx_tl
    {
      \tl_set_eq:NN \l_elsarticle_template_author_affiliation_idx_tl \l_elsarticle_template_current_affiliation_idx_tl
    }
    {
      \tl_set:Nx \l_elsarticle_template_current_affiliation_idx_tl
      {
        l\int_eval:n { \seq_count:N \l_elsarticle_template_unique_affiliations_seq + 1 }
      }
      \prop_put:NVV
        \l_elsarticle_template_affiliation_to_idx_prop
        \l_elsarticle_template_current_affiliation_tl
        \l_elsarticle_template_current_affiliation_idx_tl
      \seq_put_right:NV
        \l_elsarticle_template_unique_affiliations_seq
        \l_elsarticle_template_current_affiliation_tl
      \tl_set_eq:NN \l_elsarticle_template_author_affiliation_idx_tl \l_elsarticle_template_current_affiliation_idx_tl
    }
  }
  \exp_args:NVV
    \elsarticle_template_emit_author:nn
    \l_elsarticle_template_author_name_tl
    \l_elsarticle_template_author_affiliation_idx_tl

  \tl_set:Nx \l_elsarticle_template_author_email_tl
  {
    \tl_trim_spaces:V \l_elsarticle_template_author_email_tl
  }
  \tl_if_blank:VF \l_elsarticle_template_author_email_tl
  {
    \ead{\l_elsarticle_template_author_email_tl}
  }
}

\seq_set_eq:NN \l_elsarticle_template_emit_affiliations_seq \l_elsarticle_template_unique_affiliations_seq
\int_step_inline:nn { \seq_count:N \l_elsarticle_template_emit_affiliations_seq }
{
  \seq_pop_left:NN \l_elsarticle_template_emit_affiliations_seq \l_elsarticle_template_render_affiliation_tl
  \tl_set:Nx \l_elsarticle_template_current_affiliation_idx_tl { l\int_to_arabic:n { #1 } }
  \exp_args:NVV
    \elsarticle_template_emit_affiliation:nn
    \l_elsarticle_template_current_affiliation_idx_tl
    \l_elsarticle_template_render_affiliation_tl
}

\ExplSyntaxOff
\begin{abstract}
The \kcore of a graph is its maximal subgraph with minimum degree at least~$k$, and the core value of a vertex~$u$ is the largest~$k$ for which~$u$ is contained in the~\kcore of the graph. Among cohesive subgraphs,~\kcore and its variants have received a lot of attention recently, particularly on dynamic graphs, 
as reported by Hanauer,  Henzinger, and Schulz in their recent survey on dynamic graph algorithms. 
We answer questions on $k$-core stated in the survey, proving that there is no efficient dynamic algorithm for~\kcore or to find~$(2 - \epsilon)$-approximations for the core values, unless we can improve decades-old state-of-the-art algorithms in many areas including matrix multiplication and satisfiability, based on the established OMv conjecture and SETH.
Our results show that there is no dynamic algorithm for $k$-core asymptotically faster than the trivial ones, and we derive them via proving hardness results for a dynamic version of the circuit value problem. This explains why most recent research papers in this area focus not on a generic efficient dynamic algorithm, but on finding a bounded algorithm, which is fast when few core values change per update. However, we also prove that a bounded algorithm does not exist, based on the OMv conjecture.
We obtain lower bounds for more variants of~$k$-core, such as on directed and bipartite graphs, as well as for the edge variant, known as $k$-truss. We also show no planarizing gadget for~$k$-core exists, so a proof of hardness for planar~$k$-core is still out of reach.
On the positive side, we present a polylogarithmic dynamic algorithm for 2-core.

\ifthenelse{\equal{\template}{lcns}}{%
\keywords{community search \and dynamic graphs \and core decomposition}}
{%
}

\balance
\end{abstract}
\ExplSyntaxOn

\tl_if_blank:VF \g_elsarticle_template_keywords_tl
{
  \begin{keyword}
    \g_elsarticle_template_keywords_tl
  \end{keyword}
}

\end{frontmatter}

\ExplSyntaxOff


\defineciteauthor{hanauer_recent_2022}{Hanauer et al.}
\defineciteauthor{abboud_popular_2014}{Abboud and Williams}
\defineciteauthor{liu_parallel_2022}{Liu et al.}
\defineciteauthor{li_efficient_2014}{Li et al.}
\defineciteauthor{sariyuce_incremental_2016}{Sarıyüce et al.}
\defineciteauthor{zhang_fast_2017}{Zhang et al.}
\defineciteauthor{zhang_unboundedness_2019}{Zhang and Yu}
\defineciteauthor{tian_maximal_2023}{Tian et al.}
\defineciteauthor{anderson_p-complete_1984}{Anderson and Mayr}
\defineciteauthor{miltersen_complexity_1994}{Miltersen et al.}
\defineciteauthor{sun_fully_2020}{Sun et al.}
\defineciteauthor{giatsidis_d-cores_2011}{Giatsidis et al.}
\defineciteauthor{larsen_super-logarithmic_2023}{Larsen and Yu}
\defineciteauthor{ahmed_visualisation_2007}{Ahmed et al.}
\defineciteauthor{sawlani_near-optimal_2020}{Sawlani and Wang}
\defineciteauthor{henzinger_fine-grained_2022}{Henzinger, Paz and Sricharan}

\section{Introduction}

Graphs are widely used to model data in different applications, and the search for communities, through the extraction of cohesive subgraphs, has a variety of applications in the analysis of social, biological, and other networks~\cite{huang_community_2019}. Most classic cohesive subgraphs, such as maximum cliques, quasi-cliques, and~$k$-clans are NP-hard to find or sometimes even to approximate~\cite{hastad_clique_1999}. A looser definition of a cohesive subgraph of a graph is the~\kcore, which is the maximum subgraph of minimum degree at least~$k$~\cite{seidman_network_1983}. There exists an efficient~$\Oh(m)$ algorithm to compute the cores of a graph on~$m$ edges~\cite{batagelj_m_2003}. Due to this, cores have been widely used in many community search applications, and they have been studied extensively~\cite{malliaros_core_2020,du_core_2020}. In the~\deff{core decomposition} problem, the task is to compute, for each vertex~$v$, the maximum~$k$ for which~$v$ is in the~\kcore, called the~\deff{core value} of~$v$. We explore several variants of this problem, such as approximating the core value within an interval, and searching for the~$k$-truss, a subgraph where all edges are in at least~$k-2$ triangles~\cite{cohen_trusses_2008}.

In many real-world applications, data is constantly changing, and thus there has been increased interest in dynamic graph algorithms, where vertices/edges can be inserted/deleted over time, and one wishes to answer queries on the graph that is being modified. 
It is undesirable, at each query, to apply a static algorithm from scratch on the current graph. 
\citet{hanauer_recent_2022} recently presented a broad survey on dynamic graph problems that addresses, among others, the~\kcore, summarizing what was known and stating open questions.
There exists a dynamic algorithm for 
the~$(4+\epsilon)$-approximate 
core decomposition problem that takes polylogarithmic time per graph update and query~\cite{sun_fully_2020}. 
Mostly, research has turned to the~\deff{core maintenance} problem, in which the core value for all vertices must be maintained explicitly after each update (edge insertion or deletion) in the graph. For this problem, research focused on polynomially bounded algorithms~\cite{li_efficient_2014,sariyuce_incremental_2016,zhang_fast_2017}, whose running time per update is bounded by a polynomial in the size of the closed neighborhood of the vertices whose core value changed due to the update. This is under the assumption that, in practice, on real dynamic graphs, the core values of vertices change very rarely, though there are examples of inputs in which all core values change with every update~\cite{sun_fully_2020}.

Parallel to the search for efficient algorithms, it is desirable to find good lower bounds to either prove that an algorithm is optimal or that it is pointless to look for better algorithms for some problems. 
The survey by \citet{hanauer_recent_2022} explicitly notes that there are no known lower bounds for the~\kcore and related problems matching their upper bounds. The only known hardness result is that the core maintenance problem is unbounded considering only edge insertions, and only when under a specific limited model of computation~\cite{zhang_unboundedness_2019}.

For other dynamic graph problems, lower bounds are usually achieved by using well-established conjectures.
The OMv conjecture~\cite{henzinger_unifying_2015} states that there is no truly subcubic algorithm for ``online'' Boolean matrix multiplication, even when one has arbitrary polynomial preprocessing time. 
The OMv conjecture is as strong as many other conjectures around matrix multiplication, triangle detection in graphs, and problems such as \problem{3SUM}~\cite{huang_community_2019}. It is widely used in dynamic graph problems to give conditional lower bounds, such as for reachability, shortest path, and triangle detection problems~\cite{hanauer_recent_2022,henzinger_unifying_2015}.
The well-known SETH~\cite{impagliazzo_complexity_2001} states that there is no truly sub-exponential algorithm for~\kSAT. This hypothesis has been used to prove complexity results for both NP-hard~\cite{cygan_problems_2016} and polynomial~\cite{abboud_popular_2014,patrascu_possibility_2010} problems, including many dynamic graph problems~\cite{abboud_popular_2014}.
Falsifying the OMv conjecture or the SETH would improve decade-long state-of-the-art algorithms in several areas. See~\Cref{sec:undirected_omv_lower,sec:undirected_seth_lower} for the full statements of the OMv conjecture and the SETH, respectively.

\subsection{Our contributions}

We show lower bounds for the dynamic core decomposition problem, as well as for its~$(2-\epsilon)$-approximate version, and its directed, bipartite, and~\ktruss variants, with the following results:
\begin{itemize}
    \item These problems cannot be solved in sub-logarithmic time per update.
    \item These problems cannot be solved efficiently (in polylog time per update and query) unless dynamic versions of any problem in~$\P$ can be solved efficiently.
    \item Under the OMv conjecture, these problems require~$\Omega(m^{\frac12-\epsilon})$ time per update or~$\Omega(m^{1-\epsilon})$ time per query.
    \item Under the SETH, these problems require~$\Omega(m^{1-\epsilon})$ time per update or  query.
    \item Under the OMv conjecture, no bounded algorithm exists for dynamic core or truss maintenance.
\end{itemize}

The complexities stated here assume~$k$ is a constant. For the generalized results when~$k$ might depend on the size of the graph, see~\Cref{res:kcore_bound_all}.
We also show a polylog algorithm for answering the query on whether the core value of a given vertex in a dynamic graph is at least 2, so our reductions are tight in a sense, as they apply to the~\kcore for any~$k \geq 3$.
For the planar version of~$k$-core, we show that a planarizing gadget does not exist, which implies our techniques cannot be directly used to prove hardness of the planar~$k$-core.

\subsection{Our techniques}

The main technical contribution of this work is to obtain a set of lower bounds for dynamic graph problems via reductions from a dynamic version of the classic \emph{circuit value problem}, in which a circuit of AND-, OR-, 0- and 1-gates must be evaluated under dynamic updates to its wires. The circuit value problem is said to be \emph{inherently sequential}, or~\P-complete, and has been used extensively to prove lower bounds in parallel computation~\cite{greenlaw_limits_1995}.

We show it is also possible to obtain lower bounds for dynamic computation from this problem, using several different techniques. Unconditional lower bounds are obtained by using the cell probe model, via reduction from the dynamic XOR problem. Recent literature has shown dynamic lower bounds using famous conjectures such as OMv or SETH, and we also derive such lower bounds for circuit value via reductions from matrix multiplication and the satisfiability problem. We also bring to light a complexity result which has not been used in the recent literature for dynamic lower bounds, stating that circuit value is complete for~\P under dynamic reductions, that is, it is as hard as any dynamic problem in~\P.

With these results in hand, we reduce the circuit value problem to the dynamic core decomposition and related problems, which allows us to derive all the same lower bounds for them. The main benefit from this technique is reusing well-known complexity results when possible, instead of developing them from scratch in a slightly different model. We believe this technique has great potential to be used to prove lower bounds for other dynamic graph problems as well.

\subparagraph*{Organization.}
The rest of the paper is organized as follows. \Cref{sec:preliminaries} provides some preliminaries. \Cref{sec:mcv_to_3core} describes reductions from a dynamic version of the classic circuit value problem to our core and truss decomposition problems.
\Cref{sec:all_lower} proves the stated lower bounds, while \Cref{sec:kcore_unboundedness} presents the unboundedness result for core and truss maintenance. \Cref{sec:2core_alg} presents a polylog algorithm for the~$2$-core and \Cref{sec:no_crossing} proves the nonexistence of a planarizing gadget for~\kcore. Finally, \Cref{sec:conclusion} summarizes the work and highlights remaining open problems in the area.

\subsection{Related work}
We organize the related work by topics.

\subparagraph*{Dynamic core decomposition.}
A polylogarithmic algorithm for the~$(4+\epsilon)$-approximate core decomposition on dynamic graphs~\cite{sun_fully_2020} uses similar techniques previously applied to dynamic densest subgraphs~\cite{hu_maintaining_2017}. Another algorithm for the~$(4+\epsilon)$-approximate core decomposition was presented by \citet{liu_parallel_2022}, with a focus on parallelism and batching queries.

\subparagraph*{Core maintenance.}
A bounded deletions-only algorithm for core maintenance~\cite{zhang_fast_2017} exists, but no bounded algorithm supporting insertions exists. There are, however, progressively better heuristics aiming for boundedness: \citet{li_efficient_2014} showed an algorithm based on the fact that each insertion or deletion changes by one the core value of at most a connected set of vertices. This was later optimized by \citet{sariyuce_incremental_2016} with an elaborate graph traversal, and further optimized by \citet{zhang_fast_2017} by explicitly storing a deletion order for the vertices using an ordered-list data structure. 
This idea has since been reused~\cite{guo_simplified_2024} and applied broadly in the area, to parallel algorithms~\cite{guo_parallel_2023}, weighted graphs~\cite{zhang_order_2023}, hierarchical core decompositions~\cite{gabert_batch_2022}, partial maintenance problems~\cite{zhang_maintaining_2024}, truss maintenance problems~\cite{zhang_unboundedness_2019},~etc.

\subparagraph*{Lower bounds and unboundedness.}
\citet{zhang_unboundedness_2019} proved the unboundedness of core and truss maintenance problems under the locally persistent computation model, a very limited model in which an algorithm can store data only on nodes and move only from a node to an adjacent one. This result was extended to the directed case of truss maintenance by \citet{tian_maximal_2023}. Even though lower bounds for dynamic graph problems have been studied extensively~\cite{abboud_popular_2014,hanauer_recent_2022,henzinger_unifying_2015}, to our knowledge, there is no such result for any dynamic core decomposition variant.

\subparagraph*{Densest subgraph.}
In the densest subgraph problem (\DSP), one must find a subgraph of maximum density, defined as half its average degree. This problem is closely related to the~$k$-core problem, as a~$k$-core has density at least~$\frac{k}{2}$. In fact, the~$k$-core is a~2-approximation for the densest subgraph, so the classic~$\Oh(m)$ algorithm for~$k$-core is an algorithm for the~2-approximate \DSP, though there exist polynomial algorithms for exact \DSP~\cite{goldberg_finding_1984}. In the dynamic setting,~\citet{henzinger_fine-grained_2022} obtained polynomial lower bounds on the running time per update for exact \DSP, while~\citet{sawlani_near-optimal_2020} showed an algorithm for the~$(1+\epsilon)$-approximate \DSP with polylogarithmic running time per update. Therefore, it is surprising that we show here that no similar~$(1+\epsilon)$-approximation for dynamic~$k$-core exists.

\section{Preliminaries} \label{sec:preliminaries}

We use~$G = (V(G), E(G))$ to denote a graph, and~$uv$ to denote an edge~$\{u, v\} \in E(G)$. The number of edges incident to a vertex is its~\deff{degree}. For directed graphs (digraphs), each vertex has both an~\deff{in-degree} and an~\deff{out-degree}, for incoming and outgoing arcs, respectively.
When the (di)graph is clear from the context, we use~$n$ and~$m$ for its number of vertices and edges (arcs), respectively. For complexity functions, we assume~$m = \Omega(n)$.

\subparagraph*{Core decomposition problems.}
A~\deff{\kcore} of a graph~$G$ is a maximal subgraph of~$G$ with minimum degree at least~$k$. The~\kcore is unique, otherwise the union of two~\kcores would be a larger~\kcore. 
The~\deff{core value} of a vertex~$u$, denoted as~$K_u$, is the largest~$k$ for which~$u$ is contained in the~\kcore. \Cref{fig:kcore} shows an example of the cores of a graph. The~\deff{core decomposition problem} is the task of determining~$K_u$ for each vertex~$u$, and it can be solved in~$\Oh(m)$ time~\cite{batagelj_m_2003} by repeatedly removing a lowest degree vertex.

\begin{figure}[h]
    \centering
    \begin{tikzpicture}[
    vertex/.style={circle, draw=black, fill=white, thick, inner sep=2pt, minimum size=7pt},
    core1/.style={fill=green!20, rounded corners=20pt},
    core2/.style={fill=blue!20, rounded corners=15pt},
    core3/.style={fill=red!20, rounded corners=12pt}
]

\node[vertex] (v1) at (0,0)   {H};
\node[vertex] (v2) at (2,0)   {G};
\node[vertex] (v3) at (0,2)   {F};
\node[vertex] (v4) at (2,2)   {E};

\node[vertex] (v5) at (4,0.5) {D};
\node[vertex] (v6) at (4,1.5) {C};

\node[vertex] (v7) at (5.5,1) {B};
\node[vertex] (v8) at (-1.5,1){A};

\pgfdeclarelayer{background}
\pgfsetlayers{background,main}

\begin{pgfonlayer}{background}
    \fill[core1] ($(v8.west)+(-0.4,0)$) 
              -- ($(v3.north)+(-0.2,0.4)$) 
              -- ($(v4.north)+(0.2,0.4)$) 
              -- ($(v6.north)+(0.4,0.4)$) 
              -- ($(v7.east)+(0.4,0)$) 
              -- ($(v5.south)+(0.4,-0.4)$) 
              -- ($(v2.south)+(0.2,-0.4)$)
              -- ($(v1.south)+(-0.2,-0.4)$)
              -- cycle;

    \fill[core2] ($(v1.south west)+(-0.3,-0.3)$) 
              -- ($(v2.south east)+(0.3,-0.3)$) 
              -- ($(v5.south east)+(0.3,-0.3)$) 
              -- ($(v6.north east)+(0.3,0.3)$) 
              -- ($(v4.north west)+(-0.3,0.3)$) 
              -- ($(v3.north west)+(-0.3,0.3)$) 
              -- cycle;

    \fill[core3] ($(v1.south west)+(-0.2,-0.2)$) 
              -- ($(v2.south east)+(0.2,-0.2)$) 
              -- ($(v4.north east)+(0.2,0.2)$) 
              -- ($(v3.north west)+(-0.2,0.2)$) 
              -- cycle;
\end{pgfonlayer}

\draw[thick] (v1) -- (v2);
\draw[thick] (v1) -- (v3);
\draw[thick] (v1) -- (v4);
\draw[thick] (v2) -- (v3);
\draw[thick] (v2) -- (v4);
\draw[thick] (v3) -- (v4);

\draw[thick] (v2) -- (v5);
\draw[thick] (v4) -- (v6);
\draw[thick] (v5) -- (v6);

\draw[thick] (v6) -- (v7);
\draw[thick] (v1) -- (v8);

\end{tikzpicture}
    \caption{A graph with its \textcolor{green!80!black}{1-core}, \textcolor{blue}{2-core}, and \textcolor{red}{3-core} highlighted.}
    \Description{Example graph with the 1-core, 2-core, and 3-core highlighted.}
    \label{fig:kcore}
\end{figure}

\subparagraph*{Dynamic graph problems.}
In a~\deff{dynamic graph problem}, we are given a graph subject to edge insertions/deletions, and we have to answer queries about the current state of the graph.
We are interested in~\deff{partially} and~\deff{fully dynamic}  variants of the problems. Partially dynamic variants might be~\deff{incremental} (insertions-only) or~\deff{decremental} (deletions-only), while fully dynamic variants allow both insertions and deletions.
The following is a dynamic version of the core decomposition problem:

\begin{problemStatement}{\CoreValue}
    Given a graph, process a series of edge insertions/deletions, and queries: Given a vertex~$u$, what is the value of~$K_u$?
\end{problemStatement}

The running time of a dynamic graph algorithm is~$\dynTime{p(n, m)}{u(n, m)}{q(n, m)}$ if~$p(n, m)$ is the time the algorithm takes preprocessing a graph on~$n$ vertices and~$m$ edges before receiving any updates or queries,~$u(n, m)$ is the time per update, and~$q(n, m)$ is the time per query.

We say the running time of a dynamic graph algorithm is amortized if the update and the query running times are amortized over all updates and queries. Specifically, the total running time over~$x$ updates and~$y$ queries is~${\Oh((x + m_0) u(n, m_0) + (y+m_0) q(n, m_0))}$, where~$m_0$ is the maximum number of edges in the graph at any time.\footnotemark{}
\footnotetext{Our definition of amortization is more lenient than some in the literature, which yields better lower bounds. See~\cite[Section~1.2]{henzinger_unifying_2015} for more details.}

\subparagraph*{Approximation variants.}
The goal of~\CoreValue is to answer queries on the core value of any vertex. If an efficient algorithm for that is hard to find or unlikely to exist, it is natural to consider algorithms that give approximations to the core values. 

Given $0 < \alpha_1 \leq 1$ and~$\alpha_2 \geq 1$, a value~$s$ is an~\deff{$(\alpha_1, \alpha_2)$-approximation} for some non-negative value~$c$ if~$\alpha_1 c \leq s \leq \alpha_2 c$. We also say~$s$ is an~$\frac{\alpha_2}{\alpha_1}$-approximation for~$c$.

\begin{lemma}[Gap reduction] \label{res:gap_reduction}
    If an unknown value~$c$ satisfies~$0 \leq c \leq X$ or $c \geq Y$, then we can use any~$\alpha$-approximation for~$c$, with~$\alpha < \frac{Y}{X}$, to determine whether $c \leq X$ or~$c \geq Y$.
\end{lemma}

\section{Reductions from circuit value to core decomposition} \label{sec:mcv_to_3core}

One of the most basic and complete problems in \P is the circuit value problem, which is as hard as its bounded monotone version~\cite{greenlaw_limits_1995}. 
A \emph{bounded monotone Boolean circuit} is a directed acyclic graph (DAG) with a label on each vertex taken from the set~$\{0,1,\BAND,\BOR\}$ and a specified output vertex. We refer to vertices and arcs of the DAG as \deff{gates} and \deff{wires}, respectively. The output gate has out-degree~0, 0- and 1-gates have in-degree~0 and out-degree at most~1, and all other gates have in- and out-degree at most~2. See \Cref{fig:example_circuit} for an example of a monotone circuit.

\begin{figure}[h]
    \centering
    \documentclass[tikz, border=5mm]{standalone}
\usepackage[american]{circuitikz}
\usetikzlibrary{positioning}

\begin{document}
\begin{tikzpicture}[
    node distance=0.8cm and 1.8cm,
    gate/.style={draw, minimum size=0.8cm, no input leads, no output leads},
    inputnode/.style={draw, rectangle, minimum size=0.5cm},
    gateone/.style={very thick},
    gatezero/.style={},
    label/.style={font=\bfseries},
    zero_gate_color/.style={fill=gray!30},
    one_gate_color/.style={fill=green!30, line width=2pt},
    and_gate_color/.style={fill=orange!30},
    or_gate_color/.style={fill=cyan!30}
]

\node[inputnode, gatezero, label, zero_gate_color] (in1a) at (0, 1.38) {0};
\node[inputnode, gateone, label, one_gate_color] (in1b) at (0, 0.80) {1};

\node[inputnode, gatezero, label, zero_gate_color] (in2a) at (0, -0.42) {0};
\node[inputnode, gateone, label, one_gate_color] (in2b) at (0, -1) {1};

\node[inputnode, gateone, label, one_gate_color] (in3) at (6.5, -1.0) {1};

\node[american and port, gate, gatezero, and_gate_color] (and1) at (3, 1.1) {};
\node[american or port, gate, gateone, or_gate_color] (or1) at (3, -0.7) {};

\node[american or port, gate, gateone, or_gate_color] (or2) at (6, 0.2) {};

\node[american and port, gate, gateone, and_gate_color] (output) at (9, -0.08) {};

\draw (in1a.east) -- (and1.bin 1);
\draw (in1b.east) -- (and1.bin 2);
\draw (in2a.east) -- (or1.bin 1);
\draw (in2b.east) -- (or1.bin 2);

\draw (and1.bout) -- ++(0.5,0) |- (or2.bin 1);
\draw (or1.bout) -- ++(0.5,0) |- (or2.bin 2);

\draw (or2.bout) -- ++(0,0) |- (output.bin 1);
\draw (in3.east) -- ++(0.25,0) |- (output.bin 2);

\end{tikzpicture}
\end{document}
    \caption{A monotone circuit. \BAND-gates are in \textcolor{orange}{orange} and have a flat left side, while \BOR-gates are in \textcolor{blue}{blue} and have a curved left side. The gates in \textbf{bold} have value~1.}
    \Description{Example monotone circuit containing AND, OR, 0 and 1 gates, with wires connecting them.}
    \label{fig:example_circuit}
\end{figure}

The propagation of the values of the 0- and 1-gates through all the \BAND- and \BOR-gates results in a 0 or 1 value at each gate, and the value of the output gate is the \emph{circuit value}. \BAND- and \BOR-gates with in-degree less than 2 work as if their remaining inputs were~0, and gates with out-degree 2 output the same value on both wires. The \emph{size} of the circuit is its number of gates. 
The classic decision problem consists of determining whether the circuit value is 1. In the literature, it is referred to as the fanin 2, fanout 2 monotone circuit value problem~\cite{greenlaw_limits_1995}. Here we describe its dynamic version:

\begin{problemStatement}{Monotone circuit value (\MCV)}
  Given a bounded monotone Boolean circuit and a fixed output gate, process a series of wire insertions/deletions, and queries: Is the circuit value~1?
\end{problemStatement}

In this section, we show reductions from~\MCV to several variants of \kcore, which are used in the following sections to prove conditional and unconditional lower bounds for these variants via \MCV.

\subsection{Exact core decomposition}

We will now show how to solve \MCV using a simple decision variant of \CoreValue. Note all our dynamic reductions in this section preserve the type of time complexity. That is, if the original algorithm had worst-case (amortized) time complexity, the derived algorithm for~\MCV preserves worst-case (amortized, resp.) time complexity.

\begin{problemStatement}{\problem{$k$-Core}}
    Given a graph with a fixed vertex~$s^*$, process a series of edge insertions/ deletions, and queries: Is~$K_{s^*} \geq k$?
\end{problemStatement}

\begin{proposition} \label{result:mcv_to_3core}
  For any~$k \geq 3$, if fully (partially) dynamic \problem{$k$-Core} on a graph on $n$ vertices and $m$ edges can be solved in 
  $\dynTime{p(n, m)}{u(n, m)}{q(n, m)}$ time, then fully (partially) dynamic \MCV on a circuit of size~$N$ can be solved in $\dynOh{p(N', M')}{u(N',M')}{q(N',M')}$ time, with~$N'=\Oh(N+k)$ and~$M'=\Oh(Nk)$.
\end{proposition}
\begin{proof}
    Let us consider the case~$k = 3$. Given a circuit $C$ with output gate~$g^*$, we will describe a graph $G_C$ with a vertex $s^*$ such that $K_{s^*} \geq 3$ if and only if the circuit value of $C$ is~1.

    Each 0-gate in $C$ corresponds to a vertex in $G_C$, and each 1-gate in $C$ corresponds to a~$K_4$ in~$G_C$ 
    with a specified output vertex~$O$, as in Figures~\ref{fig:3core_0_cvp} and~\ref{fig:3core_1_cvp}. 
    Each \BOR- and \BAND-gate in $C$ corresponds to a graph as depicted in Figures~\ref{fig:3core_or_cvp} 
    and~\ref{fig:3core_and_cvp}, with their specified input and output vertices.

\begin{figure*}[h]
\centering
\captionsetup[subfigure]{justification=centering}
\begin{subfigure}{0.28\textwidth}
    \centering
    \begin{tikzpicture}
	\begin{pgfonlayer}{nodelayer}
		\node [style=node 0] (4) at (2, 0) {};
		\node [style=none] (6) at (2, 0.25) {$O$};
		\node [style=none] (9) at (3, 0) {};
		\node [style=none] (10) at (1, 0.5) {};
		\node [style=none] (11) at (1, -0.5) {};
	\end{pgfonlayer}
	\begin{pgfonlayer}{edgelayer}
		\draw [style=output edge] (4) to (9.center);
	\end{pgfonlayer}
\end{tikzpicture}
    \captionsetup{textformat=simple}
    \subcaption{$0$-gate}
    \Description{A single vertex with 1 output edge.}
    \label{fig:3core_0_cvp}
\end{subfigure}%
\begin{subfigure}{0.30\textwidth}
    \centering
    \begin{tikzpicture}
	\begin{pgfonlayer}{nodelayer}
		\node [style=node 0] (0) at (0, 0) {};
		\node [style=node 0] (1) at (1, 0.5) {};
		\node [style=node 0] (2) at (1, -0.5) {};
		\node [style=node 0] (4) at (2, 0) {};
		\node [style=none] (6) at (2, 0.25) {$O$};
		\node [style=none] (9) at (3, 0) {};
	\end{pgfonlayer}
	\begin{pgfonlayer}{edgelayer}
		\draw (0) to (1);
		\draw (0) to (2);
		\draw (1) to (4);
		\draw (2) to (4);
		\draw [style=output edge] (4) to (9.center);
		\draw (0) to (4);
		\draw (1) to (2);
	\end{pgfonlayer}
\end{tikzpicture}
    \captionsetup{textformat=simple}
    \subcaption{$1$-gate}
    \Description{A complete graph on four vertices with one output edge from one of the vertices.}
    \label{fig:3core_1_cvp}
\end{subfigure}%
\begin{subfigure}{0.40\textwidth}
    \centering
    \begin{tikzpicture}
	\begin{pgfonlayer}{nodelayer}
		\node [style=node 0] (0) at (0, 0) {};
		\node [style=node 0] (1) at (1, 0.5) {};
		\node [style=node 0] (2) at (1, -0.5) {};
		\node [style=node 0] (3) at (1, 0) {};
		\node [style=node 0] (4) at (2, 0) {};
		\node [style=node 0] (5) at (1, 0) {};
		\node [style=none] (6) at (2, 0.25) {$O$};
		\node [style=none] (7) at (-1, 0.25) {$I$};
		\node [style=hollow] (8) at (-1, 0) {};
		\node [style=none] (9) at (3, 0) {};
	\end{pgfonlayer}
	\begin{pgfonlayer}{edgelayer}
		\draw (0) to (1);
		\draw (0) to (2);
		\draw (1) to (5);
		\draw (5) to (2);
		\draw (5) to (4);
		\draw (1) to (4);
		\draw (2) to (4);
		\draw [style=input edge] (0) to (8);
		\draw [style=output edge] (4) to (9.center);
	\end{pgfonlayer}
\end{tikzpicture}
    \captionsetup{textformat=simple}
    \subcaption{``Arrow''}
    \Description{Graph with vertices A through E. Edges AB, AD, BC, CD, BE, CE and DE. Input edge to A and output edge from E.}
    \label{fig:3core_arrow_cvp}
\end{subfigure}%
\\ \vspace{3mm}
\begin{subfigure}{0.40\textwidth}
    \centering
    \begin{tikzpicture}
	\begin{pgfonlayer}{nodelayer}
		\node [style=node 0] (0) at (0, 0) {};
		\node [style=node 0] (1) at (1, 0.75) {};
		\node [style=node 0] (2) at (2, 0.25) {};
		\node [style=node 0] (3) at (2, 1.25) {};
		\node [style=node 0] (4) at (3, 0.75) {};
		\node [style=node 0] (5) at (2, 0.75) {};
		\node [style=node 0] (6) at (1, -0.75) {};
		\node [style=node 0] (7) at (2, -1.25) {};
		\node [style=node 0] (8) at (2, -0.25) {};
		\node [style=node 0] (9) at (3, -0.75) {};
		\node [style=node 0] (10) at (2, -0.75) {};
		\node [style=none] (11) at (3, 1) {$O_1$};
		\node [style=none] (12) at (3, -0.5) {$O_2$};
		\node [style=none] (13) at (-0.75, 1) {$I_1$};
		\node [style=hollow] (14) at (-0.75, 0.75) {};
		\node [style=hollow] (15) at (-0.75, -0.75) {};
		\node [style=none] (16) at (3.75, 0.75) {};
		\node [style=none] (19) at (3.75, -0.75) {};
		\node [style=none] (20) at (-0.75, -0.45) {$I_2$};
	\end{pgfonlayer}
	\begin{pgfonlayer}{edgelayer}
		\draw (0) to (1);
		\draw (1) to (3);
		\draw (1) to (2);
		\draw (5) to (3);
		\draw (2) to (5);
		\draw (3) to (4);
		\draw (4) to (2);
		\draw (5) to (4);
		\draw (6) to (8);
		\draw (6) to (7);
		\draw (10) to (8);
		\draw (7) to (10);
		\draw (8) to (9);
		\draw (9) to (7);
		\draw (10) to (9);
		\draw (0) to (6);
		\draw [style=input edge] (0) to (14);
		\draw [style=input edge] (0) to (15);
		\draw [style=output edge] (9) to (19.center);
		\draw [style=output edge] (4) to (16.center);
	\end{pgfonlayer}
\end{tikzpicture}
    \captionsetup{textformat=simple}
    \subcaption{\BOR-gate}
    \Description{Vertex with two input edges, connected to the vertices A' and A'' of two arrow graphs.}
    \label{fig:3core_or_cvp}
\end{subfigure}%
\begin{subfigure}{0.58\textwidth}
    \centering
    \begin{tikzpicture}
	\begin{pgfonlayer}{nodelayer}
		\node [style=node 0] (0) at (1, 0.75) {};
		\node [style=node 0] (1) at (2, 0.25) {};
		\node [style=node 0] (2) at (2, 0.75) {};
		\node [style=node 0] (3) at (2, 1.25) {};
		\node [style=node 0] (4) at (3, 0.75) {};
		\node [style=node 0] (5) at (2, 0.25) {};
		\node [style=node 0] (6) at (1, -0.75) {};
		\node [style=node 0] (7) at (2, -1.25) {};
		\node [style=node 0] (8) at (2, -0.75) {};
		\node [style=node 0] (9) at (2, -0.25) {};
		\node [style=node 0] (10) at (3, -0.75) {};
		\node [style=node 0] (11) at (2, -1.25) {};
		\node [style=node 0] (12) at (0, 0) {};
		\node [style=node 0] (13) at (-1, 0) {};
		\node [style=node 0] (14) at (-2, 0.75) {};
		\node [style=node 0] (15) at (-2, -0.75) {};
		\node [style=none] (16) at (-2.75, 1) {$I_1$};
		\node [style=none] (17) at (-2.75, -0.5) {$I_2$};
		\node [style=none] (18) at (3, 1) {$O_1$};
		\node [style=none] (19) at (3, -0.5) {$O_2$};
		\node [style=hollow] (20) at (-2.75, 0.75) {};
		\node [style=hollow] (21) at (-2.75, -0.75) {};
		\node [style=none] (22) at (3.75, 0.75) {};
		\node [style=none] (23) at (3.75, -0.75) {};
	\end{pgfonlayer}
	\begin{pgfonlayer}{edgelayer}
		\draw (0) to (1);
		\draw (2) to (1);
		\draw (2) to (3);
		\draw (0) to (3);
		\draw (3) to (4);
		\draw (2) to (4);
		\draw (5) to (4);
		\draw (6) to (7);
		\draw (8) to (7);
		\draw (8) to (9);
		\draw (6) to (9);
		\draw (9) to (10);
		\draw (8) to (10);
		\draw (11) to (10);
		\draw (14) to (15);
		\draw (15) to (13);
		\draw (14) to (13);
		\draw (13) to (12);
		\draw (12) to (0);
		\draw (12) to (6);
		\draw [style=input edge] (15) to (21);
		\draw [style=input edge] (14) to (20);
		\draw [style=output edge] (10) to (23.center);
		\draw [style=output edge] (4) to (22.center);
	\end{pgfonlayer}
\end{tikzpicture}
    \captionsetup{textformat=simple}
    \subcaption{\BAND-gate}
    \Description{Triangle ABC, with A and B having an input edge each, C connected to a vertex D which is connected to the vertices A' and A'' of two arrow graphs.}
    \label{fig:3core_and_cvp}
\end{subfigure}%
\caption{Building blocks of the reduction from \MCV to \threeCore. Input edges are \textcolor{red}{red} and output edges are \textcolor{blue}{blue}. Input vertices either do not exist or are output vertices of other gates.}
\label{fig:3core_cvp}
\end{figure*}

    Additionally, for each wire from a gate~$g$ 
    to a gate~$h$ in~$C$, identify an output vertex of the graph for~$g$ and an input vertex of the graph for~$h$, and let this wire correspond
    to the edge between these two vertices.
    This completes the description of the graph~$G_C$, and an example is shown in~\Cref{fig:example_3core_mcvp}. 
    Let $s^*$ be an output vertex of the graph for~$g^*$.
    The insertion or deletion of wires in~$C$ will be simulated by the insertion or deletion of the corresponding edge of~$G_C$,
    preserving the relation between~$C$ and $G_C$.

\begin{figure} 
    \centering
    \documentclass[tikz, border=10mm]{standalone}
\usetikzlibrary{backgrounds, fit}

\pgfdeclarelayer{background}
\pgfdeclarelayer{edgelayer}
\pgfdeclarelayer{nodelayer}
\pgfsetlayers{background,edgelayer,nodelayer,main}

\begin{document}
\begin{tikzpicture}[
    node 0/.style={circle, draw=black, fill=black, inner sep=1.5pt},
    hollow/.style={circle, draw=black, fill=white, inner sep=1.5pt},
    none/.style={},
    input edge/.style={dashed, red, thick},
    output edge/.style={dashed, blue, thick},
    region/.style={rounded corners, fill opacity=0.15, draw, inner sep=2pt},
	one_region/.style={line width=1.5pt},
    zero_gate_region/.style={region, fill=gray},
    one_gate_region/.style={region, fill=green},
    and_gate_region/.style={region, fill=orange},
    or_gate_region/.style={region, fill=cyan},
    3core/.style={},
    yscale=0.7,
    xscale=0.6,
]
	\begin{pgfonlayer}{nodelayer}
      \begin{scope}[xshift=-0.5cm]
		\node [style=node 0] (0) at (1, 1.25) {};
		\node [style=node 0] (1) at (2, 0.75) {};
		\node [style=node 0] (2) at (2, 1.25) {};
		\node [style=node 0] (3) at (2, 1.75) {};
		\node [style=node 0] (4) at (3, 1.25) {};
		\node [style=node 0] (5) at (2, 0.75) {}; 
		\node [style=node 0] (6) at (1, -0.25) {};
		\node [style=node 0] (7) at (2, -0.75) {};
		\node [style=node 0] (8) at (2, -0.25) {};
		\node [style=node 0] (9) at (2, 0.25) {};
		\node [style=node 0] (10) at (3, -0.25) {};
		\node [style=node 0] (11) at (2, -0.75) {}; 
		\node [style=node 0] (12) at (0, 0.5) {};
		\node [style=node 0] (13) at (-1, 0.5) {};
		\node [style=node 0] (14) at (-2, 1.25) {};
		\node [style=node 0] (15) at (-2, -0.25) {};
      \end{scope}

      \begin{scope}[xshift=0.5cm]
		\node [style=node 0, 3core] (24) at (10.5, -2.25) {};
		\node [style=node 0, 3core] (25) at (11.5, -2.75) {};
		\node [style=node 0, 3core] (26) at (11.5, -2.25) {};
		\node [style=node 0, 3core] (27) at (11.5, -1.75) {};
		\node [style=node 0, 3core] (28) at (12.5, -2.25) {};
		\node [style=node 0, 3core] (29) at (11.5, -2.75) {}; 
		\node [style=node 0, 3core] (30) at (10.5, -3.75) {};
		\node [style=node 0, 3core] (31) at (11.5, -4.25) {};
		\node [style=node 0, 3core] (32) at (11.5, -3.75) {};
		\node [style=node 0, 3core] (33) at (11.5, -3.25) {};
		\node [style=node 0, 3core] (34) at (12.5, -3.75) {};
		\node [style=node 0, 3core] (35) at (11.5, -4.25) {}; 
		\node [style=node 0, 3core] (36) at (9.5, -3) {};
        
		\node [style=node 0, 3core] (37) at (8.5, -3) {};
		\node [style=node 0, 3core] (38) at (7.5, -2) {};
		\node [style=node 0, 3core] (39) at (7.5, -4) {};
      \end{scope}
      
		\node [style=hollow, 3core] (44) at (6.75, -2) {};
		\node [style=hollow, 3core] (45) at (6.75, -4) {};
		
      \begin{scope}[xshift=-0.5cm]
		\node [style=node 0, 3core] (48) at (0, -3) {};
		\node [style=node 0, 3core] (49) at (1, -2.25) {};
		\node [style=node 0, 3core] (50) at (2, -2.75) {};
		\node [style=node 0, 3core] (51) at (2, -1.75) {};
		\node [style=node 0, 3core] (52) at (3, -2.25) {};
		\node [style=node 0, 3core] (53) at (2, -2.25) {};
		\node [style=node 0, 3core] (54) at (1, -3.75) {};
		\node [style=node 0, 3core] (55) at (2, -4.25) {};
		\node [style=node 0, 3core] (56) at (2, -3.25) {};
		\node [style=node 0, 3core] (57) at (3, -3.75) {};
		\node [style=node 0, 3core] (58) at (2, -3.75) {};
      \end{scope}
		
		\node [style=node 0, 3core] (67) at (3.75, -1.25) {};
		\node [style=node 0, 3core] (68) at (4.75, -0.5) {};
		\node [style=node 0, 3core] (69) at (5.75, -1) {};
		\node [style=node 0, 3core] (70) at (5.75, 0) {};
		\node [style=node 0, 3core] (71) at (6.75, -0.5) {};
		\node [style=node 0, 3core] (72) at (5.75, -0.5) {};
		\node [style=node 0, 3core] (73) at (4.75, -2) {};
		\node [style=node 0, 3core] (74) at (5.75, -2.5) {};
		\node [style=node 0, 3core] (75) at (5.75, -1.5) {};
		\node [style=node 0, 3core] (76) at (6.75, -2) {};
		\node [style=node 0, 3core] (77) at (5.75, -2) {};

      \begin{scope}[xshift=-1cm]
		\node [style=node 0, 3core] (86) at (-4.75, -0.25) {};
		\node [style=node 0, 3core] (87) at (-3.75, 0.25) {};
		\node [style=node 0, 3core] (88) at (-3.75, -0.75) {};
		\node [style=node 0, 3core] (89) at (-2.75, -0.25) {};
		
		\node [style=node 0, 3core] (92) at (-2.75, -3.75) {};
		\node [style=node 0, 3core] (93) at (-1.75, -3.25) {};
		\node [style=node 0, 3core] (94) at (-1.75, -4.25) {};
		\node [style=node 0, 3core] (95) at (-0.75, -3.75) {};
		
		\node [style=node 0] (104) at (-2.75, 1.25) {};
		
		\node [style=node 0] (107) at (-0.75, -2.25) {};
      \end{scope}
		
		\node [style=node 0, 3core] (98) at (4.75, -4) {};
		\node [style=node 0, 3core] (99) at (5.75, -3.5) {};
		\node [style=node 0, 3core] (100) at (5.75, -4.5) {};
		\node [style=node 0, 3core] (101) at (6.75, -4) {};
		
		\node [style=none] (108) at (13.5, -0.75) {Output};
	\end{pgfonlayer}
	
	\begin{pgfonlayer}{edgelayer}
		\draw (0) to (1); \draw (2) to (1); \draw (2) to (3); \draw (0) to (3);
		\draw (3) to (4); \draw (2) to (4); \draw (5) to (4); \draw (6) to (7);
		\draw (8) to (7); \draw (8) to (9); \draw (6) to (9); \draw (9) to (10);
		\draw (8) to (10); \draw (11) to (10); \draw (14) to (15); \draw (15) to (13);
		\draw (14) to (13); \draw (13) to (12); \draw (12) to (0); \draw (12) to (6);
		\draw (24) to (25); \draw (26) to (25); \draw (26) to (27); \draw (24) to (27);
		\draw (27) to (28); \draw (26) to (28); \draw (29) to (28); \draw (30) to (31);
		\draw (32) to (31); \draw (32) to (33); \draw (30) to (33); \draw (33) to (34);
		\draw (32) to (34); \draw (35) to (34); \draw (38) to (39); \draw (39) to (37);
		\draw (38) to (37); \draw (37) to (36); \draw (36) to (24); \draw (36) to (30);
		\draw [style=input edge] (39) to (45); \draw [style=input edge] (38) to (44);
		\draw (48) to (49); \draw (49) to (51); \draw (49) to (50); \draw (53) to (51);
		\draw (50) to (53); \draw (51) to (52); \draw (52) to (50); \draw (53) to (52);
		\draw (54) to (56); \draw (54) to (55); \draw (58) to (56); \draw (55) to (58);
		\draw (56) to (57); \draw (57) to (55); \draw (58) to (57); \draw (48) to (54);
		\draw (67) to (68); \draw (68) to (70); \draw (68) to (69); \draw (72) to (70);
		\draw (69) to (72); \draw (70) to (71); \draw (71) to (69); \draw (72) to (71);
		\draw (73) to (75); \draw (73) to (74); \draw (77) to (75); \draw (74) to (77);
		\draw (75) to (76); \draw (76) to (74); \draw (77) to (76); \draw (67) to (73);
		\draw (86) to (87); \draw (86) to (88); \draw (87) to (89); \draw (88) to (89);
		\draw (86) to (89); \draw (87) to (88); \draw (92) to (93); \draw (92) to (94);
		\draw (93) to (95); \draw (94) to (95); \draw (92) to (95); \draw (93) to (94);
		\draw (98) to (99); \draw (98) to (100); \draw (99) to (101); \draw (100) to (101);
		\draw (98) to (101); \draw (99) to (100);
		\draw [style=input edge] (15) to (89); \draw [style=input edge] (14) to (104);
		\draw [style=input edge] (67) to (10); \draw [style=input edge] (67) to (52);
		\draw [style=input edge] (48) to (107); \draw [style=input edge] (48) to (95);
		\draw [style=<-, shorten <= 2pt, shorten >= -5pt] (28) to (108);
	\end{pgfonlayer}

	\begin{pgfonlayer}{background}
        \node[zero_gate_region, fit=(104)] {};
        \node[zero_gate_region, fit=(107)] {};

        \node[one_gate_region, one_region, fit=(86)(87)(88)(89)] {};
        \node[one_gate_region, one_region, fit=(92)(93)(94)(95)] {};
        \node[one_gate_region, one_region, fit=(98)(99)(100)(101)] {};

        \node[and_gate_region, fit=(0)(1)(2)(3)(4)(6)(7)(8)(9)(10)(12)(13)(14)(15)] {};

        \node[or_gate_region, one_region, fit=(48)(49)(50)(51)(52)(53)(54)(55)(56)(57)(58)] {};

        \node[or_gate_region, one_region, fit=(67)(68)(69)(70)(71)(72)(73)(74)(75)(76)(77)] {};

        \node[and_gate_region, one_region, fit=(38)(39)(37)(24)(25)(26)(27)(28)(30)(31)(32)(33)(34)(36)] {};
    \end{pgfonlayer}
\end{tikzpicture}
\end{document}
    \caption{An example of the reduction described in~\Cref{result:mcv_to_3core} using the example circuit from~\Cref{fig:example_circuit}.}
    \Description{Example of the reduction from monotone circuit value to 3-core using the example circuit.}
    \label{fig:example_3core_mcvp}
\end{figure}

    Let us argue that, at any time, $K_{s^*} \geq 3$ if and only if the circuit value is~1.
    First note that all vertices in the graph for a 1-gate are in the 3-core of $G_C$,
    and no vertex corresponding to a 0-gate is in the 3-core, because each such vertex has maximum degree~1.
    Now note that the arrow graph shown in~\Cref{fig:3core_arrow_cvp}, which is used in both the \BAND- and \BOR-gates,
    has the property that its output vertex is in the 3-core of~$G_C$ if and only if its input vertex is in the 3-core of~$G_C$.
    This implies that the output vertices of the graph for an \BOR-gate are in the 3-core of~$G_C$ if and only if at least one of its two input 
    vertices is in the 3-core of~$G_C$. Similarly, the output vertices of the graph for an \BAND-gate are in the 3-core 
    of~$G_C$ if and only if both its input vertices are in the 3-core of~$G_C$. 
    The way the graphs for the gates of~$C$ are connected in~$G_C$ implies that, in the same way the 0- and 1-gate values propagate towards the output gate, the fact that all vertices from the 1-gate graphs are in the 3-core of~$G_C$ propagates the 3-core, reaching vertex $s^*$ if and only if the circuit value is~1.
    When~$k > 3$, the graph~$G_C$ contains~$k-3$ additional vertices connected to each other and to each vertex in~$G_C$. Note that~$K_{s^*} \geq k$ if and only if~$K_{s^*} \geq 3$ in the original graph.
        
    If $C$ has $N$ gates, then~$C$ has~$\Oh(N)$ wires since it is a bounded circuit, and~$G_C$ has~${N'=\Oh(N+k)}$ vertices and~${M'=\Oh(Nk)}$ edges. 
    The preprocessing of $C$ will be the construction of~$G_C$ plus its preprocessing, which will take~$p(N',M')$ time.
    Each wire insertion/deletion in~$C$ corresponds to a single edge insertion/deletion in~$G_C$, 
    thus each update will take~$u(N',M')$ time.
    Finally, each query in $C$ will be answered by a single query in $G_C$, which will take $q(N',M')$ time.
\end{proof}

Since our proof requires~$k \geq 3$, one might wonder whether a similar result holds for \problem{2-Core}.
In~\Cref{sec:2core_alg},
we show that this is not the case, since we present a~polylog time algorithm for~\problem{$2$-Core}, and in~\Cref{sec:all_lower} we prove that a polylog algorithm would be unlikely to exist if this result held for~\problem{$2$-Core}.

\subsection{Approximate core decomposition}
\Cref{result:mcv_to_3core,res:gap_reduction} imply
that finding any $(\frac32-\epsilon)$-approximation is as hard as~\MCV. In fact, using a more intricate reduction, we can solve~\MCV using even a~$(2-\epsilon)$-approximation. To show this, let us first define a simplified approximation version of~\CoreValue.

\begin{problemStatement}{\ApproxCore{$\alpha$}}
    Given a graph with a fixed vertex~$s^*$, process a series of edge insertions/deletions, and queries: Compute an~$\alpha$-approximation for~$K_{s^*}$.
\end{problemStatement}

\begin{proposition} \label{result:mcv_to_2-e_kcore_approx}
    For any~$\epsilon > 0$, if fully (partially) dynamic \ApproxCore{$(2-\epsilon)$} on a graph with~$n$ vertices and~$m$ edges can be solved in
      $\dynTime{p(n, m)}{u(n, m)}{q(n, m)}$ time, then fully (partially, resp.) dynamic \MCV on a circuit of size~$N$ can be solved in $\dynOh{p(N', M')}{\epsilon^{-2} u(N', M')}{q(N', M')}$ time, with~$N'=\Oh(N\epsilon^{-2})$ and~$M'=\Oh(N \epsilon^{-3})$.
\end{proposition}

\begin{proof}
    Given a circuit~$C$ with output gate~$g^*$, for any~$k \geq 2$, we will describe a graph~$G_C$ with a vertex~$s^*$ such that~$K_{s^*} \geq 2k$ if the circuit value of~$C$ is~1, otherwise~$K_{s^*} \leq k+1$. Then, by Lemma~\ref{res:gap_reduction}, any~$\alpha$-approximation for~$K_{s^*}$ for $\alpha < \frac{2k}{k+1}$ can be used to determine the circuit value of~$C$. For any~$\epsilon > 0$, we can then choose a sufficiently large~$k$, in fact any~$k > \frac{2-\epsilon}{\epsilon} = \Oh(\epsilon^{-1})$, so that~$\alpha=2-\epsilon < \frac{2k}{k+1}$, and solve \MCV using an algorithm for the~\ApproxCore{$(2-\epsilon)$}.

    Each~$0$-gate in~$C$ corresponds to~$k$ vertices in~$G_C$ with~$k$ output edges each, and each~$1$-gate corresponds to a~$K_{2k+1}$, also with~$k^2$ output edges, as in~\Cref{fig:2eps_kcore_0_cvp,fig:2eps_kcore_1_cvp}. Each~\BAND-gate corresponds to a graph like~\Cref{fig:2eps_kcore_and_cvp}, with inputs~$I_1$ and~$I_2$, and outputs~$O_1$ and~$O_2$, and it may be generalized as a sequence of~$2k$ groups of~$k$ vertices, where vertices from adjacent groups are connected by an edge. It has~$2k^2$ output edges, one from each of its $2k^2$ vertices. Each~\BOR-gate corresponds to a graph like~\Cref{fig:2eps_kcore_or_cvp}, and may similarly be generalized as a sequence of~$2k$ groups of~$k$ vertices with~$2k^2$ output edges, plus a~$K_{2k+1}$ connected to each vertex in the last group with~$k$ edges.
    
\begin{figure*}[h]
\centering
\captionsetup[subfigure]{justification=centering}
\begin{subfigure}{0.28\textwidth}
    \centering
    \begin{tikzpicture}
	\begin{pgfonlayer}{nodelayer}
		\node [style=none] (0) at (0, 1) {};
		\node [style=none] (1) at (0, -1) {};
		\node [style=node 0] (12) at (1, -0.75) {};
		\node [style=node 0] (13) at (1, 0.75) {};
		\node [style=node 0] (14) at (1, 0) {};
		\node [style=none] (15) at (1.25, 1.25) {};
		\node [style=none] (16) at (1.25, -1.25) {};
		\node [style=none] (17) at (1.3, 0) {$O_1$};
		\node [style=node 0] (18) at (0, 0.5) {};
		\node [style=node 0] (19) at (0, 0) {};
		\node [style=node 0] (20) at (0, -0.5) {};
	\end{pgfonlayer}
	\begin{pgfonlayer}{edgelayer}
		\draw [bend right=90, looseness=1.25] (1.center) to (0.center);
		\draw [bend right=90, looseness=1.25] (0.center) to (1.center);
		\draw [bend right=90, looseness=0.40] (16.center) to (15.center);
		\draw [bend right=90, looseness=0.40] (15.center) to (16.center);
		\draw [style=output edge] (18) to (13);
		\draw [style=output edge] (18) to (14);
		\draw [style=output edge] (18) to (12);
		\draw [style=output edge] (19) to (13);
		\draw [style=output edge] (19) to (14);
		\draw [style=output edge] (19) to (12);
		\draw [style=output edge] (20) to (13);
		\draw [style=output edge] (20) to (14);
		\draw [style=output edge] (20) to (12);
	\end{pgfonlayer}
\end{tikzpicture}
    \captionsetup{textformat=simple}
    \subcaption{$0$-gate}
    \Description{Zero gate gadget for the approximate core reduction. Described in proof.}
    \label{fig:2eps_kcore_0_cvp}
\end{subfigure}%
\begin{subfigure}{0.7\textwidth}
    \centering
    \begin{tikzpicture}
	\begin{pgfonlayer}{nodelayer}
		\node [style=node 0] (0) at (-2, .75) {};
		\node [style=node 0] (1) at (-2, 0) {};
		\node [style=node 0] (2) at (-2, -.75) {};
		\node [style=node 0] (40) at (-2, 1.5) {};
		\node [style=node 0] (41) at (0, .75) {};
		\node [style=node 0] (42) at (0, 0) {};
		\node [style=node 0] (43) at (0, -.75) {};
		\node [style=node 0] (44) at (-1, .75) {};
		\node [style=node 0] (45) at (-1, 0) {};
		\node [style=node 0] (46) at (-1, -.75) {};
		\node [style=node 0] (47) at (-0.5, 1.5) {};
		\node [style=node 0] (51) at (1, .75) {};
		\node [style=node 0] (52) at (1, 0) {};
		\node [style=node 0] (53) at (1, -.75) {};
		\node [style=node 0] (55) at (-1.25, 1.5) {};
		\node [style=none] (58) at (-4.25, 1) {};
		\node [style=none] (59) at (-4.25, -1) {};
		\node [style=none] (60) at (-4.25, 0) {$I_1$};
		\node [style=none] (61) at (-4.25, -0.75) {};
		\node [style=none] (63) at (-4.25, -0.25) {};
		\node [style=none] (64) at (-4.25, -0.5) {};
		\node [style=none] (65) at (-4, -0.25) {};
		\node [style=none] (66) at (-4, 0) {};
		\node [style=none] (67) at (-4, 0.25) {};
		\node [style=none] (68) at (-4.25, 0.25) {};
		\node [style=none] (69) at (-4.25, 0.5) {};
		\node [style=none] (70) at (-4.25, 0.75) {};
		\node [style=node 0] (86) at (-3, -.75) {};
		\node [style=node 0] (87) at (-3, 0) {};
		\node [style=node 0] (88) at (-3, .75) {};
		\node [style=none] (92) at (3.25, 1) {};
		\node [style=none] (93) at (3.25, -1) {};
		\node [style=none] (94) at (3.25, 0) {$I_2$};
		\node [style=none] (95) at (3.25, -0.75) {};
		\node [style=none] (96) at (3.25, -0.25) {};
		\node [style=none] (97) at (3.25, -0.5) {};
		\node [style=none] (98) at (3, -0.25) {};
		\node [style=none] (99) at (3, 0) {};
		\node [style=none] (100) at (3, 0.25) {};
		\node [style=none] (101) at (3.25, 0.25) {};
		\node [style=none] (102) at (3.25, 0.5) {};
		\node [style=none] (103) at (3.25, 0.75) {};
		\node [style=node 0] (104) at (2, -.75) {};
		\node [style=node 0] (105) at (2, 0) {};
		\node [style=node 0] (106) at (2, .75) {};
		\node [style=none] (107) at (0.25, 1.75) {};
		\node [style=none] (108) at (-2.75, 1.75) {};
		\node [style=none] (109) at (-1.25, 1.75) {$O_1$};
		\node [style=node 0] (110) at (1, 1.5) {};
		\node [style=node 0] (111) at (2.5, 1.5) {};
		\node [style=node 0] (112) at (1.75, 1.5) {};
		\node [style=none] (113) at (3.25, 1.75) {};
		\node [style=none] (114) at (0.25, 1.75) {};
		\node [style=none] (115) at (1.75, 1.75) {$O_2$};
	\end{pgfonlayer}
	\begin{pgfonlayer}{edgelayer}
		\draw (44) to (41);
		\draw (44) to (42);
		\draw (44) to (43);
		\draw (45) to (41);
		\draw (45) to (42);
		\draw (45) to (43);
		\draw (46) to (41);
		\draw (46) to (43);
		\draw [style=output edge] (44) to (47);
		\draw [style=output edge] (45) to (47);
		\draw [style=output edge] (46) to (47);
		\draw (41) to (51);
		\draw (41) to (52);
		\draw (41) to (53);
		\draw (42) to (51);
		\draw (42) to (52);
		\draw (42) to (53);
		\draw (43) to (51);
		\draw (43) to (52);
		\draw (43) to (53);
		\draw (0) to (44);
		\draw (0) to (45);
		\draw (0) to (46);
		\draw (1) to (44);
		\draw (1) to (45);
		\draw (2) to (45);
		\draw (1) to (46);
		\draw (2) to (44);
		\draw (2) to (46);
		\draw [style=output edge] (0) to (55);
		\draw [style=output edge] (1) to (55);
		\draw [style=output edge] (2) to (55);
		\draw (59.center)
			 to [bend left=90, looseness=0.75] (58.center)
			 to [bend left=90, looseness=0.75] cycle;
		\draw [style=input edge] (61.center) to (86);
		\draw [style=input edge] (64.center) to (86);
		\draw [style=input edge] (63.center) to (86);
		\draw [style=input edge] (65.center) to (87);
		\draw [style=input edge] (66.center) to (87);
		\draw [style=input edge] (67.center) to (87);
		\draw [style=input edge] (68.center) to (88);
		\draw [style=input edge] (69.center) to (88);
		\draw [style=input edge] (70.center) to (88);
		\draw [bend right=90, looseness=0.75] (93.center) to (92.center);
		\draw [bend right=90, looseness=0.75] (92.center) to (93.center);
		\draw [style=input edge] (95.center) to (104);
		\draw [style=input edge] (97.center) to (104);
		\draw [style=input edge] (96.center) to (104);
		\draw [style=input edge] (98.center) to (105);
		\draw [style=input edge] (99.center) to (105);
		\draw [style=input edge] (100.center) to (105);
		\draw [style=input edge] (101.center) to (106);
		\draw [style=input edge] (102.center) to (106);
		\draw [style=input edge] (103.center) to (106);
		\draw (51) to (106);
		\draw (51) to (105);
		\draw (51) to (104);
		\draw (52) to (106);
		\draw (52) to (105);
		\draw (52) to (104);
		\draw (53) to (106);
		\draw (53) to (105);
		\draw (53) to (104);
		\draw (88) to (0);
		\draw (88) to (1);
		\draw (88) to (2);
		\draw (87) to (0);
		\draw (87) to (1);
		\draw (87) to (2);
		\draw (86) to (0);
		\draw (86) to (2);
		\draw [style=output edge] (88) to (40);
		\draw [style=output edge] (87) to (40);
		\draw [style=output edge] (86) to (40);
		\draw (108.center)
			 to [bend left=90, looseness=0.38] (107.center)
			 to [bend left=90, looseness=0.40] cycle;
		\draw (114.center)
			 to [bend left=90, looseness=0.38] (113.center)
			 to [bend left=90, looseness=0.40] cycle;
		\draw [style=output edge] (51) to (112);
		\draw [style=output edge] (52) to (112);
		\draw [style=output edge] (53) to (112);
		\draw [style=output edge] (41) to (110);
		\draw [style=output edge] (42) to (110);
		\draw [style=output edge] (43) to (110);
		\draw [style=output edge] (106) to (111);
		\draw [style=output edge] (105) to (111);
		\draw [style=output edge] (104) to (111);
		\draw (46) to (42);
		\draw (86) to (1);
	\end{pgfonlayer}
\end{tikzpicture}
    \captionsetup{textformat=simple}
    \subcaption{\BAND-gate}
    \Description{And gate gadget for the approximate core reduction. Described in proof.}
    \label{fig:2eps_kcore_and_cvp}
\end{subfigure}%
\\ \vspace{2mm}
\begin{subfigure}{0.28\textwidth}
    \centering
    \begin{tikzpicture}
	\begin{pgfonlayer}{nodelayer}
		\node [style=none] (3) at (0, 1) {};
		\node [style=none] (4) at (0, -1) {};
		\node [style=none] (5) at (0, 1.25) {$K_7$};
		\node [style=node 0] (12) at (0.25, 0.5) {};
		\node [style=node 0] (13) at (0.25, 0) {};
		\node [style=node 0] (14) at (0.25, -0.5) {};
		\node [style=node 0] (15) at (1, -0.75) {};
		\node [style=node 0] (16) at (1, 0.75) {};
		\node [style=node 0] (17) at (1, 0) {};
		\node [style=none] (18) at (1.25, 1.25) {};
		\node [style=none] (19) at (1.25, -1.25) {};
		\node [style=none] (20) at (1.3, 0) {$O_1$};
		\node [style=node 0] (21) at (-0.25, 0.75) {};
		\node [style=node 0] (22) at (-0.25, 0.25) {};
		\node [style=node 0] (23) at (-0.25, -0.25) {};
		\node [style=node 0] (24) at (-0.25, -0.75) {};
	\end{pgfonlayer}
	\begin{pgfonlayer}{edgelayer}
		\draw [bend right=90, looseness=1.25] (4.center) to (3.center);
		\draw [bend right=90, looseness=1.25] (3.center) to (4.center);
		\draw [bend right=90, looseness=0.40] (19.center) to (18.center);
		\draw [bend right=90, looseness=0.40] (18.center) to (19.center);
		\draw [style=output edge] (15) to (13);
		\draw [style=output edge] (15) to (12);
		\draw [style=output edge] (15) to (14);
		\draw [style=output edge] (16) to (12);
		\draw [style=output edge] (16) to (14);
		\draw [style=output edge] (16) to (13);
		\draw [style=output edge] (17) to (12);
		\draw [style=output edge] (17) to (13);
		\draw [style=output edge] (17) to (14);
	\end{pgfonlayer}
\end{tikzpicture}
    \captionsetup{textformat=simple}
    \subcaption{$1$-gate}
    \Description{One gate gadget for the approximate core reduction. Described in proof.}
    \label{fig:2eps_kcore_1_cvp}
\end{subfigure}%
\begin{subfigure}{0.7\textwidth}
    \centering
    \begin{tikzpicture}
	\begin{pgfonlayer}{nodelayer}
		\node [style=node 0] (6) at (1, 1.5) {};
		\node [style=node 0] (7) at (1, .75) {};
		\node [style=node 0] (8) at (1, 0) {};
		\node [style=node 0] (9) at (1, -.75) {};
		\node [style=node 0] (14) at (2, .75) {};
		\node [style=node 0] (15) at (2, 0) {};
		\node [style=node 0] (16) at (2, -.75) {};
		\node [style=node 0] (20) at (2.5, 1.5) {};
		\node [style=node 0] (21) at (3, .75) {};
		\node [style=node 0] (22) at (3, 0) {};
		\node [style=node 0] (23) at (3, -.75) {};
		\node [style=node 0] (28) at (4, .75) {};
		\node [style=node 0] (29) at (4, 0) {};
		\node [style=node 0] (30) at (4, -.75) {};
		\node [style=none] (48) at (-1.25, 1.5) {};
		\node [style=none] (49) at (-1.25, 0) {};
		\node [style=none] (50) at (-1.3, 1.2) {$I_1$};
		\node [style=none] (53) at (-1.5, 0.25) {};
		\node [style=none] (57) at (-1.25, 0) {};
		\node [style=none] (63) at (-1.5, 0.75) {};
		\node [style=none] (64) at (-1.5, 0.5) {};
		\node [style=none] (65) at (-1.25, 0.5) {};
		\node [style=none] (66) at (-1.25, 0.75) {};
		\node [style=none] (67) at (-1.25, 1) {};
		\node [style=none] (68) at (-1, 0.75) {};
		\node [style=none] (69) at (-1, 1) {};
		\node [style=none] (70) at (-1, 1.25) {};
		\node [style=none] (71) at (-1.25, 0) {};
		\node [style=none] (78) at (-1.25, -1.5) {};
		\node [style=none] (79) at (-1.25, 0) {};
		\node [style=none] (80) at (-1.3, -1.2) {$I_2$};
		\node [style=none] (81) at (-1.5, -0.25) {};
		\node [style=none] (82) at (-1.25, 0) {};
		\node [style=none] (83) at (-1.5, -0.75) {};
		\node [style=none] (84) at (-1.5, -0.5) {};
		\node [style=none] (85) at (-1.25, -0.5) {};
		\node [style=none] (86) at (-1.25, -0.75) {};
		\node [style=none] (87) at (-1.25, -1) {};
		\node [style=none] (88) at (-1, -0.75) {};
		\node [style=none] (89) at (-1, -1) {};
		\node [style=none] (90) at (-1, -1.25) {};
		\node [style=none] (91) at (-1.25, 0) {};
		\node [style=node 0] (92) at (0, -.75) {};
		\node [style=node 0] (93) at (0, 0) {};
		\node [style=node 0] (94) at (0, .75) {};
		\node [style=node 0] (113) at (1.75, 1.5) {};
		\node [style=none] (114) at (3.25, 1.75) {};
		\node [style=none] (115) at (0.25, 1.75) {};
		\node [style=none] (116) at (1.75, 1.75) {$O_1$};
		\node [style=node 0] (117) at (4, 1.5) {};
		\node [style=node 0] (118) at (5.5, 1.5) {};
		\node [style=node 0] (119) at (4.75, 1.5) {};
		\node [style=none] (120) at (6.25, 1.75) {};
		\node [style=none] (121) at (3.25, 1.75) {};
		\node [style=none] (122) at (4.75, 1.75) {$O_2$};
		\node [style=none] (123) at (6, 1) {};
		\node [style=none] (124) at (6, -1) {};
		\node [style=none] (125) at (6, 1.25) {$K_7$};
		\node [style=node 0] (126) at (5.75, 0.5) {};
		\node [style=node 0] (127) at (5.75, 0) {};
		\node [style=node 0] (128) at (5.75, -0.5) {};
		\node [style=node 0] (129) at (5, -.75) {};
		\node [style=node 0] (130) at (5, .75) {};
		\node [style=node 0] (131) at (5, 0) {};
		\node [style=node 0] (132) at (6.25, 0.75) {};
		\node [style=node 0] (133) at (6.25, 0.25) {};
		\node [style=node 0] (134) at (6.25, -0.25) {};
		\node [style=node 0] (135) at (6.25, -0.75) {};
	\end{pgfonlayer}
	\begin{pgfonlayer}{edgelayer}
		\draw [style=output edge] (14) to (20);
		\draw [style=output edge] (15) to (20);
		\draw [style=output edge] (16) to (20);
		\draw (7) to (14);
		\draw (7) to (15);
		\draw (7) to (16);
		\draw (8) to (14);
		\draw (8) to (15);
		\draw (8) to (16);
		\draw (9) to (14);
		\draw (9) to (15);
		\draw (9) to (16);
		\draw (14) to (21);
		\draw (14) to (22);
		\draw (14) to (23);
		\draw (15) to (21);
		\draw (15) to (22);
		\draw (15) to (23);
		\draw (16) to (21);
		\draw (16) to (22);
		\draw (16) to (23);
		\draw (21) to (28);
		\draw (21) to (29);
		\draw (21) to (30);
		\draw (22) to (28);
		\draw (22) to (29);
		\draw (22) to (30);
		\draw (23) to (28);
		\draw (23) to (29);
		\draw (23) to (30);
		\draw (49.center)
			 to [bend left=90] (48.center)
			 to [bend left=90] cycle;
		\draw (79.center)
			 to [bend right=90] (78.center)
			 to [bend right=90] cycle;
		\draw [style=input edge] (88.center) to (92);
		\draw [style=input edge] (89.center) to (92);
		\draw [style=input edge] (90.center) to (92);
		\draw [style=input edge] (85.center) to (93);
		\draw [style=input edge] (86.center) to (93);
		\draw [style=input edge] (87.center) to (93);
		\draw [style=input edge] (81.center) to (94);
		\draw [style=input edge] (84.center) to (94);
		\draw [style=input edge] (83.center) to (94);
		\draw [style=output edge] (92) to (6);
		\draw (92) to (7);
		\draw (92) to (8);
		\draw (92) to (9);
		\draw [style=input edge] (53.center) to (92);
		\draw [style=input edge] (64.center) to (92);
		\draw [style=input edge] (63.center) to (92);
		\draw [style=output edge] (93) to (6);
		\draw (93) to (7);
		\draw (93) to (8);
		\draw (93) to (9);
		\draw [style=input edge] (65.center) to (93);
		\draw [style=input edge] (66.center) to (93);
		\draw [style=input edge] (67.center) to (93);
		\draw [style=output edge] (94) to (6);
		\draw (94) to (7);
		\draw (94) to (8);
		\draw (94) to (9);
		\draw [style=input edge] (68.center) to (94);
		\draw [style=input edge] (69.center) to (94);
		\draw [style=input edge] (70.center) to (94);
		\draw (115.center)
			 to [bend left=90, looseness=0.38] (114.center)
			 to [bend left=90, looseness=0.40] cycle;
		\draw [bend left=90, looseness=0.38] (121.center) to (120.center);
		\draw [bend left=90, looseness=0.40] (120.center) to (121.center);
		\draw [style=output edge] (7) to (113);
		\draw [style=output edge] (8) to (113);
		\draw [style=output edge] (9) to (113);
		\draw [style=output edge] (21) to (117);
		\draw [style=output edge] (22) to (117);
		\draw [style=output edge] (23) to (117);
		\draw [style=output edge] (28) to (119);
		\draw [style=output edge] (29) to (119);
		\draw [style=output edge] (30) to (119);
		\draw [bend left=90, looseness=1.25] (124.center) to (123.center);
		\draw [bend left=90, looseness=1.25] (123.center) to (124.center);
		\draw (129) to (127);
		\draw (129) to (126);
		\draw (129) to (128);
		\draw (130) to (126);
		\draw (130) to (128);
		\draw (130) to (127);
		\draw (131) to (126);
		\draw (131) to (127);
		\draw (131) to (128);
		\draw (28) to (130);
		\draw (28) to (131);
		\draw (28) to (129);
		\draw (29) to (130);
		\draw (29) to (131);
		\draw (29) to (129);
		\draw (30) to (130);
		\draw (30) to (131);
		\draw (30) to (129);
		\draw [style=output edge] (130) to (118);
		\draw [style=output edge] (131) to (118);
		\draw [style=output edge] (129) to (118);
		\draw [style=output edge] (130) to (118);
		\draw [style=output edge] (131) to (118);
		\draw [style=output edge] (129) to (118);
	\end{pgfonlayer}
\end{tikzpicture}
    \captionsetup{textformat=simple}
    \subcaption{\BOR-gate}
    \Description{Or gate gadget for the approximate core reduction. Described in proof.}
    \label{fig:2eps_kcore_or_cvp}
\end{subfigure}%
\caption{Building blocks of the reduction from \MCV to \ApproxCore{$(2-\epsilon)$} for~$k=3$. Input edges are \textcolor{red}{red} and output edges are \textcolor{blue}{blue}.}
\label{fig:2eps_kcore_cvp}
\end{figure*}

    For each wire from a gate~$g$ to a gate~$h$ in~$C$, the~$k^2$ output edges from~$O_1$ or~$O_2$ of the graph for~$g$ are identified with the~$k^2$ input edges from either~$I_1$ or~$I_2$ of the graph for~$h$, that is, $O_1$ or $O_2$ from~$g$ will play the role of $I_1$ or $I_2$ for~$h$. This completes the description of~$G_C$, and an example is shown in~\Cref{fig:example_core_approx_mcvp}. Each vertex in the graph for a gate incident to some of its output edges is called an output vertex for the gate. Let~$s^*$ be an output vertex in the graph for~$g^*$. The insertion or deletion of wires in~$C$ will be simulated by the insertion or deletion of~$k^2$ edges in~$G_C$ (from the input to the output of the wire), preserving the relation between~$C$ and~$G_C$.

\begin{figure}[h] 
    \centering
    \documentclass[tikz, border=10mm]{standalone}
\usetikzlibrary{backgrounds, fit, graphs, graphs.standard}

\pgfdeclarelayer{background}
\pgfdeclarelayer{edgelayer}
\pgfdeclarelayer{nodelayer}
\pgfsetlayers{background,edgelayer,nodelayer,main}

\begin{document}
\begin{tikzpicture}[
    node 0/.style={circle, draw=black, fill=black, inner sep=1.5pt},
    hollow/.style={circle, draw=black, fill=white, inner sep=1.5pt},
    none/.style={},
    input edge/.style={dashed, red, thick},
    output edge/.style={dashed, blue, thick},
    region/.style={rounded corners, fill opacity=0.15, draw, inner sep=2pt},
	one_region/.style={line width=1.5pt},
    zero_gate_region/.style={region, fill=gray},
    one_gate_region/.style={region, fill=green},
    and_gate_region/.style={region, fill=orange},
    or_gate_region/.style={region, fill=cyan},
    3core/.style={},
    scale=0.7,
]
	\begin{pgfonlayer}{nodelayer}
		\node [style=node 0] (116) at (0, 1) {};
		\node [style=node 0] (117) at (0, 0.5) {};

        \begin{scope}[xshift=0cm,yshift=-2.5cm]
            \foreach \x in {0,1,2,3,4} {
        	   \node [style=node 0] (\x) at (-90+\x*72:.5) {};
            };
        \end{scope}

		\node [style=node 0] (123) at (0, -3.75) {};
		\node [style=node 0] (124) at (0, -4.25) {};
        
        \begin{scope}[xshift=0cm,yshift=-6cm]
            \foreach \l[count=\x] in {118,119,120,121,122} {
            	   \node [style=node 0] (\l) at (-162+\x*72:.5) {};
            };
        \end{scope}
        
		\node [style=node 0] (69) at (1.25, 0) {};
		\node [style=node 0] (70) at (1.25, -0.5) {};
		\node [style=node 0] (71) at (2.25, -0.5) {};
		\node [style=node 0] (72) at (2.25, 0) {};
		\node [style=node 0] (73) at (1.25, -1) {};
		\node [style=node 0] (74) at (1.25, -1.5) {};
		\node [style=node 0] (75) at (2.25, -1.5) {};
		\node [style=node 0] (76) at (2.25, -1) {};

        \begin{scope}[xshift=4.3cm,yshift=-5cm]
            \foreach \l[count=\x] in {30,31,32,33,34} {
            	   \node [style=node 0] (\l) at (72+\x*72:.5) {};
            };
        \end{scope}        
		\node [style=node 0] (35) at (1.75, -5.5) {};
		\node [style=node 0] (36) at (2.25, -5.5) {};
		\node [style=node 0] (37) at (2.25, -4.5) {};
		\node [style=node 0] (38) at (1.75, -4.5) {};
		\node [style=node 0] (41) at (2.75, -5.5) {};
		\node [style=node 0] (42) at (3.25, -5.5) {};
		\node [style=node 0] (43) at (3.25, -4.5) {};
		\node [style=node 0] (44) at (2.75, -4.5) {};

        \begin{scope}[xshift=0.5cm]
        \begin{scope}[xshift=5.8cm,yshift=-0.75cm]
            \foreach \l[count=\x] in {103,104,105,106,107} {
            	   \node [style=node 0] (\l) at (72+\x*72:.5) {};
            };
        \end{scope}        
		\node [style=node 0] (108) at (3.25, -1.25) {};
		\node [style=node 0] (109) at (3.75, -1.25) {};
		\node [style=node 0] (110) at (3.75, -0.25) {};
		\node [style=node 0] (111) at (3.25, -0.25) {};
		\node [style=node 0] (112) at (4.25, -1.25) {};
		\node [style=node 0] (113) at (4.75, -1.25) {};
		\node [style=node 0] (114) at (4.75, -0.25) {};
		\node [style=node 0] (115) at (4.25, -0.25) {};
        \end{scope}

        \begin{scope}[xshift=6.3cm,yshift=-5.2cm]
            \foreach \l[count=\x] in {125,126,127,128,129} {
            	   \node [style=node 0] (\l) at (-18+\x*72:.5) {};
            };
        \end{scope}

        \begin{scope}[xshift=0.55cm]
		\node [style=node 0] (61) at (5.25, -2.25) {};
		\node [style=node 0] (62) at (5.25, -2.75) {};
		\node [style=node 0] (63) at (6.25, -2.75) {};
		\node [style=node 0] (64) at (6.25, -2.25) {};
		\node [style=node 0] (65) at (5.25, -3.25) {};
		\node [style=node 0] (66) at (5.25, -3.75) {};
		\node [style=node 0] (67) at (6.25, -3.75) {};
		\node [style=node 0] (68) at (6.25, -3.25) {};
		\node [style=none] (130) at (7.5, -1.5) {\footnotesize Output};
        \end{scope}
	\end{pgfonlayer}
	\begin{pgfonlayer}{edgelayer}
		\draw (0) to (4);
		\draw (4) to (3);
		\draw (3) to (2);
		\draw (2) to (1);
		\draw (1) to (0);
		\draw (0) to (3);
		\draw (0) to (2);
		\draw (4) to (1);
		\draw (4) to (2);
		\draw (1) to (3);
		\draw (30) to (34);
		\draw (34) to (33);
		\draw (33) to (32);
		\draw (32) to (31);
		\draw (31) to (30);
		\draw (30) to (33);
		\draw (30) to (32);
		\draw (34) to (31);
		\draw (34) to (32);
		\draw (31) to (33);
		\draw (36) to (35);
		\draw (36) to (38);
		\draw (37) to (35);
		\draw (37) to (38);
		\draw (42) to (41);
		\draw (42) to (44);
		\draw (43) to (41);
		\draw (43) to (44);
		\draw (36) to (44);
		\draw (36) to (41);
		\draw (41) to (37);
		\draw (37) to (44);
		\draw (62) to (61);
		\draw (62) to (64);
		\draw (63) to (61);
		\draw (63) to (64);
		\draw (66) to (65);
		\draw (66) to (68);
		\draw (67) to (65);
		\draw (67) to (68);
		\draw (62) to (68);
		\draw (62) to (65);
		\draw (65) to (63);
		\draw (63) to (68);
		\draw (70) to (69);
		\draw (70) to (72);
		\draw (71) to (69);
		\draw (71) to (72);
		\draw (74) to (73);
		\draw (74) to (76);
		\draw (75) to (73);
		\draw (75) to (76);
		\draw (70) to (76);
		\draw (70) to (73);
		\draw (73) to (71);
		\draw (71) to (76);
		\draw (42) to (30);
		\draw (42) to (31);
		\draw (43) to (30);
		\draw (43) to (31);
		\draw (103) to (107);
		\draw (107) to (106);
		\draw (106) to (105);
		\draw (105) to (104);
		\draw (104) to (103);
		\draw (103) to (106);
		\draw (103) to (105);
		\draw (107) to (104);
		\draw (107) to (105);
		\draw (104) to (106);
		\draw (109) to (108);
		\draw (109) to (111);
		\draw (110) to (108);
		\draw (110) to (111);
		\draw (113) to (112);
		\draw (113) to (115);
		\draw (114) to (112);
		\draw (114) to (115);
		\draw (109) to (115);
		\draw (109) to (112);
		\draw (112) to (110);
		\draw (110) to (115);
		\draw (113) to (103);
		\draw (113) to (104);
		\draw (114) to (103);
		\draw (114) to (104);
		\draw [style=input edge] (117) to (69);
		\draw [style=input edge] (117) to (72);
		\draw [style=input edge] (116) to (69);
		\draw [style=input edge] (116) to (72);
		\draw [style=input edge] (1) to (74);
		\draw [style=input edge] (1) to (75);
		\draw [style=input edge] (2) to (74);
		\draw [style=input edge] (2) to (75);
		\draw (118) to (122);
		\draw (122) to (121);
		\draw (121) to (120);
		\draw (120) to (119);
		\draw (119) to (118);
		\draw (118) to (121);
		\draw (118) to (120);
		\draw (122) to (119);
		\draw (122) to (120);
		\draw (119) to (121);
		\draw [style=input edge] (124) to (35);
		\draw [style=input edge] (124) to (38);
		\draw [style=input edge] (123) to (35);
		\draw [style=input edge] (123) to (38);
		\draw [style=input edge] (119) to (35);
		\draw [style=input edge] (119) to (38);
		\draw [style=input edge] (120) to (35);
		\draw [style=input edge] (120) to (38);
		\draw [style=input edge] (71) to (111);
		\draw [style=input edge] (72) to (111);
		\draw [style=input edge] (76) to (108);
		\draw [style=input edge] (75) to (108);
		\draw [style=input edge] (38) to (111);
		\draw [style=input edge] (37) to (111);
		\draw [style=input edge] (44) to (108);
		\draw [style=input edge] (43) to (108);
		\draw (125) to (129);
		\draw (129) to (128);
		\draw (128) to (127);
		\draw (127) to (126);
		\draw (126) to (125);
		\draw (125) to (128);
		\draw (125) to (127);
		\draw (129) to (126);
		\draw (129) to (127);
		\draw (126) to (128);
		\draw [style=input edge] (108) to (61);
		\draw [style=input edge] (109) to (61);
		\draw [style=input edge] (112) to (64);
		\draw [style=input edge] (113) to (64);
		\draw [style=<-, shorten <=0pt, shorten >=-2pt] (63) to (130);
		\draw [style=input edge] (125) to (66);
		\draw [style=input edge] (125) to (67);
		\draw [style=input edge] (126) to (66);
		\draw [style=input edge] (126) to (67);
	\end{pgfonlayer}

    	\begin{pgfonlayer}{background}
        \node[zero_gate_region, fit=(116)(117)] {};
        \node[zero_gate_region, fit=(123)(124)] {};

        \node[one_gate_region, one_region, fit=(0)(1)(2)(3)(4)] {};
        \node[one_gate_region, one_region, fit=(118)(119)(120)(121)(122)] {};
        \node[one_gate_region, one_region, fit=(125)(126)(127)(128)(129)] {};

        \node[and_gate_region, fit=(69)(70)(71)(72)(73)(74)(75)(76)] {};

        \node[or_gate_region, one_region, fit=(30)(31)(32)(33)(34)(35)(36)(37)(38)(41)(42)(43)(44)] {};

        \node[or_gate_region, one_region, fit=(103)(104)(105)(106)(107)(108)(109)(110)(111)(112)(113)(114)(115)] {};

        \node[and_gate_region, one_region, fit=(61)(62)(63)(64)(65)(66)(67)(68)] {};
    \end{pgfonlayer}
\end{tikzpicture}
\end{document}
    \caption{An example for the reduction described in~\Cref{result:mcv_to_2-e_kcore_approx}, for~$k=2$, using the example circuit from~\Cref{fig:example_circuit}.}
    \Description{Example of the approximate core reduction using the example circuit.}
    \label{fig:example_core_approx_mcvp}
\end{figure}

    We argue that, at any time, if the circuit value of~$C$ is~1, then~${K_{s^*} \geq 2k}$, otherwise~${K_{s^*} \leq k+1}$. First note that all vertices in a 1-gate graph are in the~$2k$-core of~$G_C$, and no vertex in a 0-gate graph is in the~$(k+1)$-core (hence neither in the $(k+2)$-core) of~$G_C$, because they have maximum degree~$k$. The output vertices of the graph for an~\BAND-gate are in the~$2k$-core if the vertices in both~$I_1$ and~$I_2$ are in the~$2k$-core, and are not in the~$(k+2)$-core if no vertex in~$I_1$ or~no vertex in~$I_2$ is in the~$(k+2)$-core. An analogous argument holds for an~\BOR-gate (the~$K_{2k+1}$ on the right in~\Cref{fig:example_core_approx_mcvp} assures that only one between~$I_1$ and~$I_2$ is sufficient to push all vertices of the gate into the~$2k$-core). The fact that all vertices from the~1-gate graphs are in the~$2k$-core propagates through~$G_C$ reaching~$s^*$ if the circuit value is~1, otherwise it does not propagate and~$s^*$ is not in the~$(k+2)$-core. This proves the correctness of the reduction.

    If~$C$ has size~$N$, then~$G_C$ will have~$N'=\Oh(Nk^2)$ vertices and~$M'=\Oh(Nk^3)$ edges. Each update in~$C$ corresponds to~$\Oh(k^2)$ updates in~$G_C$ and each query in~$C$ corresponds to a single query in~$G_C$. So~\MCV can be solved in~$\dynOh{p(N',M')}{k^2 u(N',M')}{q(N',M')}$ time, and remember we can pick~$k=\Oh(\epsilon^{-1})$.
\end{proof}

The bounds on approximation ratios for efficient algorithms are not yet tight. Indeed, in the following sections we show that a polylogarithmic algorithm with ratio~${2-\epsilon}$ is unlikely to exist. \citet{sun_fully_2020} gave a polylogarithmic algorithm with ratio~${4+\epsilon}$, and it is still not known whether a polylogarithmic algorithm with ratio between~$2$ and~$4$ exists. \citet{sun_fully_2020} also gave a polylogarithmic algorithm for the incremental \ApproxCore{$(2+\epsilon)$}, so the gap for incremental algorithms is much smaller, but it is still unknown whether an efficient 2-approximation exists.

These bounds are similar to the ones by~\citet{anderson_p-complete_1984} for \NC approximation algorithms for core value on static graphs, who showed an \NC algorithm with ratio $2-\epsilon$ is unlikely to exist and an \NC algorithm with ratio~$2+\epsilon$ exists. This highlights the similarity between hardness to parallelize (\P-hardness under \NC reductions) and hardness to dynamize (\P-hardness under \IPL reductions), a topic that will be discussed further in~\Cref{sec:ipl_lower}.

\subsection{Truss decomposition} \label{sec:truss}

The notion of \kcores and core value can be generalized to edges~\cite{cohen_trusses_2008,saito_extracting_2006} instead of vertices. The~\deff{\ktruss} of a graph is its maximal subgraph in which each edge is in at least~$k-2$ triangles. This definition captures a stricter notion of connectedness, and identifies denser cohesive subgraphs than \kcores. For example, the \ktruss of a graph is contained in its $(k-1)$-core, but the opposite does not necessarily hold. \Cref{fig:truss_example} shows an example of the truss decomposition of a graph. There exists an~$\Oh(m^{1.5})$ algorithm~\cite{wang_truss_2012} for finding the truss decomposition of a graph. We can define truss value and a \problem{TrussValue} problem analogous to \CoreValue, and we can prove lower bounds on a weaker version of this problem as follows.

\begin{figure}[h]
    \centering
    \documentclass[tikz]{standalone}
\usepackage{tikz}
\usetikzlibrary{backgrounds, calc}

\begin{document}
\begin{tikzpicture}[
    vertex/.style={circle, draw=black, fill=white, thick, inner sep=1pt, minimum size=7pt},
    every edge/.style={draw, black, thick},
    truss4_region/.style={fill=red!20, rounded corners=12pt},
    truss3_region/.style={fill=blue!20, rounded corners=15pt},
    truss2_region/.style={fill=green!20, rounded corners=18pt},
    scale=0.7
]
    \node[vertex] (a) at (0,1.5) {A};
    \node[vertex] (b) at (2,1.5) {B};
    \node[vertex] (c) at (2,-0.5) {C};
    \node[vertex] (d) at (0,-0.5) {D};

    \node[vertex] (e) at (3.5,0.5) {E};
    \node[vertex] (f) at (1,-2) {F};
    
    \node[vertex] (g) at (3,-1.5) {G};

    \pgfdeclarelayer{background}
    \pgfsetlayers{background,main}

    \begin{pgfonlayer}{background}
        \fill[truss2_region] 
            ($(f.south west)+(-0.5,-0.5)$) -- 
            ($(d.south west)+(-0.5,-0.2)$) -- 
            ($(a.north west)+(-0.5,0.5)$) -- 
            ($(b.north east)+(0.5,0.5)$) -- 
            ($(e.north east)+(0.5,0.5)$) -- 
            ($(e.east)+(0.5,0)$) -- 
            ($(g.east)+(0.5,-0.3)$) -- 
            ($(g.south east)+(0.3,-0.5)$) -- 
            ($(f.south)+(0,-0.5)$) --
            cycle;

        \fill[truss3_region] 
            ($(f.south)+(-0.5,-0.4)$) -- 
            ($(d.south west)+(-0.4,-0.1)$) -- 
            ($(a.north west)+(-0.4,0.4)$) -- 
            ($(b.north east)+(0.4,0.4)$) -- 
            ($(e.north east)+(0.4,0.4)$) -- 
            ($(e.south east)+(0.4,-0.4)$) --
            ($(c.south east)+(0.4,-0.4)$) --
            ($(f.south)+(0.5,-0.4)$) --
            cycle;

        \fill[truss4_region] 
            ($(d.south west)+(-0.3,-0.3)$) rectangle ($(b.north east)+(0.3,0.3)$);
    \end{pgfonlayer}

    \draw (a) -- (b); \draw (a) -- (c); \draw (a) -- (d);
    \draw (b) -- (c); \draw (b) -- (d); \draw (c) -- (d);

    \draw (b) -- (e); \draw (c) -- (e);
    \draw (c) -- (f); \draw (d) -- (f);
    
    \draw (e) -- (g);
    \draw (g) -- (f);

\end{tikzpicture}
\end{document}
    \caption{A graph with its \textcolor{green!80!black}{2-truss}, \textcolor{blue}{3-truss}, and \textcolor{red}{4-truss} highlighted.}
    \Description{Example graph with the 2-truss, 3-truss, and 4-truss subgraphs highlighted.}
    \label{fig:truss_example}
\end{figure}
\begin{problemStatement}{\problem{$k$-Truss}}
    Given a graph~$G$ with a fixed edge~$e^*$, process a series of edge insertions/deletions, and queries: Is~$e^*$ in the~$k$-truss of~$G$?
\end{problemStatement}

\begin{proposition} \label{res:mcv_to_truss}
    For any~$k \geq 4$, if fully (partially) dynamic \problem{$k$-Truss} on a graph with~$n$ vertices and~$m$ edges can be solved in~$\dynTime{p(n,m)}{u(n,m)}{q(n,m)}$ time, then fully (partially) dynamic \MCV on a circuit of size~$N$ can be solved in~$\dynOh{p(N',M')}{u(N',M')}{q(N',M')}$ time, with~$N'=\Oh(N+k)$ and~$M'=\Oh(Nk)$.
\end{proposition}

\begin{proof}
    Let us consider the case~$k=4$. Given a circuit~$C$ with output gate~$g^*$, we will describe a graph~$G_C$ with edge~$e^*$ such that~$e^*$ is in the $4$-truss if and only if the circuit value of~$C$ is~1.

    Gates in~$C$ correspond to subgraphs in~$G_C$ as depicted in~\Cref{fig:truss_mcvp}, where each gate~$i$ has one or two output edges~$O_i$, and wires from a gate~$j$ to gate~$i$ are depicted as red dashed edges. Edge~$e^*$ will be an output edge of gate~$g^*$, and this concludes the description of~$G_C$, with an example shown in~\Cref{fig:example_truss_mcvp}. The insertion or deletion of wires in~$C$ will be simulated by the insertion or deletion of the corresponding edges in~$G_C$, preserving the relation between~$C$ and~$G_C$. 

\begin{figure*}[h]
\centering
\captionsetup[subfigure]{justification=centering}
\begin{subfigure}{0.32\textwidth}
    \centering
    \begin{tikzpicture}
	\begin{pgfonlayer}{nodelayer}
		\node [style=node 0] (17) at (-1, 0) {};
		\node [style=node 0] (32) at (1, 0) {};
		\node [style=node 0] (33) at (-1, 1) {};
		\node [style=node 0] (34) at (1, 1) {};
		\node [style=hollow] (35) at (-0.5, -0.5) {};
		\node [style=hollow] (36) at (0.5, -0.5) {};
	\end{pgfonlayer}
	\begin{pgfonlayer}{edgelayer}
		\draw (17) to (32);
		\draw [style=thick edge] (33) to (34);
		\draw (17) to (35);
		\draw (35) to (36);
		\draw (36) to (32);
		\draw (17) to (36);
		\draw (35) to (32);
	\end{pgfonlayer}
\end{tikzpicture}
    \subcaption{A double edge and below the subgraph it represents \label{fig:truss_mcvp:double}}
    \Description{Two vertices with a double edge between them. Below the subgraph it represents, a K4 from which we are using a single edge.}
\end{subfigure}%
\begin{subfigure}{0.25\textwidth}
    \centering
    \begin{tikzpicture}
	\begin{pgfonlayer}{nodelayer}
		\node [style=node 0] (0) at (0, 1) {};
		\node [style=node 0] (1) at (0, -1) {};
		\node [style=none] (2) at (0.5, 0) {$O_i$};
	\end{pgfonlayer}
	\begin{pgfonlayer}{edgelayer}
		\draw [style=output edge] (0) to (1);
	\end{pgfonlayer}
\end{tikzpicture}
    \subcaption{$\beta_i = 0$}
    \Description{Two vertices with a single edge between them.}
\end{subfigure}%
\begin{subfigure}{0.25\textwidth}
    \centering
    \begin{tikzpicture}
	\begin{pgfonlayer}{nodelayer}
		\node [style=node 0] (0) at (0, 1) {};
		\node [style=node 0] (1) at (0, -1) {};
		\node [style=none] (4) at (0.5, 0) {$O_i$};
	\end{pgfonlayer}
	\begin{pgfonlayer}{edgelayer}
		\draw [style=thick output edge] (0) to (1);
	\end{pgfonlayer}
\end{tikzpicture}
    \subcaption{$\beta_i = 1$}
    \Description{Two vertices with a double edge between them.}
\end{subfigure}%
\\ \vspace{3mm}
\begin{subfigure}{0.40\textwidth}
    \centering
    \begin{tikzpicture}
	\begin{pgfonlayer}{nodelayer}
		\node [style=node 0] (2) at (2, 2) {};
		\node [style=node 0] (4) at (0, 0) {};
		\node [style=node 0] (5) at (1, 1) {};
		\node [style=none] (6) at (0.7, 0.3) {$O_i$};
		\node [style=none] (7) at (-0.25, -0.65) {$O_i$};
		\node [style=node 0] (14) at (2, -1) {};
		\node [style=node 0] (20) at (-2, 0) {};
		\node [style=node 0] (21) at (-2, -1) {};
		\node [style=none] (22) at (-2.25, -0.5) {$O_k$};
		\node [style=node 0] (23) at (-1, -1) {};
		\node [style=node 0] (24) at (-1, 2) {};
		\node [style=node 0] (25) at (-2, 1) {};
		\node [style=node 0] (26) at (-2, 2) {};
		\node [style=none] (27) at (-2.25, 1.5) {$O_j$};
	\end{pgfonlayer}
	\begin{pgfonlayer}{edgelayer}
		\draw [style=output edge] (4) to (5);
		\draw (5) to (2);
		\draw (4) to (14);
		\draw (5) to (14);
		\draw [style=thick edge] (14) to (2);
		\draw [style=output edge] (20) to (21);
		\draw (23) to (24);
		\draw [style=input edge] (25) to (24);
		\draw [style=output edge] (25) to (26);
		\draw [style=thick input edge] (25) to (23);
		\draw [style=thick input edge] (26) to (24);
		\draw (24) to (23);
		\draw [style=input edge] (20) to (23);
		\draw [style=thick input edge] (20) to (24);
		\draw [style=thick input edge] (21) to (23);
		\draw [style=thick edge] (24) to (2);
		\draw (5) to (24);
		\draw (4) to (24);
		\draw [style=thick edge] (14) to (23);
		\draw [style=output edge] (23) to (4);
	\end{pgfonlayer}
\end{tikzpicture}
    \subcaption{$\beta_i = \vee(\beta_j, \beta_k)$} \label{fig:truss_mcvp_or}
    \Description{Graph with vertices A through J, double edges AB, BC, CD, single edges AD, AE, AF, EC, EF, output edges DE and EF, input double edges GA, HD, IA, JD, input single edges HA, ID, and the output edges from the previous gate are GH and IJ. }
\end{subfigure}%
\begin{subfigure}{0.58\textwidth}
    \centering
    \documentclass[tikz]{standalone}
\usepackage{common/tikzit}

\begin{document}
\begin{tikzpicture}
	\begin{pgfonlayer}{nodelayer}
		\node [style=none] (6) at (0.7, 0.3) {$O_i$};
		\node [style=none] (7) at (-0.25, -0.65) {$O_i$};
		\node [style=none] (10) at (-2.25, 0.5) {$O_j$};
		\node [style=none] (24) at (3.3, 0.5) {$O_k$};
		\node [style=node 0] (0) at (-1, -1) {};
		\node [style=node 0] (1) at (-1, 2) {};
		\node [style=node 0] (20) at (2, -1) {};
		\node [style=node 0] (21) at (2, 2) {};
		\node [style=node 0] (4) at (0, 0) {};
		\node [style=node 0] (5) at (1, 1) {};
		\node [style=node 0] (8) at (-2, 1) {};
		\node [style=node 0] (9) at (-2, 0) {};
		\node [style=node 0] (22) at (3, 0) {};
		\node [style=node 0] (23) at (3, 1) {};
	\end{pgfonlayer}
	\begin{pgfonlayer}{edgelayer}
		\draw [style=output edge] (0) to (4);
		\draw [style=output edge] (4) to (5);
		\draw (1) -- (0) node[midway,left,xshift=1pt] {\footnotesize $e_1$};
		\draw (5) -- (1) node[midway,above] {\footnotesize $e_3$};
		\draw (4) -- (1) node[midway, above,xshift=3pt] {\footnotesize $e_2$};
		\draw [style=thick input edge] (8) to (1);
		\draw [style=thick input edge] (9) to (0);
		\draw [style=output edge] (8) to (9);
		\draw [style=input edge] (8) to (0);
		\draw (20) -- (21);
		\draw [style=input edge] (22) to (21);
		\draw [style=output edge] (22) to (23);
		\draw [style=thick input edge] (22) to (20);
		\draw [style=thick input edge] (23) to (21);
		\draw (5) -- (21) node[midway,above,xshift=-1pt] {\footnotesize $e_4$};
		\draw [style=thick edge] (1) to (21);
		\draw (4) -- (20) node[midway,above] {\footnotesize $e_6$};
		\draw (5) -- (20) node[midway,above,xshift=3pt] {\footnotesize $e_5$};
		\draw [style=thick edge] (0) to (20);
	\end{pgfonlayer}
\end{tikzpicture}

\end{document}
    \subcaption{$\beta_i = \wedge(\beta_j, \beta_k)$} \label{fig:truss_mcvp_and}
    \Description{Graph with vertices A through J, double edges AB, CD, single edges AD, BC, AE, AF, EC, EF, output edges DE and EF, input double edges GA, HD, IB, JC, input single edges GD, JB, and the output edges from the previous gate are GH and IJ. }
\end{subfigure}%
\caption{Building blocks of the reduction from \MCV to \fourTruss,~$\beta_i$ is the value of gate~$i$. Dashed \textcolor{red}{red} edges represent the wires~$ji$ and~$ki$. Double edges are actually a~$K_4$ from which we are using a single edge, as in (a). With this reduction, each edge~$O_i$ is in the~$4$-truss if and only if the value of gate~$i$ is~$1$.} \label{fig:truss_mcvp}
\end{figure*}

\begin{figure}[h] 
    \centering
    \documentclass[tikz, border=10mm]{standalone}
\usetikzlibrary{backgrounds, fit}
\usepackage{common/tikzit}

\pgfdeclarelayer{background}
\pgfdeclarelayer{edgelayer}
\pgfdeclarelayer{nodelayer}
\pgfsetlayers{background,edgelayer,nodelayer,main}

\begin{document}
\begin{tikzpicture}[
    node 0/.style={circle, draw=black, fill=black, inner sep=1.5pt},
    hollow/.style={circle, draw=black, fill=white, inner sep=1.5pt},
    none/.style={},
    input edge/.append style={thick},
    thick edge/.append style={double=cyan!15},
    thick input edge/.append style={thick,double=red!15},
    region/.style={rounded corners, fill opacity=0.15, draw, inner sep=3pt},
	one_region/.style={line width=1.5pt},
    zero_gate_region/.style={region, fill=gray},
    one_gate_region/.style={region, fill=green},
    and_gate_region/.style={region, fill=orange},
    or_gate_region/.style={region, fill=cyan},
    3core/.style={},
    scale=0.65,
]
	\begin{pgfonlayer}{nodelayer}
		\node [style=node 0] (94) at (-2, 4.75) {};
		\node [style=node 0] (99) at (-2, 5.25) {};
		\node [style=node 0] (126) at (-2, 1.25) {};
		\node [style=node 0] (127) at (-2, 1.75) {};
		\node [style=node 0] (76) at (-2, -1.5) {};
		\node [style=node 0] (77) at (-2, -1) {};
		\node [style=node 0] (128) at (-2, -3.5) {};
		\node [style=node 0] (129) at (-2, -3) {};
        
		\node [style=node 0] (120) at (1.25, 2.25) {};
		\node [style=node 0] (121) at (1.25, 4.25) {};
		\node [style=node 0] (122) at (-0.25, 2.75) {};
		\node [style=node 0] (123) at (0.75, 3.75) {};
		\node [style=node 0] (124) at (-0.75, 2.25) {};
		\node [style=node 0] (125) at (-0.75, 4.25) {};
		
        \node [style=node 0] (46) at (-0.25, -2.75) {};
		\node [style=node 0] (47) at (0.75, -1.75) {};
		\node [style=node 0] (52) at (-0.75, -3.25) {};
		\node [style=node 0] (53) at (1.25, -3.25) {};
		\node [style=node 0] (56) at (1.25, -1.25) {};
		\node [style=node 0] (57) at (-0.75, -1.25) {};
        
		\node [style=node 0] (114) at (3.5, 0) {};
		\node [style=node 0] (115) at (4.5, 1) {};
		\node [style=node 0] (116) at (3, -0.5) {};
		\node [style=node 0] (117) at (5, -0.5) {};
		\node [style=node 0] (118) at (5, 1.5) {};
		\node [style=node 0] (119) at (3, 1.5) {};
        
		\node [style=node 0] (136) at (5, -3) {};
		\node [style=node 0] (137) at (5, -2.5) {};
        
		\node [style=node 0] (130) at (8.25, -2) {};
		\node [style=node 0] (131) at (8.25, 0) {};
		\node [style=node 0] (132) at (6.75, -1.5) {};
		\node [style=node 0] (133) at (7.75, -0.5) {};
		\node [style=node 0] (134) at (6.25, -2) {};
		\node [style=node 0] (135) at (6.25, 0) {};
        
		\node [style=none] (138) at (7.75, 1.5) {\footnotesize Output};
		\node [style=none] (139) at (7.25, -1) {};
	\end{pgfonlayer}
	\begin{pgfonlayer}{edgelayer}
		\draw (46) to (47);
		\draw (52) to (53);
		\draw (56) to (57);
		\draw (47) to (56);
		\draw (46) to (57);
		\draw [style=thick edge] (56) to (57);
		\draw (57) to (47);
		\draw (46) to (53);
		\draw (47) to (53);
		\draw [style=thick edge] (52) to (53);
		\draw (46) to (52);
		\draw [style=thick edge] (53) to (56);
		\draw (52) to (57);
		\draw (57) to (52);
		\draw (77) to (76);
		\draw (99) to (94);
		\draw (114) to (115);
		\draw (116) to (117);
		\draw (118) to (119);
		\draw (115) to (118);
		\draw (114) to (119);
		\draw [style=thick edge] (118) to (119);
		\draw (119) to (115);
		\draw (114) to (117);
		\draw (115) to (117);
		\draw [style=thick edge] (116) to (117);
		\draw (114) to (116);
		\draw [style=thick edge] (117) to (118);
		\draw (116) to (119);
		\draw (119) to (116);
		\draw [style=thick edge,double=orange!15] (120) to (121);
		\draw (122) to (123);
		\draw (123) to (121);
		\draw (122) to (120);
		\draw (123) to (120);
		\draw (124) to (125);
		\draw (122) to (125);
		\draw [style=thick edge,double=orange!15] (125) to (124);
		\draw (124) to (120);
		\draw (121) to (125);
		\draw (122) to (124);
		\draw (125) to (123);
		\draw [style=thick edge,double=green!15] (127) to (126);
		\draw [style=thick edge,double=green!15] (129) to (128);
		\draw [style=thick input edge] (94) to (125);
		\draw [style=thick input edge] (99) to (121);
		\draw [style=input edge] (99) to (125);
		\draw [style=thick input edge] (127) to (124);
		\draw [style=thick input edge] (126) to (120);
		\draw [style=input edge] (126) to (124);
		\draw [style=thick input edge] (77) to (57);
		\draw [style=thick input edge] (76) to (52);
		\draw [style=thick input edge] (129) to (57);
		\draw [style=thick input edge] (128) to (52);
		\draw [style=input edge] (76) to (57);
		\draw [style=input edge] (129) to (52);
		\draw [style=thick edge,double=orange!15] (130) to (131);
		\draw (132) to (133);
		\draw (133) to (131);
		\draw (132) to (130);
		\draw (133) to (130);
		\draw (134) to (135);
		\draw (132) to (135);
		\draw [style=thick edge,double=orange!15] (135) to (134);
		\draw (134) to (130);
		\draw (131) to (135);
		\draw (132) to (134);
		\draw (135) to (133);
		\draw [style=thick edge,double=green!15] (137) to (136);
		\draw [style=thick input edge] (137) to (134);
		\draw [style=thick input edge] (136) to (130);
		\draw [style=input edge] (136) to (134);
		\draw [style=thick input edge] (46) to (119);
		\draw [style=thick input edge] (47) to (116);
		\draw [style=input edge] (46) to (116);
		\draw [style=thick input edge] (122) to (116);
		\draw [style=thick input edge] (123) to (119);
		\draw [style=input edge] (122) to (119);
		\draw [style=thick input edge] (115) to (131);
		\draw [style=thick input edge] (114) to (135);
		\draw [style=input edge] (115) to (135);
		\draw [style=->, shorten <= -3pt, shorten >= 2pt] (138) to (139);
	\end{pgfonlayer}

    \begin{pgfonlayer}{background}
        \node[zero_gate_region, fit=(94)(99)] {};
        \node[zero_gate_region, fit=(76)(77)] {};

        \node[one_gate_region, one_region, fit=(126)(127)] {};
        \node[one_gate_region, one_region, fit=(128)(129)] {};
        \node[one_gate_region, one_region, fit=(136)(137)] {};

        \node[and_gate_region, fit=(120)(121)(122)(123)(124)(125)] {};

        \node[or_gate_region, one_region, fit=(46)(47)(52)(53)(56)(57)] {};

        \node[or_gate_region, one_region, fit=(114)(115)(116)(117)(118)(119)] {};

        \node[and_gate_region, one_region, fit=(130)(131)(132)(133)(134)(135)] {};
    \end{pgfonlayer}
\end{tikzpicture}

\end{document}
    \caption{An example of the reduction described in~\Cref{res:mcv_to_truss}, using double edges from~\Cref{fig:truss_mcvp:double} and the example circuit from~\Cref{fig:example_circuit}.}
    \Description{Example of the reduction from monotone circuit value to 4-truss using the example circuit.}
    \label{fig:example_truss_mcvp}
\end{figure}

    Let us argue that, at any time,~$e^*$ is in the~$4$-truss of~$G_C$ if and only if the circuit value of~$C$ is~1. First note that the output edge in the graph for a~1-gate is in the~$4$-truss of~$G_C$, and the output edge for a 0-gate is not in the~$4$-truss, because it is in at most one triangle even when the edges representing a wire are added. For the \BAND-gate in~\Cref{fig:truss_mcvp_and}, if~$O_j$ is not in the~$4$-truss, then neither is~$e_1$, and consequently neither are~$e_2, \ldots, e_6$, nor both output edges~$O_i$, even if each~$O_i$ has a wire connected to it. Conversely,~$O_i$ is always in the~$4$-truss when both~$O_j$ and~$O_k$ are in the~$4$-truss. The analogous properties hold for the \BOR-gate in~\Cref{fig:truss_mcvp_or}. This way, the~$4$-truss propagates from the output edges of the 1-gate graphs, reaching edge $e^*$ if and only if the circuit value of~$C$ is~1.
    
    When~$k > 4$, the graph~$G_C$ contains~$k-4$ additional vertices connected to each other and to all other vertices. Each original edge will be connected to $k-4$ other vertices, and so will have exactly~$k-4$ extra triangles. Additionally, the original graph has a~4-truss if and only if it forms a~$k$-truss using these additional vertices. Thus, the output edge~$e^*$ is in the~$k$-truss if and only if the circuit value of~$C$ is~1.

    If~$C$ has size~$N$, then~$G_C$ has~$N'=\Oh(N+k)$ vertices and~$M'=\Oh(Nk)$ edges. For any~$k \geq 4$, each update in~$C$ can be simulated with~$7=\Oh(1)$ updates in~$G_C$, and each query in~$C$ corresponds to a query in~$G_C$. Thus, \MCV on a circuit of size~$N$ can be solved in~$\dynOh{p(N',M')}{u(N',M')}{q(N',M')}$ time.
\end{proof}


\subsection{Directed core decomposition} \label{sec:directed_core}

\citet{giatsidis_d-cores_2011} extended the notion of \kcores to directed graphs. The~\klcore of a directed graph is its maximal subgraph in which each vertex has in-degree at least~$k$ and out-degree at least~$\ell$. Computing a fixed~\klcore on a directed graph on~$m$ arcs can also be done with an~$\Oh(m)$ algorithm, but vertices do not have a single analog to ``core value'', so the decomposition problem is harder, and less explored. Still, there are algorithms with non-trivial space-complexity to speed up listing all vertices in any~\klcore in time proportional to its size, after some preprocessing~\cite{fang_effective_2019,giatsidis_d-cores_2011}.

\begin{problemStatement}{\problem{$(k,\ell)$-Core}}
    Given a directed graph~$G$ with a fixed vertex~$s^*$, process a series of arc insertions/deletions, and queries: Is~$s^*$ in the~$(k,\ell)$-core of~$G$?
\end{problemStatement}

The~\problem{$(k,\ell)$-Core} problem is the directed version of the~\problem{$k$-Core} problem. It is easy to create an algorithm for~\problem{$k$-Core} from an algorithm for~\problem{$(k,k)$-Core}, by simply replacing each edge with two arcs, one in each direction.
We can thus prove that \MCV can be reduced to \problem{$(k,k)$-Core} for any~$k \geq 3$ like in~\Cref{result:mcv_to_3core}. However, we can strengthen this by taking advantage of the directed edges and reducing \MCV directly to~\problem{$(2,0)$-Core}, yielding the following.

\begin{proposition} \label{result:mcv_to_02core}
  For any~$k\geq 2$ and~$\ell \geq 0$, if fully (partially) dynamic \problem{$(k, \ell)$-Core} on a directed graph on $n$ vertices and $m$ arcs can be solved in 
  $\dynTime{p(n, m)}{u(n, m)}{q(n, m)}$ time, then fully (partially) dynamic \MCV on a circuit of size~$N$ can be solved in $\dynOh{p(N', M')}{u(N', M')}{q(N', M')}$ time, with~$N'=\Oh(N+k+\ell)$ and~$M'=\Oh(N(k+\ell))$.
\end{proposition}

\begin{proof}
    When~$k=2$ and~$\ell=0$ the proof is analogous to the proof of~\Cref{result:mcv_to_3core}, but using the building blocks from~\Cref{fig:02core_cvp}. As arcs are directed, the construction is simpler than in the undirected version.
    
    When~$k>2$ or~$\ell>0$, add a new directed clique of size~$\max(k-1,\ell+1)$, with~$k-2$ outgoing arcs to each original vertex and~$\ell$ incoming arcs from each original vertex. 
\end{proof}

\begin{figure*}[h]
\centering
\captionsetup[subfigure]{justification=centering}
\begin{subfigure}{0.18\textwidth}
    \centering
\begin{tikzpicture}
	\begin{pgfonlayer}{nodelayer}
		\node [style=node 0] (4) at (2, 0) {};
		\node [style=none] (6) at (2, 0.25) {$O$};
		\node [style=none] (9) at (3, 0) {};
		\node [style=none] (10) at (1, 0.5) {};
		\node [style=none] (11) at (1, -0.5) {};
	\end{pgfonlayer}
	\begin{pgfonlayer}{edgelayer}
		\draw [style=output arc] (4) to (9.center);
	\end{pgfonlayer}
\end{tikzpicture}
    \captionsetup{textformat=simple}
    \subcaption{$0$-gate}
    \Description{One vertex with an output arc.}
    \label{fig:02core_0_cvp}
\end{subfigure}%
\begin{subfigure}{0.26\textwidth}
    \centering
    \begin{tikzpicture}
	\begin{pgfonlayer}{nodelayer}
		\node [style=node 0] (1) at (1, 0.5) {};
		\node [style=node 0] (2) at (1, -0.5) {};
		\node [style=node 0] (4) at (2, 0) {};
		\node [style=none] (6) at (2.1, 0.25) {$O$};
		\node [style=none] (9) at (3, 0) {};
	\end{pgfonlayer}
	\begin{pgfonlayer}{edgelayer}
		\draw [style=output arc] (4) to (9.center);
		\draw [style=directed, bend left=15] (1) to (2);
		\draw [style=directed, bend left=15] (1) to (4);
		\draw [style=directed, bend left=15] (4) to (2);
		\draw [style=directed, bend left=15] (2) to (1);
		\draw [style=directed, bend left=15] (4) to (1);
		\draw [style=directed, bend left=15] (2) to (4);
	\end{pgfonlayer}
\end{tikzpicture}
    \captionsetup{textformat=simple}
    \subcaption{$1$-gate}
    \Description{A complete digraph on three vertices with one output arc from one of them.}
    \label{fig:02core_1_cvp}
\end{subfigure}%
\begin{subfigure}{0.26\textwidth}
    \centering
    \begin{tikzpicture}
	\begin{pgfonlayer}{nodelayer}
		\node [style=node 0] (0) at (0, 0.5) {};
		\node [style=none] (13) at (-1, 1.1) {$I_1$};
		\node [style=hollow] (14) at (-1, 0.75) {};
		\node [style=hollow] (15) at (-1, 0.25) {};
		\node [style=none] (20) at (-1, -0.1) {$I_2$};
		\node [style=none] (21) at (0, 0.85) {$O$};
		\node [style=node 0] (23) at (0.5, -0.5) {};
		\node [style=node 0] (24) at (-0.5, -0.5) {};
		\node [style=none] (25) at (1, 0.75) {};
		\node [style=none] (26) at (1, 0.25) {};
	\end{pgfonlayer}
	\begin{pgfonlayer}{edgelayer}
		\draw [style=input arc] (14) to (0);
		\draw [style=input arc] (15) to (0);
		\draw [style=directed, bend left=15] (24) to (0);
		\draw [style=output arc] (0) to (25.center);
		\draw [style=output arc] (0) to (26.center);
		\draw [style=directed, bend left=15] (24) to (23);
		\draw [style=directed, bend left=15] (0) to (24);
		\draw [style=directed] (0) to (23);
		\draw [style=directed, bend left=15] (23) to (24);
	\end{pgfonlayer}
\end{tikzpicture}
    \captionsetup{textformat=simple}
    \subcaption{\BOR-gate}
    \Description{Vertices A, B, C with all arcs except CA. Vertex A has two input and two output arcs.}
    \label{fig:02core_or_cvp}
\end{subfigure}%
\begin{subfigure}{0.26\textwidth}
    \centering
    \begin{tikzpicture}
	\begin{pgfonlayer}{nodelayer}
		\node [style=node 0] (0) at (0, 0) {};
		\node [style=none] (13) at (-1, 0.6) {$I_1$};
		\node [style=hollow] (14) at (-1, 0.25) {};
		\node [style=hollow] (15) at (-1, -0.25) {};
		\node [style=none] (16) at (1, 0.25) {};
		\node [style=none] (19) at (1, -0.25) {};
		\node [style=none] (20) at (-1, -0.6) {$I_2$};
		\node [style=none] (21) at (0, 0.35) {$O$};
	\end{pgfonlayer}
	\begin{pgfonlayer}{edgelayer}
		\draw [style=output arc] (0) to (16.center);
		\draw [style=output arc] (0) to (19.center);
		\draw [style=input arc] (14) to (0);
		\draw [style=input arc] (15) to (0);
	\end{pgfonlayer}
\end{tikzpicture}
    \captionsetup{textformat=simple}
    \subcaption{\BAND-gate}
    \Description{Single vertex with two input and two output arcs.}
    \label{fig:02core_and_cvp}
\end{subfigure}%
\caption{Building blocks of the reduction from \MCV to \problem{$(2,0)$-Core}.}
\label{fig:02core_cvp}
\end{figure*}

\subsection{Bipartite core decomposition} \label{sec:bipartite_core}

\citet{ahmed_visualisation_2007} extended the notion of \kcores to bipartite graphs. The~$(\alpha,\beta)$-\deff{bcore} of a bipartite graph ${G=(V = V_A \cupdot V_B,E)}$ is its maximal subgraph in which each vertex from~$V_A$ has degree at least~$\alpha$ and each vertex from~$V_B$ has degree at least~$\beta$. The~\problem{$(\alpha,\beta)$-BCore} problem is defined analogously to~\problem{$k$-Core}.

\begin{proposition} \label{res:3core_to_k2_bcore}
  For any~$\alpha \geq 3$ and~$\beta \geq 2$, if fully (partially) dynamic \problem{$(\alpha, \beta)$-BCore} on a bipartite graph on $n$ vertices and $m$ edges can be solved in 
  $\dynTime{p(n, m)}{u(n, m)}{q(n, m)}$ time, then fully (partially) dynamic \problem{$\alpha$-Core} on a graph on~$N$ vertices and~$M$ edges can be solved in $\dynOh{p(N', M')}{u(N', M')}{q(N', M')}$ time, with~$N'=\Oh(M+\alpha+\beta)$ and~$M'=\Oh((M+\alpha)\beta)$.
\end{proposition}

\begin{proof}
    When~$\beta = 2$, given a graph~$G = (V,E)$ on~$N$ vertices and~$M$ edges, we can build its~\deff{incidence graph} $G' = (V \cup E, E')$ with~$E' = \{\{u,e\} : u \in V, u \in e \in E\}$, which is a bipartite graph with~$N+M$ vertices and $2M$ edges. It is easy to see that a vertex~$u \in V$ is in the~$\alpha$-core of~$G$ if and only if~$u$ is in the~$(\alpha,2)$-bcore of~$G'$ (see~\Cref{fig:incidence_example}).

\begin{figure}[h]
    \centering
    \documentclass[border=1cm,tikz]{standalone}

\usetikzlibrary{positioning, shapes.geometric, calc, fit, backgrounds}

\pgfdeclarelayer{background}
\pgfsetlayers{background,main}

\begin{document}

\begin{tikzpicture}[
    vertex/.style={circle, draw=black, fill=white, thick, inner sep=1pt, minimum size=4pt},
    edgevertex/.style={circle, draw=black, fill=black, inner sep=0pt, minimum size=3pt},
    myedge/.style={thick, black},
    k4vertex/.style={vertex},
    k4edgevertex/.style={edgevertex},
    k4edge/.style={myedge},
    k4region/.style={fill=orange!30, draw=orange!50, thick, rounded corners=5pt},
    every label/.style={text height=1.5ex},
        every node/.append style={font=\footnotesize},
        scale=0.7,
    auto]

    \begin{scope}[local bounding box=GraphG]
        
        \node[vertex] (A) at (-2.5, 0) {A};
        
        \node[k4vertex] (B) at (-1, -1.5) {B};
        \node[k4vertex] (D) at (-1, 1.5) {D};
        \node[k4vertex] (C) at (1, -1.5) {C};
        \node[k4vertex] (E) at (1, 1.5) {E};
        
        \node[vertex] (F) at (2.5, 0) {F};

        \draw[myedge] (A) -- node[below left] {a} (B);
        \draw[myedge] (A) -- node[above left] {b} (D);
        \draw[myedge] (F) -- node[below right] {c} (C);
        \draw[myedge] (F) -- node[above right] {d} (E);
        
        \draw[k4edge] (B) -- node[below] {e} (C);
        \draw[k4edge] (C) -- node[right] {f} (E);
        \draw[k4edge] (E) -- node[above] {g} (D);
        \draw[k4edge] (D) -- node[left]  {h} (B);
        
        \draw[k4edge] (B) -- node[above left, pos=0.7] {i} (E);
        \draw[k4edge] (D) -- node[below left, pos=0.7] {j} (C);

        \begin{pgfonlayer}{background}
            \node[k4region, inner sep=5pt, fit=(B) (C) (D) (E)] {};
        \end{pgfonlayer}
    \end{scope}

    \begin{scope}[shift={(5,1.5)}, local bounding box=GraphIG]
        
        \node[vertex]   (vA) at (0, 0) {A};
        \node[k4vertex] (vB) at (1, 0) {B};
        \node[k4vertex] (vC) at (2, 0) {C};
        \node[k4vertex] (vD) at (3, 0) {D};
        \node[k4vertex] (vE) at (4, 0) {E};
        \node[vertex]   (vF) at (5, 0) {F};

        \node[edgevertex, label=below:a] (ea) at (0.0, -3) {};
        \node[edgevertex, label=below:b] (eb) at (0.6, -3) {};
        \node[edgevertex, label=below:c] (ec) at (1.2, -3) {};
        \node[edgevertex, label=below:d] (ed) at (1.8, -3) {};
        
        \node[k4edgevertex, label=below:e] (ee) at (2.4, -3) {};
        \node[k4edgevertex, label=below:f] (ef) at (3.0, -3) {};
        \node[k4edgevertex, label=below:g] (eg) at (3.6, -3) {};
        \node[k4edgevertex, label=below:h] (eh) at (4.2, -3) {};
        \node[k4edgevertex, label=below:i] (ei) at (4.8, -3) {};
        \node[k4edgevertex, label=below:j] (ej) at (5.4, -3) {};

        \draw[myedge] (vA) -- (ea); \draw[myedge] (vB) -- (ea);
        \draw[myedge] (vA) -- (eb); \draw[myedge] (vD) -- (eb);
        \draw[myedge] (vF) -- (ec); \draw[myedge] (vC) -- (ec);
        \draw[myedge] (vF) -- (ed); \draw[myedge] (vE) -- (ed);
        
        \draw[myedge] (vB) -- (ee); \draw[myedge] (vC) -- (ee);
        \draw[myedge] (vC) -- (ef); \draw[myedge] (vE) -- (ef);
        \draw[myedge] (vE) -- (eg); \draw[myedge] (vD) -- (eg);
        \draw[myedge] (vD) -- (eh); \draw[myedge] (vB) -- (eh);
        \draw[myedge] (vB) -- (ei); \draw[myedge] (vE) -- (ei);
        \draw[myedge] (vD) -- (ej); \draw[myedge] (vC) -- (ej);

        \begin{pgfonlayer}{background}
            \filldraw[k4region] 
                ($(vB.north west)+(-0.4, 0.2)$) -- 
                ($(vE.north east)+(0.1, 0.2)$) -- 
                ($(ej.south east)+(0.3, -0.7)$) -- 
                ($(ee.south west)+(-0.0, -0.7)$) -- 
                cycle;
        \end{pgfonlayer}
    \end{scope}

\end{tikzpicture}

\end{document}
    \caption{Example of a graph and its incidence graph. Its $3$-core, and the~$(3,2)$-bcore of the incidence graph are shaded.}
    \Description{Example graph and its incidence graph, with the 3-core and the 3 by 2 bcore highlighted.}
    \label{fig:incidence_example}
\end{figure}

    For any~$\beta > 2$, we can add a~$K_{\alpha,\beta}$ subgraph to~$G'$ and connect~$\beta-2$ vertices in the partition of size $\beta$ to each vertex in~$E$. This way, a vertex~$u \in V$ is in the~$\alpha$-core of~$G$ if and only if it is in the~$(\alpha,\beta)$-bcore of~$G'$, and~$G'$ has~$N'=\Oh(N+M+\alpha+\beta)$ vertices and~$M'=\Oh((M+\alpha)\beta)$. As we assumed that~$M=\Omega(N)$ for simplicity, $N'= \Oh(M+\alpha+\beta)$.

    Whenever an edge~$e=uv$ is inserted/removed from~$G$, we update~$G'$ by inserting/removing~$ue$ and~$ve$. This maintains the relation between~$G$ and~$G'$, and we derive an~$\dynOh{p(N',M')}{u(N',M')}{q(N',M')}$ time algorithm for~\problem{$\alpha$-Core}.    
\end{proof}

By combining \Cref{result:mcv_to_3core,res:3core_to_k2_bcore}, we obtain the following corollary.

\begin{corollary}
    For any~$\alpha \geq 3$ and~$\beta \geq 2$, if fully (partially) dynamic \problem{$(\alpha, \beta)$-BCore} on a bipartite graph on $n$ vertices and $m$ edges can be solved in 
    $\dynTime{p(n, m)}{u(n, m)}{q(n, m)}$ time, then fully (partially) dynamic \MCV on a circuit of size~$N$ can be solved in $\dynOh{p(N',M')}{u(N',M')}{q(N',M')}$ time, with~$N'=\Oh(N\alpha+\beta)$ and~$M'=\Oh(N\alpha\beta)$.
\end{corollary}

\subsection{Planar core decomposition} \label{sec:planar_core}

Now we consider the core decomposition problem in planar graphs. The reduction from~\Cref{result:mcv_to_3core} for $k=3$ preserves planarity when the original circuit is planar. However, no hardness results exist for the planar monotone circuit value problem,\footnote{In some models, such as in parallel computation, the problem can actually be solved efficiently~\cite{ramachandran_efficient_1996}.} so the results from~\Cref{sec:all_lower} do \emph{not} follow for dynamic planar~\kcore.

In \Cref{sec:no_crossing}, we prove that one of the most common techniques to prove hardness of planar problems cannot be used for the~\kcore problem. If this is indeed a hard problem, more intricate techniques will be needed to prove it.

 \section{Hardness results} \label{sec:all_lower}
We start by deriving several hardness results for \MCV using different strategies, and at the end of the section we apply these results for core decomposition variants as well.

\subsection{Computational complexity lower bounds} \label{sec:ipl_lower}
We will first explore hardness results adapted from classic computational complexity, which are similar to~$\NP$-completeness results and the~$\P \neq \NP$ conjecture.
They follow from an intriguing connection between dynamic and parallel problems which seems to have passed unnoticed in the recent years.

In 1994, \citet{miltersen_complexity_1994} defined the following complexity class for dynamic problems. 
A decision problem~$\pi$ is in \complexityClass{incr\text{-}POLYLOGTIME} (we use \IPL, for short) if, after some polynomial preprocessing time over instance~$I$ (a binary string), we can maintain~$\pi(I) \in \{0,1\}$, the answer to~$\pi$ on instance~$I$, given a sequence of 1-bit updates to~$I$ in $\text{polylog} (|I|)$ time per update.
A \deff{polylog incremental reduction} from a decision problem~$\pi_1$ to a decision problem~$\pi_2$ is a polynomial algorithm that maps each instance~$I_1$ of~$\pi_1$ to an instance~$I_2$ of~$\pi_2$ with ${\pi_1(I_1)=\pi_2(I_2)}$, and a polylogarithmic algorithm that maps a 1-bit update to~$I_1$ to a polylog amount of 1-bit updates to~$I_2$, such that~$\pi_1(I_1)=\pi_2(I_2)$ after the updates are applied to both instances.
A decision problem~$\pi$ is~\P-\deff{complete under \IPL reductions} if~$\pi$ is in~\P and there is a polylog incremental reduction from any problem in~\P to~$\pi$. This implies that if~$\pi$ is in \IPL then so is \emph{every} problem in~\P, which is highly unlikely (i.e., it is conjectured that~$\P \not\subseteq \IPL$). For example, the circuit value problem, which we use extensively, is~\P-complete under \IPL reductions~\cite{miltersen_complexity_1994}.

Our definitions of dynamic problems and dynamic reductions given in~\Cref{sec:mcv_to_3core} behave differently from the definitions in the previous paragraph, since we base ourselves in modern definitions of dynamic graph problems~\cite{hanauer_recent_2022}. However, the following result formalizes the intuitive notion that these definitions are somewhat equivalent.

\begin{theorem} \label{res:mcv_bound_ipl}
   There exists no algorithm for the fully dynamic \MCV on a circuit of size~$n$ with~$\dynTime{\poly(n)}{\polylog(n)}{\polylog(n)}$ worst-case time, unless~${\P \subseteq \IPL}$. 
\end{theorem}
\begin{proof}
    Suppose there is an algorithm~\A for \MCV with complexity~$\dynTime{\poly(n)}{\polylog(n)}{\polylog(n)}$. Take any problem~$\pi \in \P$ with initial instance~$I$ with~$|I| = n$. Using the standard reduction~\cite{greenlaw_limits_1995}, we simulate the Turing machine for~$\pi$ using circuits, that is, we obtain a bounded monotone Boolean circuit~$C^\pi_I$ such that
    \begin{enumerate}[label = (\arabic*)]
        \item the circuit value of~$C^\pi_I$ is 1 if and only if~$\pi(I) = 1$,
        \item the size of~$C^\pi_I$ is~$\Oh(\poly(n))$ since~$\pi \in \P$, and
        \item the only part of~$C^\pi_I$ which depends on~$I$ is, for each~$i \in [n]$, a~0-gate when~${I_i = 0}$ and a~1-gate when~${I_i = 1}$.
    \end{enumerate}

    For each~$i \in [n]$, we add the opposite gates to~$C^\pi_I$ so that there exists a 0-gate~$g^0_i$ and a~1-gate~$g^1_i$. Note that~$g^{1-I_i}_i$ is not connected to any other gates. This concludes the description of the circuit~$C^\pi_I$.

    Whenever we receive a~1-bit change to~$I$, that is,~$i \in [n]$ such that we want to flip~$I_i$, if there exists a (single) outgoing wire from~$g^{I_i}_i$ to some gate~$f$, we delete that wire and insert the wire from~$g^{1-I_i}_i$ to $f$ instead.
    After the circuit modifications, note that the new circuit is exactly the circuit~$C^\pi_{I'}$ for instance~$I'$ which is~$I$ with the~$i$-th bit flipped, so properties~\mbox{(1)--(3)} are maintained and~$\pi(I') = 1$ if and only if the circuit value of~$C^\pi_{I'}$ is~1, which can be determined with a single query.

    Building and preprocessing~$C^\pi_I$ using~\A takes~$\poly(|C^\pi_I|) = \poly(n)$ time, and each~1-bit change results in at most one deletion, one insertion, and one query, which takes~$\polylog(|C^\pi_I|) = \polylog(n)$ time, so~$\pi \in \IPL$.
\end{proof}

In this case, we say that \MCV is \deff{\P-hard under \IPL reductions}, meaning that solving it in~$\dynTime{\poly(n)}{\polylog(n)}{\polylog(n)}$ time would imply that all problems in~\P can be solved efficiently in the incremental setting. This is similar to the~$\P/\NP$ and the~$\NC/\P$ relation.

Surprisingly, all~\P-complete problems (under \NC reductions\footnote{A reduction is \NC if it takes polylog parallel time using a polynomial number of processors~\cite{greenlaw_limits_1995}.}), which are widely believed not to have efficient parallel algorithms, are also~\P-complete under~\IPL reductions, with the exception of artificially constructed problems which have extensive amount of internal duplication~\cite{miltersen_complexity_1994}. In 1995, Greenlaw, Hoover, and Ruzzo~\cite{greenlaw_limits_1995} provided an extensive list of~\P-complete problems. Recently, Couto has made available online a compendium of~\P-complete problems~\cite{yan_p_completeness_compendium_2026}, including the list from~\cite{greenlaw_limits_1995} but also newer results published after 1995. Among other things, these sources may be used to determine which problems are known to be~\P-hard under~\IPL reductions, that is, hard to solve in the dynamic setting.

\subsection{OMv-based lower bounds} \label{sec:undirected_omv_lower}

In the \OMv problem for ``online'' $N \x N$ Boolean matrix multiplication, one is given a matrix~$A$, and then the columns of a matrix~$B$ one at a time. Matrix~$A$ may be preprocessed, and the task is to output the multiplication of each column vector of~$B$ by~$A$ before receiving the next column. In Boolean matrix multiplication, \BAND plays the role of multiplication and \BOR plays the role of addition. The OMv conjecture essentially states that there is no algorithm faster by a polynomial factor than the trivial~$\Theta(N^3)$ one. 

\begin{conjecture}[OMv~\cite{henzinger_unifying_2015}] \label{conj:omv}
    For any~$\epsilon > 0$, there is no algorithm for \OMv with polynomial preprocessing time and~$\Oh(N^{3-\epsilon})$ computation time, with an error probability of at most~$\frac{1}{3}$.
\end{conjecture}

We will work with a simpler variant of the same problem which makes reductions easier, but that is as hard as the original problem~\cite{henzinger_unifying_2015}.

\begin{problemStatement}{\OuMv}
    Given an~$N \times N$ Boolean matrix~$M$ and an online sequence of~$Q$ Boolean vectors pairs $(u^1, v^1),\allowbreak \ldots,\allowbreak (u^Q, v^Q)$, compute each~$(u^k)^\top Mv^k \in \{0,1\}$ before receiving~$(u^{k+1}, v^{k+1})$. Let~$p(N)$ be the preprocessing time before receiving the first vector pair and~$c(N, Q)$ be the computation time.
\end{problemStatement}

We say \OuMv has~\deff{size}~$N$ and~$Q$ queries. Notice that~$u^\top M v = 1$ if and only if there exist~$i, j \in [N]$ such that~$u_i = M_{ij} = v_j = 1$, and the pair~$(i, j)$ is called a~\deff{witness} in this case. On the left part of \Cref{fig:omv_to_mcv_example}, an instance of OuMv is shown with its corresponding witness.

\begin{figure}[h]
    \centering
    \documentclass[tikz, border=5mm]{standalone}
\usepackage[american]{circuitikz}
\usetikzlibrary{positioning, matrix, decorations.pathreplacing, calc}

\begin{document}

\begin{tikzpicture}[
    gate/.style={draw, minimum size=0.2cm, scale=0.8},
    inputnode/.style={draw, rectangle, scale=1, minimum size = 0.8cm, font=\Large},
    one_gate_color/.style={fill=green!30, no input leads, line width=2pt},
    or_gate_color/.style={fill=cyan!30, no input leads, no output leads},
    gateone/.style={very thick},
    label/.style={font=\bfseries},
    many/.style={loosely dotted, very thick, shorten >= 5pt, shorten <= 5pt},
    cell/.style={rectangle, draw=none, minimum size=0.4cm, font=\normalsize},
    highlight/.style={cell, fill=red!30},
    matrix_paren/.style={decorate, decoration={curveto, amplitude=10pt, raise=5pt}, thick},
    vector_paren/.style={decorate, decoration={curveto, amplitude=5pt, raise=3pt}, thick},
    every edge/.append style={in=180,out=0,looseness=0.6},
]

\begin{scope}[xshift=-5.25cm, yshift=-15pt]

    \matrix (M) [matrix of nodes, nodes={cell}, column sep=0.1cm, row sep=0.1cm]
    {
        1 & 1 & 0 & 1 \\
        0 & 0 & |[highlight]| 1 & 0 \\
        0 & 1 & 0 & 1 \\
        1 & 0 & 1 & 1 \\
    };
    \draw[matrix_paren, decoration={mirror}] (M-1-1.north west) -- (M-4-1.south west);
    \draw[matrix_paren] (M-1-4.north east) -- (M-4-4.south east);
    \matrix (u) [matrix of nodes, nodes={cell}, left=0.5cm of M, row sep=0.1cm]
    {
        0 \\
        |[highlight]|1 \\
        1 \\
        0 \\
    };
    \draw[vector_paren, decoration={mirror}] (u-1-1.north west) -- (u-4-1.south west);
    \draw[vector_paren] (u-1-1.north east) -- (u-4-1.south east);
    \matrix (v) [matrix of nodes, nodes={cell}, above=0.5cm of M, column sep=0.1cm]
    {
        1 & 0 & |[highlight]| 1 & 0 \\
    };
    \draw[vector_paren, yshift=1cm] (v-1-1.north west) -- (v-1-4.north east);
    \draw[vector_paren, decoration={mirror}] (v-1-1.south west) -- (v-1-4.south east);
    \node[font=\normalsize, above=0.1cm of u] {$u$};
    \node[font=\normalsize, left=0.1cm of v] {$v$};
    \node[font=\normalsize, above left=0.1cm and 0.1cm of M-1-1.north west] {$M$};
\end{scope}

\node[align=center] at (-2.25, -2.5) {\textbf{if $u_i = 1$}};
\node[align=center] at (0.5, -2.5) {\textbf{if $M_{ij} = 1$}};
\node[align=center] at (3, -2.5) {\textbf{if $v_j = 1$}};


\node[gate, inputnode, one_gate_color, label] (in1) at (-3, 0) {1};

\node[american or port, gate, or_gate_color, number inputs=3] (L1) at (-0.5, 1.5) {};
\node[american or port, gate, or_gate_color, number inputs=3, gateone] (L2) at (-0.5, 0.5) {};
\node[american or port, gate, or_gate_color, number inputs=3, gateone] (L3) at (-0.5, -0.5) {};
\node[american or port, gate, or_gate_color, number inputs=3] (L4) at (-0.5, -1.5) {};

\node[american or port, gate, or_gate_color, number inputs=2] (R1) at (2.5, 1.5) {};
\node[american or port, gate, or_gate_color, number inputs=2, gateone] (R2) at (2.5, 0.5) {};
\node[american or port, gate, or_gate_color, number inputs=2, gateone] (R3) at (2.5, -0.5) {};
\node[american or port, gate, or_gate_color, number inputs=3, gateone] (R4) at (2.5, -1.5) {};

\node[american or port, gate, or_gate_color, number inputs=2, gateone] (g_star) at (4.5, 0) {};

\node[above=2pt of L1.east] {$L_1$};
\node[above=2pt of L2.east] {$L_2$};
\node[above=2pt of L3.east] {$L_3$};
\node[above=2pt of L4.east] {$L_4$};
\node[above=2pt of R1.east] {$R_1$};
\node[above=2pt of R2.east] {$R_2$};
\node[above=2pt of R3.east] {$R_3$};
\node[above=2pt of R4.east] {$R_4$};
\node[right=-4pt of g_star] {$g^*$};

\path[very thick, red] ($(in1.south east) !2/3! (in1.north east)$) edge (L2.bin 2);
\path[very thick, red] (L2.bout) edge (R3.bin 1);
\path[very thick, red] (R3.bout) edge (g_star.bin 2);

\path ($(in1.south east) !1/3! (in1.north east)$) edge (L3.bin 2);
\path (L1.bout) edge (R1.bin 1);
\path (L1.bout) edge (R2.bin 1);
\path (L1.bout) edge (R4.bin 1);
\path (L3.bout) edge (R2.bin 2);
\path (L3.bout) edge (R4.bin 2);
\path (L4.bout) edge (R1.bin 2);
\path (L4.bout) edge (R3.bin 2);
\path (L4.bout) edge (R4.bin 3);
\path (R1.bout) edge (g_star.bin 1);

\end{tikzpicture}
\end{document}
    \caption{On the left, an example input to \OuMv, in \textcolor{red}{red} a witness that its answer is~1. On the right, the corresponding circuit~$C^{u,v}_M$ described in~\Cref{result:oumv_to_mcv}.}
    \Description{OuMv example and circuit construction. The left side shows a matrix and vectors with one witness pair, and the right side shows the corresponding circuit.}
    \label{fig:omv_to_mcv_example}
\end{figure}

\begin{proposition}[Theorem 2.7 in~\cite{henzinger_unifying_2015}] \label{prop:oumv}
   For any~$\epsilon > 0$, there is no algorithm for \OuMv of size $N$ and $Q$ queries with $p(N) = \poly(N)$ preprocessing time and $c(N, Q) = {\Oh(N^{2-\epsilon} Q + N^2 Q^{1-\epsilon})}$ computation time unless the \OMvConj is false.
\end{proposition}

We will use the \OMvConj to prove lower bounds for \MCV.
We start by arguing that, in a bounded circuit, it is possible to simulate gates with unbounded in- or out-degree, which will make our reductions simpler to understand.
Note the \BAND and \BOR functions are commutative and associative, so their inputs may be combined in any order.
Thus, we can replace any gate~$g$ of fixed maximum in-degree~$d$ with a leaves-to-root balanced\footnotemark{} binary tree made of gates of the same type as~$g$ and with at most~$d$ gates. Then we connect input wires to the leaf gates and the value of the root gate will always have the same value as~$g$. Similarly, for unbounded out-degree, we use a root-to-leaves binary tree of~\BOR-gates, and connect output wires to the leaves, which will all have the same value as~$g$. See~\Cref{fig:unbounded_to_bounded_mcv}. \footnotetext{In fact, the tree does not need to be balanced for this to work, but it does keep the depth of the circuit logarithmic.}

\begin{figure}[ht]
\centering
    \centering
    \documentclass[tikz, border=5mm]{standalone}
\usepackage{../common/tikzit}
\usepackage{../content/tikzit-styles}
\usepackage[american]{circuitikz}
\usetikzlibrary{positioning}

\begin{document}
\begin{tikzpicture}[xscale=0.25,yscale=0.35,
    gate/.style={draw, no input leads, no output leads,scale=0.5},
    and_gate_color/.style={fill=orange!30},
    or_gate_color/.style={fill=cyan!30, number inputs=3},
	every edge/.append style={in=180,out=0,looseness=0.6},
]
	\begin{pgfonlayer}{nodelayer}
		\node [american and port, gate, and_gate_color] (and_L3_N1) at (0, 0) {};
		\node [american and port, gate, and_gate_color] (and_L2_N1) at (-4, 2) {};
		\node [american and port, gate, and_gate_color] (and_L2_N2) at (-4, -2) {};
		\node [american and port, gate, and_gate_color] (and_L1_N3) at (-8, -1) {};
		\node [american and port, gate, and_gate_color] (and_L1_N2) at (-8, 1) {};
		\node [american and port, gate, and_gate_color] (and_L1_N1) at (-8, 3) {};
		\node [style=none,above=10pt of and_L3_N1] (17) {$g$};
		\node [american or port, gate, or_gate_color] (or_L4_N1) at (4, 1.5) {};
		\node [american or port, gate, or_gate_color] (or_L4_N2) at (4, -1.5) {};
		\node [american or port, gate, or_gate_color] (or_L5_N1) at (8, 0) {};
		\node [american or port, gate, or_gate_color] (or_L5_N2) at (8, -3) {};
		\node [style=none, left=1.5 of and_L2_N2.bin 2] (input_L1_N4_B1) {};
		\node [style=none, left=0.5 of and_L1_N3.bin 1] (input_L1_N3_B1) {};
		\node [style=none, left=0.5 of and_L1_N2.bin 1] (input_L1_N2_B1) {};
		\node [style=none, left=0.5 of and_L1_N1.bin 1] (input_L1_N1_B1) {};
		\node [style=none, left=0.5 of and_L1_N1.bin 2] (input_L1_N1_B2) {};
		\node [style=none, left=0.5 of and_L1_N2.bin 2] (input_L1_N2_B2) {};
		\node [style=none, left=0.5 of and_L1_N3.bin 2] (input_L1_N3_B2) {};
		\node [style=none, above right=0.5 of or_L5_N1.bout] (output_L5_N1_O2) {};
		\node [style=none, above right=0.5 of or_L5_N2.bout] (output_L5_N2_O2) {};
		\node [style=none, below right=0.5 of or_L5_N1.bout] (output_L5_N1_O1) {};
		\node [style=none, below right=0.5 of or_L5_N2.bout] (output_L5_N2_O1) {};
		\node [style=none, right=0.5 of or_L4_N1.bout] (output_L4_N1_O1) {};
	\end{pgfonlayer}
	\begin{pgfonlayer}{edgelayer}
		\path (and_L2_N1.bout) edge (and_L3_N1.bin 1);
		\path (and_L2_N2.bout) edge (and_L3_N1.bin 2);

		\path (and_L1_N1.bout) edge (and_L2_N1.bin 1);
		\path (and_L1_N2.bout) edge (and_L2_N1.bin 2);

		\path (and_L1_N3.bout) edge (and_L2_N2.bin 1);

		\path (and_L3_N1.bout) edge (or_L4_N1.bin 2);
		\path (and_L3_N1.bout) edge (or_L4_N2.bin 2);

		\path (or_L4_N2.bout) edge (or_L5_N1.bin 2);
		\path (or_L4_N2.bout) edge (or_L5_N2.bin 2);

		\path [style=dashed] (input_L1_N1_B1.center) edge (and_L1_N1.bin 1);
		\path [style=dashed] (input_L1_N1_B2.center) edge (and_L1_N1.bin 2);

		\path [style=dashed] (input_L1_N2_B1.center) edge (and_L1_N2.bin 1);
		\path [style=dashed] (input_L1_N2_B2.center) edge (and_L1_N2.bin 2);

		\path [style=dashed] (input_L1_N3_B1.center) edge (and_L1_N3.bin 1);
		\path [style=dashed] (input_L1_N3_B2.center) edge (and_L1_N3.bin 2);

		\path [style=dashed] (input_L1_N4_B1.center) edge (and_L2_N2.bin 2);

		\path [style=dashed] (or_L4_N1.bout) edge (output_L4_N1_O1.center);
		\path [style=dashed] (or_L5_N1.bout) edge (output_L5_N1_O1.center);
		\path [style=dashed] (or_L5_N1.bout) edge (output_L5_N1_O2.center);
		\path [style=dashed] (or_L5_N2.bout) edge (output_L5_N2_O1.center);
		\path [style=dashed] (or_L5_N2.bout) edge (output_L5_N2_O2.center);
	\end{pgfonlayer}
\end{tikzpicture}
\end{document}
    \caption{Simulating an \BAND-gate with in-degree 7 and out-degree 5.}
    \Description{6 AND gates and 4 OR gates, with wires A1-A4, A2-A4, A3-A5, A4-A6, A5-A6, A6-O1, A6-O2, O2-O3, O2-O4. Unused AND inputs are input wires, and unused OR outputs are output wires.}
    \label{fig:unbounded_to_bounded_mcv}
\end{figure}%

\begin{proposition} \label{result:oumv_to_mcv}
    If fully (partially) dynamic \MCV on a circuit of size~$n$ can be solved in $\dynTime{p(n)}{u(n)}{q(n)}$ amortized (worst-case, resp.) time, then \OuMv of size~$N$ and~$Q$ queries can be solved in~$\Oh(p(N^2))$ preprocessing time and~${\Oh((NQ + N^2) u(N^2) + (Q+N^2) q(N^2))}$ computation time.
\end{proposition}

\begin{proof}
Given an~$N \x N$ Boolean matrix~$M$, we will describe a circuit~$C_M$ with output gate~$g^*$ and, when given a vector pair~$(u,v)$, we will show how to modify~$C_M$ to build~$C^{u,v}_M$ so that the circuit value of~$C^{u,v}_M$ is 1 if and only if~$u^\top M v = 1$. 

Build $C_M$ with~a 1-gate~$\1$ and \BOR-gates~$L_1, \ldots,L_N, R_1, \ldots, R_N$, and~$g^*$. Add wire~$L_i R_j$ for each~$i, j \in [N]$ such that~$M_{ij}=1$. See~\Cref{fig:omv_to_mcv}.
To build~$C^{u,v}_M$ from~$C_M$, $u$, and~$v$, add wires~$\1 L_i$ for each~$i \in [N]$ with~$u_i = 1$ and wires~$R_j g^*$ for each~$j \in [N]$ with~$v_j = 1$. This concludes the description of~$C_M$ and~$C^{u,v}_M$, and an example is shown in~\Cref{fig:omv_to_mcv_example}.
Note that~$C_M$ and~$C^{u,v}_M$ can be simulated with~$\Oh(N^2)$ gates of bounded in- and out-degree, and~$C^{u,v}_M$ adds at most~$2N$ wires to~$C_M$.

\begin{figure}[h]
\centering
    \documentclass[tikz, border=5mm]{standalone}
\usepackage{../common/tikzit}
\usepackage{../content/tikzit-styles}
\usepackage[american]{circuitikz}
\usetikzlibrary{positioning,calc}
\usepackage{dsfont}

\newcommand{\1}{\mathds{1}}

\begin{document}
\begin{tikzpicture}[xscale=0.8,
    gate/.style={draw, no input leads, no output leads, scale=0.75},
    inputnode/.style={draw, rectangle},
    zero_gate_color/.style={fill=gray!30},
    one_gate_color/.style={fill=green!30},
    or_gate_color/.style={fill=cyan!30, number inputs=3},
    dashed-directed/.style={dashed},
	every edge/.append style={in=180,out=0,looseness=0.6},
]
	\begin{pgfonlayer}{nodelayer}
		\node [american or port, gate, or_gate_color] (or_L2_N1) at (-0.5, 1.25) {};
		\node [inputnode, one_gate_color] (input_1) at (-4, 0) {1};
		\node [american or port, gate, or_gate_color] (or_L2_N2) at (-0.5, 0) {};
		\node [american or port, gate, or_gate_color] (or_L2_N3) at (-0.5, -2) {};
		\node [american or port, gate, or_gate_color] (or_L3_N1) at (3, 1.25) {};
		\node [american or port, gate, or_gate_color] (or_L3_N2) at (3, -0.5) {};
		\node [american or port, gate, or_gate_color] (or_L3_N3) at (3, -2) {};
		\node [american or port, gate, or_gate_color] (or_L4_N1) at (6, -0.5) {};
		\node [style=none, above=5pt of or_L4_N1] (12) {$g^*$};
		\node [style=none, above=5pt of input_1] (14) {$\1$};
		\node [style=none, right=3pt of or_L2_N1.bout] (18) {$L_1$};
		\node [style=none, above right=3pt of or_L2_N2.bout] (19) {$L_i$};
		\node [style=none, right=3pt of or_L2_N3.bout] (20) {$L_N$};
		\node [style=none, left=4pt of or_L3_N1.bin 2] (21) {$R_1$};
		\node [style=none, above left=4pt and 7pt of or_L3_N2.bin 2] (22) {$R_j$};
		\node [style=none, left=4pt of or_L3_N3.bin 2] (23) {$R_N$};
		\node [style=none, above right=3pt and 1pt of or_L3_N2.bout] (24) {if $v_j = 1$};
	\end{pgfonlayer}
	\begin{pgfonlayer}{edgelayer}
		\draw [style=cdots arrow] (or_L2_N1.south) to (or_L2_N2.north);
		\draw [style=cdots arrow] (or_L2_N2.south) to (or_L2_N3.north);
		\draw [style=cdots arrow] (or_L3_N1.south) to (or_L3_N2.north);
		\draw [style=cdots arrow] (or_L3_N2.south) to (or_L3_N3.north);
		
		\path [style=dashed-directed] (or_L2_N2.bout) edge node[midway,below=3pt] {if $M_{ij} = 1$} (or_L3_N2.bin 2);
		\path [style=dashed-directed] (or_L3_N1.bout) edge (or_L4_N1.bin 1);
		\path [style=dashed-directed] (or_L3_N2.bout) edge (or_L4_N1.bin 2);
		\path [style=dashed-directed] (or_L3_N3.bout) edge (or_L4_N1.bin 3);
		\path [style=dashed-directed] ($(input_1.north east) !1/4! (input_1.south east)$) edge[looseness=0.3] (or_L2_N1.bin 2);
		\path [style=dashed-directed] (input_1.east) edge node[midway,below] {if $u_i = 1$} (or_L2_N2.bin 2);
		\path [style=dashed-directed] ($(input_1.north east) !3/4! (input_1.south east)$) edge[looseness=0.3] (or_L2_N3.bin 2);
	\end{pgfonlayer}
\end{tikzpicture}
\end{document}
    \captionsetup{textformat=simple}
\caption{
Sketch of~$C^{u,v}_M$, omitting binary trees that simulate unbounded in- and out-degree.}
\Description{Sketch of the circuit used in the OuMv reduction, as described in the proof.}
\label{fig:omv_to_mcv}
\end{figure}

The circuit value of~$C^{u,v}_M$ is 1 if and only if a directed path exists from the single 1-gate~$\1$ to~$g^*$, as other gates are~\BOR-gates. This occurs if and only if there are~$i,j \in [N]$ for which wires~$\1 L_i$, $L_i R_j$, and~$R_j g^*$ exist, i.e.,~$u_i = M_{ij} = v_j = 1$, so~$(i,j)$ is a witness to $u^\top M v=1$.

Now let us describe the algorithm for~\OuMv based on \MCV, starting first by the amortized fully dynamic case. Build~$C_M$, an instance to~\MCV, and preprocess it in~$\Oh(p(N^2))$ time. Then, after receiving each vector pair~$(u,v)$, build~$C^{u,v}_M$ by adding at most~$2N$ wires, perform a single query to decide whether~$u^\top M v = 1$, and finally delete at most~$2N$ wires to recover~$C_M$. Do this for each of the~$Q$ vector pairs, for a total of~$\Oh(NQ)$ updates and~$Q$ queries, which by our definition of amortization takes~$\Oh((NQ + N^2)u(N^2) + (Q+N^2)q(N^2))$ computation time.

For the incremental case of \MCV, proceed the same way but, when building~$C^{u,v}_M$ for a vector pair~$(u,v)$, record all changes done to the data structures of the algorithm, which sum up to~$\Oh(N u(N^2) + q(N^2))$ changes, and then undo them (rollback) to restore~$C_M$.
For the decremental case, before preprocessing, add wires~$\1 L_i$ and~$R_j g^*$ for each~$i,j \in [N]$ to~$C_M$. Then, to build~$C^{u,v}_M$, delete all wires~$\1 L_i$ for each~$i \in [N]$ such that~$u_i=0$, and~$R_j g^*$ for each~$j \in [N]$ such that~$v_j = 0$. Finally, use the same rollback technique to restore the initial circuit. Both partially dynamic cases take~$\Oh(N u(N^2) + q(N^2))$ per vector pair, and thus take~$\Oh(NQ u(N^2) + Qq(N^2))$ computation time.

In the partially dynamic cases, note the rewinding forbids an amortized analysis over the whole sequence of insertions, which is why we require a worst-case time bound for those cases, otherwise building~$C^{u,v}_M$ could take~$\Theta(N^2 u(N^2))$~time.
\end{proof}

By using this reduction, we are able to obtain lower bounds on \MCV using the \OMvConj.

\begin{theorem} \label{result:mcv_bound_omv}
   For any~$\epsilon > 0$, there exists no algorithm for the fully (partially) dynamic \MCV on a circuit of size~$n$ with~$\dynTime{\poly(n)}{\Oh(n^{\frac{1}{2}-\epsilon})}{\Oh(n^{1-\epsilon})}$ amortized (worst-case, resp.) time, unless the \OMvConj is false. 
\end{theorem}
\begin{proof}
    Suppose there is an algorithm~\A for \MCV with complexity~$\dynTime{\poly(n)}{\Oh(n^{\frac{1}{2}-\epsilon})}{\Oh(n^{1-\epsilon})}$.
    Choose ${Q = N^2}$. By Proposition~\ref{result:oumv_to_mcv}, we can solve \OuMv with size~$N$ and~$Q$ queries with~$\Oh(\poly(N^2)) = \poly(N)$ preprocessing time and
    \(
    \Oh((NQ + N^2) u(N^2) + (Q + N^2) q(N^2))
    = \Oh(NQN^{1-2\epsilon} + QN^{2-2\epsilon}) = \Oh(N^{2-\epsilon}Q),
    \)
    computation time which, by \Cref{prop:oumv}, falsifies the \OMvConj.
\end{proof}


\subsection{SETH-based lower bounds} \label{sec:undirected_seth_lower}

In the \kSAT problem, we are given a~$k$-CNF formula on~$N$ variables and~$M$ clauses, that is, the formula is a conjunction where all clauses have exactly~$k$ literals. The task is to determine whether the formula is satisfiable by some assignment for its variables.
The strong exponential time hypothesis (SETH) essentially states that there is no algorithm for SAT faster than the trivial~$\Theta^*(2^N)$ one\footnote{Recall that $\Ohp$ omits polynomial factors, that is,~$\Ohp(f(n)) = \Oh(f(n) \poly(n))$. The same holds for $\Theta^*$.} by an exponential factor:

\begin{conjecture}[SETH~\cite{impagliazzo_complexity_2001}] \label{conj:seth}
    For any~$\epsilon > 0$, there exists a~$k$ such that no algorithm for~\kSAT on $N$ variables with~$\Ohp(2^{(1-\epsilon)N})$ running time exists.
\end{conjecture}

By using the \SETH, we can show stronger conditional lower bounds, indicating that there is essentially no algorithm for~\MCV better than the trivial one, in which we just apply the static algorithm to the current circuit on every query. 

\begin{proposition} \label{result:ksat_to_mcv}
    For any positive~$\delta < \frac12$, if the fully (partially) dynamic \MCV on a circuit of size~$n$ can be solved in~$\dynTime{p(n)}{u(n)}{q(n)}$ amortized (worst-case, resp.) time, then \kSAT on $N$ variables with~$\Oh(N)$ clauses can be solved in~$\Oh(p(N') + 2^{(1-\delta)N} \left(N u(N') + q(N')\right))$ time, where $N' = \Oh(2^{\delta N} N)$.
\end{proposition}

\begin{proof}
    Given a CNF formula~$F$ on a set~$V$ of~$N$ variables and a set~$C$ of~$\Oh(N)$ clauses, choose any~$U \subseteq V$ of size~$\ceil{\delta N}$. We will describe a circuit~$C^U_F$ with output gate~$g^*$ and how to build a circuit~$C^U_F(v)$ for any partial assignment~$v \in \Bzo^{V \setminus U}$ to~$V \setminus U$, so that the circuit value of~$C^U_F(v)$ is 0 if and only if there exists a partial assignment~$u \in \Bzo^U$ to~$U$ such that the assignment~$u + v$ to~$V$ satisfies~$F$. 

    Build~$C^U_F$ as in \Cref{fig:seth_to_mcv}, with a 1-gate~$\1$, an \BAND-gate~$g^*$, and \BOR-gates~$L_c$, for each clause~$c \in C$, and~$R_u$, for each partial assignment~$u \in \Bzo^U$. Add wire~$L_c R_u$ for each~$c \in C$ and~$u \in \Bzo^U$ such that~$u$ does \emph{not} satisfy clause~$c$, and wire~$R_u g^*$ for each~$u \in \Bzo^U$. To build~$C^U_F(v)$ from~$C^U_F$ and~$v$, add wires~$\1 L_c$ for each~$c \in C$ such that~$v$ does \emph{not} satisfy~$c$. This concludes the description of~$C^U_F$ and~$C^U_F(v)$, and an example is shown in~\Cref{fig:seth_mcv_example}. Note that both circuits can be simulated with~$N'=\Oh(2^{\delta N} N)$ gates of bounded in- and out-degree, and~$C^U_F(v)$ adds at most~$|C| = \Oh(N)$ wires to~$C^U_F$.

\begin{figure}[h]
    \centering
    \documentclass[tikz, border=5mm]{standalone}
\usepackage{../common/tikzit}
\usepackage{../content/tikzit-styles}
\usepackage[american]{circuitikz}
\usetikzlibrary{positioning,fit}
\usepackage{dsfont}

\newcommand\Bzo{\{0,1\}}
\newcommand\1{\mathds{1}}

\begin{document}
\begin{tikzpicture}[xscale=0.65,yscale=0.6,
    gate/.style={draw, no input leads, no output leads, scale=0.6},
    inputnode/.style={draw, rectangle},
    one_gate_color/.style={fill=green!30},
    and_gate_color/.style={fill=orange!30},
    or_gate_color/.style={fill=cyan!30, number inputs=3},
    directed/.style={},
    dashed-directed/.style={dashed},
	every edge/.append style={in=180,out=0,looseness=0.6},
	region/.style={rounded corners=10pt, inner sep=5pt, fill=gray!15, draw=none},
]
	\begin{pgfonlayer}{nodelayer}
		\node [american or port, gate, or_gate_color] (or_L2_N1) at (0, 3) {};
		\node [american or port, gate, or_gate_color] (or_L2_N2) at (0, 0.5) {};
		\node [american or port, gate, or_gate_color] (or_L2_N3) at (0, -2) {};
		\node [inputnode, one_gate_color] (input_1) at (-5, 0.5) {1};
		\node [american or port, gate, or_gate_color] (or_L3_N1) at (5, 5) {};
		\node [american or port, gate, or_gate_color] (or_L3_N2) at (5, 0.5) {};
		\node [american or port, gate, or_gate_color] (or_L3_N3) at (5, -4) {};
        \ctikzset{tripoles/american and port/height=1.2}
		\node [american and port, gate, and_gate_color, number inputs=3] (and_L4_N1) at (7.5, 0.5) {};
		\node [style=none, above=15pt of or_L2_N1] (left_layer_label) {$C$};
		\node [style=none, above=15pt of or_L3_N1] (right_layer_label) {$\Bzo^U$};
		\node [style=none, above right=3pt and 1pt of or_L2_N2.bout] (label_Lc) {$L_c$};
		\node [style=none, above right=3pt and 0pt of or_L3_N2.bout] (label_Ru) {$R_u$};
		\node [style=none] (stage_v_label) at (-4, 5) {Stage $v \in \Bzo^{V \setminus U}$};
		\node [style=none, above=5pt of and_L4_N1] (output_g_star_label) {$g^*$};
		\node [style=none, above=5pt of input_1] (input_1_label) {$\1$};
	\end{pgfonlayer}
	\begin{pgfonlayer}{edgelayer}
        \node [region, fit=(or_L2_N1) (or_L2_N3)] (left_box) {};
        \node [region, fit=(or_L3_N1) (or_L3_N3)] (right_box) {};
		\draw [style=cdots arrow] (or_L2_N1.south) to (or_L2_N2.north);
		\draw [style=cdots arrow] (or_L2_N3.north) to (or_L2_N2.south);
		\draw [style=cdots arrow] (or_L3_N3.north) to (or_L3_N2.south);
		\draw [style=cdots arrow] (or_L3_N2.north) to (or_L3_N1.south);
		
		\path [style=dashed-directed] (or_L2_N2.bout) edge node[above, midway, font=\small] {if $u$ does} node[below, midway, font=\small] {not satisfy $c$} (or_L3_N2.bin 2);
		\path [style=directed] (or_L3_N1.bout) edge (and_L4_N1.bin 1);
		\path [style=directed] (or_L3_N2.bout) edge (and_L4_N1.bin 2);
		\path [style=directed] (or_L3_N3.bout) edge (and_L4_N1.bin 3);

		\path [style=dashed-directed] ($(input_1.north east) !1/4! (input_1.south east)$) edge[out=30] (or_L2_N1.bin 2);
		\path [style=dashed-directed] (input_1.east) edge node[above, midway, font=\small] {if $v$ does} node[below, midway, font=\small] {not satisfy $c$} (or_L2_N2.bin 2);
		\path [style=dashed-directed] ($(input_1.north east) !3/4! (input_1.south east)$) edge[out=-30] (or_L2_N3.bin 2);
	\end{pgfonlayer}
\end{tikzpicture}
\end{document}
    \caption{Reduction from \kSAT to \MCV on stage~$v$, omitting binary trees used to simulate unbounded in- and out-degree.}
    \Description{Sketch of the circuit used in the k-SAT reduction, as described in the proof.}
    \label{fig:seth_to_mcv}
\end{figure}%
\begin{figure}[h]
    \centering
    \documentclass[tikz, border=5mm]{standalone}
\usepackage[american]{circuitikz}
\usepackage{expl3}

\usetikzlibrary{positioning, matrix, decorations.pathreplacing, calc}


\begin{document}
\begin{tikzpicture}[
    gate/.style={draw, minimum size=0.2cm, scale=0.8},
    inputnode/.style={draw, rectangle, scale=1, minimum size = 0.8cm, font=\Large},
    one_gate_color/.style={fill=green!30, no input leads, line width=2pt},
    or_gate_color/.style={fill=cyan!30, no input leads, no output leads},
    and_gate_color/.style={fill=orange!30, no input leads, no output leads},
    gateone/.style={very thick},
    label/.style={font=\bfseries},
    many/.style={loosely dotted, very thick, shorten >= 5pt, shorten <= 5pt},
    cell/.style={rectangle, draw=none, minimum size=0.4cm, font=\normalsize},
    highlight/.style={cell, fill=red!30},
    matrix_paren/.style={decorate, decoration={curveto, amplitude=10pt, raise=5pt}, thick},
    vector_paren/.style={decorate, decoration={curveto, amplitude=5pt, raise=3pt}, thick},
    every edge/.append style={in=180,out=0,looseness=0.6},
    every node/.append style={font=\footnotesize},
]


\node[gate, inputnode, one_gate_color, label] (in1) at (-3, 0) {1};

\ctikzset{tripoles/american or port/height=0.5}
\renewcommand{\d}{1.1}
\node[american or port, gate, or_gate_color, number inputs=3] (C1) at ($(-0.5,  3*\d/2)$) {};
\node[american or port, gate, or_gate_color, number inputs=3, gateone] (C2) at ($(-0.5,  \d/2)$) {};
\node[american or port, gate, or_gate_color, number inputs=3, gateone] (C3) at ($(-0.5,  -\d/2)$) {};
\node[american or port, gate, or_gate_color, number inputs=3, gateone] (C4) at ($(-0.5,  -3*\d/2)$) {};

\ctikzset{tripoles/american or port/height=0.5}
\renewcommand{\d}{0.75}
\node[american or port, gate, or_gate_color, number inputs=3] (V000) at ($(2.5,  7*\d/2)$)   {000};
\node[american or port, gate, or_gate_color, number inputs=2, gateone] (V001) at ($(V000) - (0, \d)$)   {001};
\node[american or port, gate, or_gate_color, number inputs=2, gateone] (V010) at ($(V001) - (0, \d)$)   {010};
\node[american or port, gate, or_gate_color, number inputs=2, gateone] (V011) at ($(V010) - (0, \d)$)   {011};
\node[american or port, gate, or_gate_color, number inputs=3, gateone] (V100) at ($(V011) - (0, \d)$)   {100};
\node[american or port, gate, or_gate_color, number inputs=3, gateone] (V101) at ($(V100) - (0, \d)$)   {101};
\node[american or port, gate, or_gate_color, number inputs=2, gateone] (V110) at ($(V101) - (0, \d)$)   {110};
\node[american or port, gate, or_gate_color, number inputs=3, gateone] (V111) at ($(V110) - (0, \d)$)   {111};

\ctikzset{tripoles/american and port/height=2}
\node[american and port, gate, and_gate_color, number inputs=8] (g_star) at (4.5, 0) {};
\ctikzset{tripoles/american and port/height=1}
\ctikzset{tripoles/american or port/height=1}

\node[above=0pt of C1] {$x_1 \vee \overline{x_2} \vee x_4$};
\node[above=0pt of C2] {$\overline{x_1} \vee x_3 \vee \overline{x_5}$};
\node[above=0pt of C3] {$x_2 \vee \overline{x_4} \vee x_6$};
\node[above=0pt of C4] {$\overline{x_3} \vee x_5 \vee \overline{x_6}$};
\node[above=2pt of V000] {\small $(x_4, x_5, x_6) \in \{0,1\}^U$};
\node[right=-4pt of g_star] {\small $g^*$};


\path ($(in1.south east) !3/4! (in1.north east)$) edge (C2.bin 2);
\path ($(in1.south east) !2/4! (in1.north east)$) edge (C3.bin 2);
\draw ($(in1.south east) !1/4! (in1.north east)$) to[out=0,in=90,looseness=0.7] ++(0.3, -0.8) to[out=-90,in=180,looseness=0.7] (C4.bin 2);

\path (C1.bout) edge (V000.bin 2);
\path (C1.bout) edge (V001.bin 1);
\path (C1.bout) edge (V010.bin 1);
\path (C1.bout) edge (V011.bin 1);
\path (C2.bout) edge (V010.bin 2);
\path (C2.bout) edge (V011.bin 2);
\path (C2.bout) edge (V110.bin 1);
\path (C2.bout) edge (V111.bin 2);
\path (C3.bout) edge (V100.bin 2);
\path (C3.bout) edge (V110.bin 2);
\path (C4.bout) edge (V001.bin 2);
\path (C4.bout) edge (V101.bin 2);

\path (V000.bout) edge (g_star.bin 1);
\path (V001.bout) edge (g_star.bin 2);
\path (V010.bout) edge (g_star.bin 3);
\path (V011.bout) edge (g_star.bin 4);
\path (V100.bout) edge (g_star.bin 5);
\path (V101.bout) edge (g_star.bin 6);
\path (V110.bout) edge (g_star.bin 7);
\path (V111.bout) edge (g_star.bin 8);

\end{tikzpicture}
\end{document}
    \caption{For~$F=(x_1 \vee \overline{x_2} \vee x_4) (\overline{x_1} \vee x_3 \vee \overline{x_5})(x_2 \vee \overline{x_4} \vee x_6)(\overline{x_3} \vee x_5 \vee \overline{x_6})$, instance of~3-SAT, the figure shows the corresponding circuit~$C^U_F(v)$ described in~\Cref{result:oumv_to_mcv}, where~$U = \{x_4,x_5,x_6\}$ and~$v$ is the partial assignment~$(x_1,x_2,x_3) = (0,1,0)$. Since the circuit value is~0, we can determine that~$(0,1,0,0,0,0)$ is a satisfying assignment.}
    \Description{Example SAT instance and the corresponding circuit for one partial assignment.}
    \label{fig:seth_mcv_example}
\end{figure}%

    Note that in~$C^U_F(v)$ the value of gate~$R_u$ for any~$u \in \Bzo^U$ is 1 if and only if there is a path from~$\1$ to~$R_u$, since~$\1$ is the only~1-gate and all~$L$ and~$R$ gates are~\BOR-gates. Thus the circuit value of~$C^U_F(v)$ is 0 if and only if~$R_u$ is unreachable from~$\1$ for some~$u \in \Bzo^U$. In that case, for each~$c \in C$, we have that~$\1 L_c$ or~$L_c R_u$ does not exist. By our construction, that means~$u$ or~$v$ satisfies each~$c$, and thus the assignment~$u+v$ satisfies~$F$.

    Now let us describe the algorithm for~\kSAT based on \MCV, starting first by the amortized fully dynamic case. Build~$C^U_F$ and preprocess it in~$p(N')$ time. Then proceed in~$2^{N-\ceil{\delta N}}$ stages, one for each partial assignment~$v \in \Bzo^{V \setminus U}$. On each stage~$v$, build~$C^U_F(v)$ by adding~$\Oh(N)$ wires to~$C^U_F$, perform a single query and if the circuit value of~$C^U_F$ is~0 output that~$F$ is satisfiable, otherwise delete~$\Oh(N)$ wires to recover~$C^U_F$. We do a total of~$\Oh(2^{(1-\delta)N}N)$ updates and~$\Oh(2^{(1-\delta)N})$ queries, which by our definition of amortization takes
    \[
    \Oh\left(\left(2^{(1-\delta)N} N + N'\right) u(N') + \left(2^{(1-\delta)N} + N'\right) q(N')\right)
    = \Oh\left(2^{(1-\delta)N} \left(N u(N') + q(N')\right)\right)
    \]
    time since~$\delta < \frac12$, finishing the proof for this case.

    For the incremental case of \MCV, build~$C^U_F$ normally and, for each $v$, build~$C^U_F(v)$ keeping track of the changes, as in the proof of Proposition~\ref{result:oumv_to_mcv}, then query the circuit value of $C^U_F(v)$ and use rollback to recover~$C^U_F$. In the decremental case of \MCV, before preprocessing, add wire~$\1 L_c$ to~$C^U_F$ for each~$c\in C$ and, to obtain circuit~$C^U_F(v)$, remove wires instead, then rollback to restore the initial circuit. Both partially dynamic cases take~${\Oh(N u(N') + q(N'))}$ time per stage, which sum up to~$p(N') + \Oh(2^{(1-\delta)N} \left(N u(N') + q(N')\right))$ time.
\end{proof}

Based on the Sparsification Lemma~\cite{impagliazzo_which_2001}, it can be shown that, to falsify the \SETH, it suffices to consider \kSAT instances on $N$ variables with~$\Oh(N)$ clauses~\cite[Proposition~1]{abboud_popular_2014_full}, as in 
\Cref{result:ksat_to_mcv}.
So, now, with an appropriate choice of~$\delta$, we derive lower bounds on \MCV using the \SETH.

\begin{theorem} \label{res:mcv_bound_seth}
    For any~$\epsilon > 0$, there exists no algorithm for the fully (partially) dynamic \MCV on a circuit of size~$n$ with~$\dynTime{\poly(n)}{\Oh(n^{1-\epsilon})}{\Oh(n^{1-\epsilon})}$ amortized (worst-case, resp.) time, unless the \SETH is false.
\end{theorem}

\begin{proof}
    For any positive~$\eps$ and~$t$, suppose there exists an algorithm as in the statement, that runs in time $\dynOh{n^t}{n^{1-\epsilon}}{n^{1-\epsilon}}$ (amortized, in the fully dynamic case). Take~${\delta = \min\{\frac{1-\epsilon}{t}, 0.49\}}$. By Proposition~\ref{result:ksat_to_mcv}, \kSAT on~$N$ variables with~$\Oh(N)$ clauses can be solved in
   \(\Oh((2^{\delta N}N)^t + 2^{(1-\delta) N} N (2^{\delta N}N)^{(1 - \epsilon)})\) \(= \Ohp(2^{(1-\epsilon)N} + 2^{(1-\epsilon \delta) N})\) time.
    This contradicts the \SETH with~$\epsilon \delta$ in the role of its~$\epsilon$. The partially dynamic case follows in the same manner.
\end{proof}

\subsection{Unconditional lower bounds} \label{sec:unconditional_core}

\newcommand{\CPROBE}[1]{\textup{\textsc{cprobe[}}{\small#1}\textup{\textsc{]}}}
\newcommand{\Bxor}{\oplus}

Our previous lower bounds are based on famous conjectures such as \OMvConj and \SETH, and thus are \emph{conditional}. Unconditional lower bounds are harder to obtain, and the strongest are typically derived within the \deff{cell probe} model of computation~\cite{miltersen_complexity_1994,patrascu_lower_2008,yao_should_1981}, denoted as \CPROBE{$B$} for a number $B$. In this model, memory is divided into cells, each storing~$\Oh(B)$ bits of data. The cost of an operation is measured by the number of memory cells accessed, or \emph{probed}. Computation is free of charge, and the addresses of the cells to probe can be determined as an arbitrary function of the query and the contents of previously probed cells. This abstraction allows for bounds that inherently apply to more realistic models like the RAM model~\cite{cormen_ram_2022}, especially when considering $B = \lg N$ for a problem of size $N$, which is common practice to allow addressing all memory.
The following dynamic array problem is commonly used to get unconditional lower bounds for dynamic problems~\cite{miltersen_complexity_1994,patrascu_lower_2008}. We use $\Bxor$ to denote the exclusive OR operation, that is, for any~$a,b \in \Bzo$,~$a \Bxor b$ is 1 if exactly one of~$a$ and~$b$ are 1.

\newcommand{\DynXor}{\problem{DynXor}}
\renewcommand{\Bzo}{\{0,1\}}

\begin{problemStatement}{\DynXor}
    Given a vector~$(x_1, \ldots, x_n) \in \Bzo^n$, process a series of updates to values in the vector, and queries: Given~$i$, is~$x_1 \Bxor x_2 \Bxor \cdots \Bxor x_i$ equal to~$1$?
\end{problemStatement}

\begin{lemma}[\cite{fredman_cell_1989}] \label{res:dynxor_lower_bound}
    There exists no algorithm for \DynXor with~$o(\frac{\lg n}{\lg \lg n})$ time for updates and queries in the \CPROBE{$\lg n$} model.
\end{lemma}

We can then obtain lower bounds to \MCV via \DynXor.

\newcommand{\n}[1]{\overline{#1}}

\begin{proposition} \label{res:dynxor_to_mcv}
    If fully dynamic \MCV on a circuit of size~$N$ can be solved in~$\dynTime{p(N)}{u(N)}{q(N)}$ time, then \DynXor on a vector of size~$n$ can be solved with~$\Oh(u(n))$ time per update and~$\Oh(u(n)+q(n))$ time per query.
\end{proposition}
\begin{proof}
    Given a vector~$X \coloneqq (x_1, \ldots, x_n) \in \Bzo^n$, we will describe a circuit~$C_X$ with gates~$g_1, \ldots, g_n$ such that, for each~$i \in [n]$,~$x_1 \Bxor \cdots \Bxor x_i$ is equal to the value of~$g_i$ in the circuit.

    Create a~0-gate~$g_0$ and a~1-gate~$\n{g_0}$. Each~$x_i$ corresponds to two~$\BOR$-gates~$g_i$ and~$\n{g_i}$. If~$x_i = 0$, then add wires~$g_{i-1}g_i$ and~$\n{g_{i-1}}\,\n{g_i}$, and if~$x_i = 1$, then add wires~$\n{g_{i-1}}g_i$ and~$g_{i-1}\n{g_i}$, like in \Cref{fig:xor_wires}. Additionally, create the output $\BOR$-gate~$g^*$.
    This concludes the description of the circuit~$C_X$. Changing some value~$x_i$ of vector~$X$ will be simulated by changing the wires between~$\{g_{i-1},\n{g_{i-1}}\}$ and~$\{g_i, \n{g_i}\}$, preserving the relation between~$X$ and~$C_X$.

\begin{figure*}[h]
\centering
\captionsetup[subfigure]{justification=centering}
\begin{subfigure}{0.49\textwidth}
    \centering
    \documentclass[tikz, border=5mm]{standalone}
\usepackage{../common/tikzit}
\usepackage{../content/tikzit-styles}
\usepackage[american]{circuitikz}
\usetikzlibrary{positioning}

\begin{document}
\begin{tikzpicture}[
    gate/.style={draw, no input leads, no output leads, number inputs=3,scale=0.5},
    or_gate_color/.style={fill=cyan!30},
]
	\begin{pgfonlayer}{nodelayer}
		\node [american or port, gate, or_gate_color] (gi1) at (-1, 0) {};
		\node [american or port, gate, or_gate_color] (ngi1) at (0, 0) {};
		\node [american or port, gate, or_gate_color] (gi) at (1, 0) {};
		\node [american or port, gate, or_gate_color] (ngi) at (2, 0) {};
		\node [style=none, below=5pt of gi1] (4) {$g_{i-1}$};
		\node [style=none, below=5pt of gi] (5) {$g_i$};
		\node [style=none, below=5pt of ngi] (7) {$\overline{g_i}$};
		\node [style=none, below=5pt of ngi1] (6) {$\overline{g_{i-1}}$};
		\node [style=none, left=8pt of gi1.bin 2] (input_gi1_B1) {};
		\node [style=none, left=8pt of ngi1.bin 2] (input_ngi1_B1) {};
		\node [style=none, right=8pt of gi.bout] (output_gi_O1) {};
		\node [style=none, right=8pt of ngi.bout] (output_ngi_O1) {};
	\end{pgfonlayer}
	\begin{pgfonlayer}{edgelayer}
		\path (gi1.bout) edge[in=120, out=60, looseness=1.1] (gi.bin 2);
		\path (ngi1.bout) edge[in=120, out=60, looseness=1.1] (ngi.bin 2);
		\path (input_gi1_B1.center) edge[dashed] (gi1.bin 2);
		\path (input_ngi1_B1.center) edge[dashed] (ngi1.bin 2);
		\path (gi.bout) edge[dashed] (output_gi_O1.center);
		\path (ngi.bout) edge[dashed] (output_ngi_O1.center);
	\end{pgfonlayer}
\end{tikzpicture}
\end{document}
    \captionsetup{textformat=simple}
    \subcaption{$x_i=0$}
    \Description{OR gates g i-1 connected to bi and g i-1 bar connected to g i bar.}
\end{subfigure}%
\begin{subfigure}{0.49\textwidth}
    \centering
    \documentclass[tikz, border=5mm]{standalone}
\usepackage{../common/tikzit}
\usepackage{../content/tikzit-styles}
\usepackage[american]{circuitikz}
\usetikzlibrary{positioning}

\begin{document}
\begin{tikzpicture}[
    gate/.style={draw, no input leads, no output leads, number inputs=3,scale=0.5},
    or_gate_color/.style={fill=cyan!30},
]
	\begin{pgfonlayer}{nodelayer}
		\node [american or port, gate, or_gate_color] (gi1) at (-1, 0) {};
		\node [american or port, gate, or_gate_color] (ngi1) at (0, 0) {};
		\node [american or port, gate, or_gate_color] (gi) at (1, 0) {};
		\node [american or port, gate, or_gate_color] (ngi) at (2, 0) {};
		\node [style=none, below=5pt of gi1] (4) {$g_{i-1}$};
		\node [style=none, below=5pt of gi] (5) {$g_i$};
		\node [style=none, below=5pt of ngi] (7) {$\overline{g_i}$};
		\node [style=none, below=5pt of ngi1] (6) {$\overline{g_{i-1}}$};
		\node [style=none, left=8pt of gi1.bin 2] (input_gi1_B1) {};
		\node [style=none, left=8pt of ngi1.bin 2] (input_ngi1_B1) {};
		\node [style=none, right=8pt of gi.bout] (output_gi_O1) {};
		\node [style=none, right=8pt of ngi.bout] (output_ngi_O1) {};
	\end{pgfonlayer}
	\begin{pgfonlayer}{edgelayer}
		\path (gi1.bout) edge[in=120, out=60, looseness=1.1] (ngi.bin 2);
		\path (ngi1.bout) edge (gi.bin 2);
		\path (input_gi1_B1.center) edge[dashed] (gi1.bin 2);
		\path (input_ngi1_B1.center) edge[dashed] (ngi1.bin 2);
		\path (gi.bout) edge[dashed] (output_gi_O1.center);
		\path (ngi.bout) edge[dashed] (output_ngi_O1.center);
	\end{pgfonlayer}
\end{tikzpicture}
\end{document}
    \captionsetup{textformat=simple}
    \subcaption{$x_i=1$}
    \Description{OR gate g i-1 bar connected to g i and g i-1 connected to g i bar.}
\end{subfigure}%
\caption{Wiring gates in~$C_X$ depending on the value of~$x_i$.}
\label{fig:xor_wires}
\end{figure*}

    Let us argue that, for each~$i \in \{0\}\cup [n]$,~$x_1 \Bxor \cdots \Bxor x_i$ is equal to the value of~$g_i$, and~$\n{g_i}$ has the opposite value of~$g_i$. This can be shown inductively, and it trivially holds when~$i = 0$ by our choice of gate type for~$g_0$ and~$\n{g_0}$. For~$i > 0$, if~$x_i=0$, then the wires~$g_{i-1}g_i$ and~$\n{g_{i-1}}\,\n{g_i}$ will preserve the values from~$i-1$, and if~$x_i = 1$, the wires will flip the values, since~$g_i$ and~$\n{g_i}$ are \BOR-gates with a single input.

    Therefore, to answer a query for whether~$x_1 \Bxor \cdots \Bxor x_i$ equals~$1$, we simply add a wire from~$g_i$ to~$g^*$ and query the circuit value of~$C_X$, then remove that wire. Circuit~$C_X$ has~$\Oh(n)$ gates and wires, and all its gates have in-degree and out-degree at most~2 at all times. Each update in~$X$ translates to~$\Oh(1)$ updates in~$C_X$, and a query in \DynXor translates to two updates and one query in \MCV, so the proof is finished.
\end{proof}

\begin{theorem} \label{res:mcv_bound_cellprobe}
    There exists no algorithm for fully dynamic \MCV on a circuit of size~$n$ with~$\dynTime{\poly(n)}{o(\frac{\lg n}{\lg \lg n})}{o(\frac{\lg n}{\lg \lg n})}$ time in the~\CPROBE{$\lg n$} model.
\end{theorem}
\begin{proof}
    By \Cref{res:dynxor_to_mcv}, this would imply an algorithm for \DynXor with~$o(\frac{\lg n}{\lg \lg n})$ update and query time, contradicting \Cref{res:dynxor_lower_bound}.
\end{proof}

An unconditional~$\lg^{\omega(1)} n$ lower bound, however, is far beyond the scope of current techniques, as even the most recent works provide only~$\Omega\left(\frac{\lg^{\nicefrac{3}{2}} n}{\lg \lg n}\right)$ lower bounds for dynamic problems~\cite{clifford_new_2015,larsen_super-logarithmic_2023}.

\subsection{Core decomposition lower bounds}\label{sec:kcore_corollary}

This section showed, through \Cref{res:mcv_bound_ipl,result:mcv_bound_omv,res:mcv_bound_seth,res:mcv_bound_cellprobe}, several hardness results for \MCV, conditional and unconditional. \Cref{sec:mcv_to_3core}, on the other hand, through \Cref{result:mcv_to_2-e_kcore_approx,res:mcv_to_truss,result:mcv_to_02core,res:3core_to_k2_bcore}, showed how to solve \MCV using several~\kcore variants. Putting these together, we can obtain hardness results for the~\kcore variants, which is formalized in the following corollary.

\begin{corollary} \label{res:kcore_bound_all}
    For any~$\epsilon>0$,~$\delta > 0$,~$k \geq 3$ and~$\ell \geq 2$, consider each of the following dynamic problems, and associated constants~$s$ and~$c_q$:
    \begin{enumerate}
        \item[(i)] \CoreValue or \TrussValue on a graph on~$n$ vertices and~$m$ edges, with~$s = m$ and~$c_q=1$; or
        \item[(ii)] \problem{$k$-Core} or \problem{$(k+1)$-Truss} on a graph on~$n$ vertices and~$m$ edges, with~$s = \nicefrac{m}{k}$ and~$c_q=1$; or
        \item[(iii)] \ApproxCore{$(2-\epsilon)$} on a graph on~$n$ vertices and~$m$ edges, with~$s = m \epsilon^3$ and~$c_q=\epsilon^2$; or
        \item[(iv)] \problem{$(k-1,\ell-2)$-Core} on a digraph on~$n$ vertices and~$m$ arcs, with~$s = \nicefrac{m}{k+\ell}$ and~$c_q=1$; or
        \item[(v)] \problem{$(k,\ell)$-BCore} on a bipartite graph on~$n$ vertices and~$m$ edges, with~$s = \nicefrac{m}{k\ell}$ and~$c_q=1$.
    \end{enumerate}
    No algorithm exists for any of these problems under these constraints:
    \begin{enumerate}
        \item[(a)] fully dynamic with~$\dynTime{\poly(s)}{o(c_q \frac{\lg s}{\lg \lg s})}{o(\frac{\lg s}{\lg \lg s})}$ worst-case time in the \CPROBE{$\lg n$} model; or
        \item[(b)] fully dynamic with~$\dynTime{\poly(s)}{\Oh(c_q \polylog(s))}{\polylog(s)}$ worst-case time, unless~${\P\subseteq\IPL}$; or
        \item[(c)] fully (partially) dynamic with~$\dynTime{\poly(s)}{\Oh(c_q s^{\frac{1}{2}-\delta})}{\Oh(s^{1-\delta})}$ amortized (worst-case, resp.) time, unless the \OMvConj is false; or
        \item[(d)] fully (partially) dynamic with~$\dynTime{\poly(s)}{\Oh(c_q s^{1-\delta})}{\Oh(s^{1-\delta})}$ amortized (worst-case, resp.) time, unless the \SETH is false.
    \end{enumerate}
\end{corollary}

\begin{proof}
    For any problem~$\pi$ in this list, suppose there exists an algorithm~\A for~$\pi$ with the stated complexities. By using the corresponding proposition from \Cref{sec:mcv_to_3core}, obtain an algorithm~\A' for \MCV. Then, by using the corresponding theorem from this section, we obtain a contradiction.

    For example, let us apply this for case~(iii,c), that is, for obtaining lower bounds for~\ApproxCore{$(2-\epsilon)$} via the~\OMvConj. For any~$\epsilon > 0$, assume there is a fully dynamic algorithm~\A for~\ApproxCore{$(2-\epsilon)$} on a graph on~$n$ vertices and~$m$ edges with~$\dynTime{\poly(m\epsilon^3)}{\Oh(\epsilon^2 (m\epsilon^3)^{\frac{1}{2}-\delta})}{\Oh((m\epsilon^3)^{1-\delta})}$ amortized time. By~\Cref{result:mcv_to_2-e_kcore_approx}, there exists an algorithm~\A' for~\MCV on a circuit of size~$N$ with~$\dynTime{\poly(N\epsilon^{-3} \epsilon^3)}{\Oh(\epsilon^{-2} \epsilon^2 (N\epsilon^{-3}\epsilon^3)^{\frac{1}{2}-\delta})}{\Oh((N\epsilon^{-3}\epsilon^3)^{1-\delta})}$ amortized time. Note that, by~\Cref{result:mcv_bound_omv}, this~$\dynTime{\poly(N)}{\Oh(N^{\frac{1}{2}-\delta})}{\Oh(N^{1-\delta})}$ time algorithm implies that the~\OMvConj is false.
    The other cases have similar proofs, so we omit them.
\end{proof}

The results from~\Cref{res:kcore_bound_all} are summarized in~\Cref{table:summary}. Take~\problem{$k$-Core}, for example, to explore what these bounds imply for different values of~$k$. When~$k=\polylog(m)$, no algorithm for~\problem{$k$-Core} exists with~$\polylog(m)$ update/query time, under~$\P \not\subseteq \IPL$, or~$\Oh(m^{1-\delta})$ update/query time, under the~\SETH. When~$k$ is a polynomial on~$n$, such as~$k=\Oh(n^c)$, for some~$c < 1$, then the same~$\polylog(m)$ bound still holds under~$\P \not\subseteq \IPL$, while the \SETH provides a smaller, but still polynomial,~$\Oh(m^{1-c-\delta})$ bound.

\begin{table}[h]
    \centering
    \begin{tabular}{|l|c|c|c|c|}\hline
         Problem \textbackslash\ Condition & Unconditional & $\P \not\subseteq \IPL$ & \OMvConj & \SETH \\ \hline
         \multirow{2}{*}{\CoreValue} & \multirow{2}{*}{$o(\frac{\lg m}{\lg \lg m})$} & \multirow{2}{*}{$\polylog(m)$} & $\Oh(m^{\frac{1}{2}-\delta})$ & \multirow{2}{*}{$\Oh(m^{1-\delta})$} \\
           & &  &  $\Oh(m^{1-\delta})$ & \\ \hline

         \multirow{2}{*}{\TrussValue} & \multirow{2}{*}{$o(\frac{\lg m}{\lg \lg m})$} & \multirow{2}{*}{$\polylog(m)$} & $\Oh(m^{\frac{1}{2}-\delta})$ & \multirow{2}{*}{$\Oh(m^{1-\delta})$} \\
         & &  &  $\Oh(m^{1-\delta})$ & \\ \hline

         \multirow{2}{*}{\problem{$k$-Core}} & \multirow{2}{*}{$o(\frac{\lg \nicefrac{m}{k}}{\lg \lg \nicefrac{m}{k}})$} & \multirow{2}{*}{$\polylog(\nicefrac{m}{k})$} & $\Oh((\nicefrac{m}{k})^{\frac{1}{2}-\delta})$ & \multirow{2}{*}{$\Oh((\nicefrac{m}{k})^{1-\delta})$} \\
            & &  &  $\Oh((\nicefrac{m}{k})^{1-\delta})$ & \\ \hline

         \problem{$(2-\epsilon)$-Approx} & $o(\epsilon^2  \frac{\lg m \epsilon^3}{\lg \lg m \epsilon^3}) $ & $\epsilon^2 \polylog(m \epsilon^3)$ & $\Oh(\epsilon^2 (m \epsilon^3)^{\frac{1}{2}-\delta})$ & $\Oh(\epsilon^2 (m \epsilon^3)^{1-\delta})$ \\ 
         \problem{CoreValue} & $o(\frac{\lg m \epsilon^3}{\lg \lg m \epsilon^3}) $ & $\polylog(m \epsilon^3)$ & $\Oh((m \epsilon^3)^{1-\delta})$ & $\Oh((m \epsilon^3)^{1-\delta})$ \\ \hline
         
         \multirow{2}{*}{\problem{$(k+1)$-Truss}} & \multirow{2}{*}{$o(\frac{\lg \nicefrac{m}{k}}{\lg \lg \nicefrac{m}{k}})$} & \multirow{2}{*}{$\polylog(\nicefrac{m}{k})$} & $\Oh((\nicefrac{m}{k})^{\frac{1}{2}-\delta})$ & \multirow{2}{*}{$\Oh((\nicefrac{m}{k})^{1-\delta})$} \\
         & &  &  $\Oh((\nicefrac{m}{k})^{1-\delta})$ & \\ \hline

         \multirow{2}{*}{\problem{$(k-1,\ell-2)$-Core}} & \multirow{2}{*}{$o(\frac{\lg \nicefrac{m}{k+\ell}}{\lg \lg \nicefrac{m}{k+\ell}})$} & \multirow{2}{*}{$\polylog(\nicefrac{m}{k+\ell})$} & $\Oh((\nicefrac{m}{k+\ell})^{\frac{1}{2}-\delta})$ & \multirow{2}{*}{$\Oh((\nicefrac{m}{k+\ell})^{1-\delta})$} \\
          & &  &  $\Oh((\nicefrac{m}{k+\ell})^{1-\delta})$ & \\ \hline
         
         \multirow{2}{*}{\problem{$(k,\ell)$-BCore}} & \multirow{2}{*}{$o(\frac{\lg \nicefrac{m}{k\ell}}{\lg \lg \nicefrac{m}{k\ell}})$} & \multirow{2}{*}{$\polylog(\nicefrac{m}{k\ell})$} & $\Oh((\nicefrac{m}{k\ell})^{\frac{1}{2}-\delta})$ & \multirow{2}{*}{$\Oh((\nicefrac{m}{k\ell})^{1-\delta})$} \\
           & &  &  $\Oh((\nicefrac{m}{k\ell})^{1-\delta})$ & \\ \hline
    \end{tabular}
    \caption{Summary of lower bounds for core decomposition variants. For any~$\epsilon>0$,~$\delta > 0$,~$k \geq 3$ and~$\ell \geq 2$, no algorithm exists for each problem (row) and condition (column) with the given update/query complexity. The query time is omitted when it is the same as the update time, otherwise it is below the update time. The preprocessing time is omitted.}
    \label{table:summary}
\end{table}

\section{Core maintenance unboundedness} \label{sec:kcore_unboundedness}

Many papers on dynamic~\kcore in graphs~\cite{guo_simplified_2024,sariyuce_incremental_2016,zhang_maintaining_2024} approach the problem from a different perspective. They consider a \emph{maintenance} version of the problem, which tries to take advantage of the fact that, in real life dynamic graphs, few core values change with each update.

\begin{problemStatement}{Full Core Maintenance (\KcoreMaint)~\cite{zhang_maintaining_2024}}
    Given a graph, process a series of edge insertions/deletions, keeping the core values of all vertices explicitly and up to date.
\end{problemStatement}

For any edge insertion or deletion on a graph~$G$, its \deff{affected set}~$V^* \subseteq V(G)$ is the set of vertices of the inserted or deleted edge together with any vertices which had their core value changed due to the operation.
Let $||V^*||_d$ be the number of vertices within distance~$d$ in~$G$ from any vertex in~$V^*$. An algorithm~\A for \KcoreMaint is~\deff{bounded} if there exist a constant~$d \in \mathbb{N}$ and a non-decreasing function~$f$ such that~\A takes at most~$f(||V^*||_d)$ time on any operation with affected set $V^*$.
Algorithm~\A is \deff{polynomially bounded} if~$f$ is a polynomial function. A problem is called \deff{bounded} if a bounded algorithm exists for it, and \deff{unbounded} otherwise. The idea is that operations should take time proportional to how many core values are affected. The disadvantage is that we always explicitly store the core values of all vertices, so operations that change the core value of many vertices may be slow even in a bounded algorithm. This is similar to what we call ``output-sensitive algorithms''.

\citet{zhang_unboundedness_2019} proved that, under the locally persistent~(LP) model of computation, incremental~\KcoreMaint is unbounded, but decremental~\KcoreMaint is bounded. The LP model, however, is not a very strong model of computation, as it disallows global data or pointers from a vertex to other non-adjacent vertices. Algorithms like the one in~\cite{zhang_fast_2017} are \emph{not} in LP, so this result does not show that they could not be improved to become a bounded algorithm.
We will instead adopt the RAM model~\cite{cormen_ram_2022}, which is much more generic and stronger than~LP, and we will show the unboundedness of incremental~\KcoreMaint under the~\OMvConj.

\begin{theorem} \label{res:kcore_unbounded_omv}
    There exists no bounded algorithm in the RAM model for incremental \KcoreMaint, unless the \OMvConj is false.
\end{theorem}

\begin{proof}
Suppose a bounded algorithm~\A in the RAM model for incremental \KcoreMaint exists. Given an~$N \x N$ Boolean matrix~$M$ and a vector pair~$(u,v)$, we will describe a circuit~$C^{u,v}_M$ and a graph~$G^{u,v}_M$ such that, for each gate~$g$, if its value is~1 \emph{and} the circuit value is~1, then all vertices in its graph have core value 3, otherwise the core value of all those vertices is~2.
The graph will be similar to the one built in~\Cref{result:oumv_to_mcv}, but we modify the 1-gate, which used to be a~$K_4$, so that it is in the~3-core \emph{only} when the output gate also has value 1. This has the effect that, if the circuit value is~0, then no vertex is in the 3-core, and hence a bounded algorithm must take constant time on any update that does not change the circuit value.

First, build~$C_M$ with output gate~$g^*$ from~$M$ as in the proof of~\Cref{result:oumv_to_mcv}. Let~$\1$ be the single~1-gate in~$C_M$. Then, build~$G_M$ from~$C_M$ as in~\Cref{result:mcv_to_3core}, with the following changes.
The graph of~$\1$ will instead be the arrow shown in~\Cref{fig:3core_arrow_cvp}. Vertex~$s^*$, an output vertex of the graph of~$g^*$, will play the role of the input vertex of~$\1$.  The mapping of wires to edges and the construction of~$C^{u,v}_M$ from~$C_M$, $u$, and~$v$ continues to work in the same way as in~\Cref{result:mcv_to_3core,result:oumv_to_mcv}, and this way we conclude the description of graph~$G^{u,v}_M$. An example is shown in~\Cref{fig:omv_unboundedness}. Note that~$G_M$ has~$\Oh(N^2)$ vertices and edges, and~$G^{u,v}_M$ adds at most~$2N$ edges to~$G_M$.

\begin{figure}[h]
    \centering
    \documentclass[tikz]{standalone}
\usepackage[american]{circuitikz}
\usetikzlibrary{positioning, matrix, decorations.pathreplacing, calc, fit}

\begin{document}
\begin{tikzpicture}[
    node 0/.style={circle, draw=black, fill=black, inner sep=1.5pt},
    gate/.style={draw, minimum size=0.2cm, scale=0.8},
    inputnode/.style={draw, rectangle, scale=1, minimum size = 0.8cm, font=\Large},
    one_gate_color/.style={fill=green!30, no input leads, line width=2pt},
    or_gate_color/.style={fill=cyan!30, no input leads, no output leads},
    gateone/.style={very thick},
    label/.style={font=\bfseries},
    many/.style={loosely dotted, very thick, shorten >= 5pt, shorten <= 5pt},
    cell/.style={rectangle, draw=none, minimum size=0.4cm, font=\normalsize},
    highlight/.style={cell, fill=red!30},
    matrix_paren/.style={decorate, decoration={curveto, amplitude=10pt, raise=5pt}, thick},
    vector_paren/.style={decorate, decoration={curveto, amplitude=5pt, raise=3pt}, thick},
    every edge/.append style={in=180,out=0,looseness=0.6},
    region/.style={rounded corners, fill opacity=0.15, draw, inner sep=2pt},
	one_region/.style={line width=1.5pt},
    or_gate_region/.style={region, fill=cyan},
    one_gate_region/.style={region, fill=green},
]

\useasboundingbox (-3.3, -2.1) rectangle (8.7, 2.25);

\begin{scope}[xshift=-1.5cm, yshift=-15pt]

    \matrix (M) [matrix of nodes, nodes={cell}, column sep=0.1cm, row sep=0.1cm]
    {
        1 & 0 \\
        1 & |[highlight]| 1 \\
    };
    \draw[matrix_paren, decoration={mirror}] (M-1-1.north west) -- (M-2-1.south west);
    \draw[matrix_paren] (M-1-2.north east) -- (M-2-2.south east);
    \matrix (u) [matrix of nodes, nodes={cell}, left=0.5cm of M, row sep=0.1cm]
    {
        0 \\
        |[highlight]|1 \\
    };
    \draw[vector_paren, decoration={mirror}] (u-1-1.north west) -- (u-2-1.south west);
    \draw[vector_paren] (u-1-1.north east) -- (u-2-1.south east);
    \matrix (v) [matrix of nodes, nodes={cell}, above=0.5cm of M, column sep=0.1cm]
    {
        0 & |[highlight]| 1 \\
    };
    \draw[vector_paren, yshift=1cm] (v-1-1.north west) -- (v-1-2.north east);
    \draw[vector_paren, decoration={mirror}] (v-1-1.south west) -- (v-1-2.south east);
    \node[font=\normalsize, above=0.05cm of u] {$u$};
    \node[font=\normalsize, left=0.05cm of v] {$v$};
    \node[font=\normalsize, above left=0.05cm and 0.05cm of M-1-1.north west] {$M$};
\end{scope}


\newcommand{\drawOneGate}[3]{
    \node [style=node 0] (#31) at  ($(#1, #2)$) {};      
    \node [style=node 0] (#32) at  ($(#1 + 1, #2 - 0.5)$) {};  
    \node [style=node 0] (#33) at  ($(#1 + 1, #2)$) {};  
    \node [style=node 0] (#34) at  ($(#1 + 1, #2 + 0.5)$) {};  
    \node [style=node 0] (#3 out1) at  ($(#1 + 2, #2)$) {};  
    \draw (#31) to (#32); \draw (#33) to (#32); \draw (#33) to (#34); \draw (#31) to (#34);
    \draw (#34) to (#3 out1); \draw (#33) to (#3 out1); \draw (#32) to (#3 out1);
    \node[one_gate_region, fit=(#31) (#32) (#33) (#34) (#3 out1)] {};
}

\newcommand{\drawOrGate}[3]{
    \node [style=node 0] (#31) at  ($(#1, #2 + 0.75)$) {};      
    \node [style=node 0] (#32) at  ($(#1 + 1, #2 + 0.25)$) {};  
    \node [style=node 0] (#33) at  ($(#1 + 1, #2 + 0.75)$) {};  
    \node [style=node 0] (#34) at  ($(#1 + 1, #2 + 1.25)$) {};  
    \node [style=node 0] (#3 out1) at  ($(#1 + 2, #2 + 0.75)$) {};  
    \node [style=node 0] (#36) at  ($(#1 + 1, #2 + 0.25)$) {};  
    \node [style=node 0] (#37) at  ($(#1, #2 - 0.75)$) {};      
    \node [style=node 0] (#38) at  ($(#1 + 1, #2 - 1.25)$) {};  
    \node [style=node 0] (#39) at  ($(#1 + 1, #2 - 0.75)$) {};  
    \node [style=node 0] (#310) at ($(#1 + 1, #2 - 0.25)$) {}; 
    \node [style=node 0] (#3 out2) at ($(#1 + 2, #2 - 0.75)$) {}; 
    \node [style=node 0] (#312) at ($(#1 + 1, #2 - 1.25)$) {}; 
    \node [style=node 0] (#3 in) at ($(#1 - 1, #2)$) {};        
    \draw (#31) to (#32); \draw (#33) to (#32); \draw (#33) to (#34); \draw (#31) to (#34);
    \draw (#34) to (#3 out1); \draw (#33) to (#3 out1); \draw (#36) to (#3 out1); \draw (#37) to (#38);
    \draw (#39) to (#38); \draw (#39) to (#310); \draw (#37) to (#310); \draw (#310) to (#3 out2);
    \draw (#39) to (#3 out2); \draw (#312) to (#3 out2);
    \draw (#3 in) to (#31); \draw (#3 in) to (#37);
    \node[or_gate_region, fit=(#31) (#32) (#33) (#34) (#3 out1) (#36) (#37) (#38) (#39) (#310) (#3 out2) (#312) (#3 in)] {};
}

\begin{scope}[xscale=0.6, yscale=0.6]
\drawOneGate{0}{0}{in1}

\drawOrGate{4}{1.75}{A}
\drawOrGate{4}{-1.75}{B}

\drawOrGate{8}{1.75}{C}
\drawOrGate{8}{-1.75}{D}

\drawOrGate{12}{0}{E}
\end{scope}

\node[above left=-4pt and 2pt of A1] {$L_1$};
\node[above left=-4pt and 2pt of B1] {$L_2$};
\node[above left=-4pt and 2pt of C1] {$R_1$};
\node[above left=-4pt and 2pt of D1] {$R_2$};
\node[above left=-4pt and 2pt of E1] {$g^*$};
\path[red] (in1 out1) edge (B in);
\path[red] (B out2) edge (D in);
\path[red] (D out1) edge (E in);

\path (A out1) edge (C in);
\path (B out1) edge (C in);

\path let \p1 = (E out1), \p2 = (in11), \p3 = (A4) in node (mid) at ($( \x1 * 0.5 + \x2 *0.5 , \y3 + 10)$) {};

\path (E out1) edge[out=90,in=0,looseness=0.85,orange,orange,very thick] (mid.center);
\path (mid.center) edge[out=180,in=90,looseness=1.2,orange,orange,very thick] (in11);

\end{tikzpicture}
\end{document}
    \caption{On the left, an example input to \OuMv, and on the right the final graph~$G^{u,v}_M$ described in the proof of~\Cref{res:kcore_unbounded_omv}. The different 1-gate and the \textcolor{orange}{orange} bold edge is the main difference from~\Cref{fig:example_3core_mcvp}. Note that if any of the \textcolor{red}{red} edges, which correspond to the witness on the \OuMv input, are removed, \emph{all} vertices leave the~$3$-core.}
    \Description{OuMv example and the constructed graph for the reduction.}
    \label{fig:omv_unboundedness}
\end{figure}

Note that the graphs of all gates, even when no input/output edges are present, have minimum degree at least two,\footnotemark{} so all vertices are in the~$2$-core at all times.
Suppose the circuit value of~$C^{u,v}_M$ is 0. Then~$s^*$ is not in the~$3$-core and the graph of~$\1$ is not in the~$3$-core of~$G^{u,v}_M$ because~$s^*$ is~$\1$'s input vertex, so all other vertices are not in the~$3$-core as well, and thus all vertices have core value 2.
\footnotetext{This is not the case for the~0-gate, but our reduction from \OuMv to \MCV uses only 1- and \BOR-gates.}
Now suppose the circuit value of~$C^{u,v}_M$ is 1, and consider the subcircuit of~$C^{u,v}_M$ with all gates of value 1, and the corresponding subgraph of~$G^{u,v}_M$. This subgraph has minimum degree 3, since each \BOR-gate has at least one of its input edges, and the input edge of~$\1$ is present from $s^*$ as the circuit value is 1. In fact, this subgraph is exactly the~$3$-core of~$G^{u,v}_M$. Therefore, even with the modification done to the~1-gate, we have shown that if the circuit value is 1 then the graph of each gate~$g$ is in the~$3$-core if and only if the value of~$g$ is~1.

Now we show how to solve~$\OuMv$ of size~$N$ with~$Q$ queries using our constructions. Given an~$N \x N$ matrix~$M$, build~$G_M$, which can be done in~$\Oh(N^2)$ time. After receiving each vector pair~$(u, v)$, we want to build~$G^{u,v}_M$ by adding at most~$2N$ edges to~$G_M$ and then~$u^\top M v = 1$ if and only if~$K_{s^*} = 3$. Since~\A is bounded, there exist~$d$ and~$f$ such that any insertion takes at most~$f(||V^*||_d)$ steps, where~$V^*$ is the affected set. Add edges to~$G_M$ one at a time. If, after adding an edge, the core value of~$s^*$ remains~2, then no core value changes, so~$|V^*| = 2$ and thus~$||V^*||_d \leq 2 \cdot 4^d \leq 2^{2d+1}$ since the maximum degree of the graph is~4. We modify the algorithm so that, if it takes at most~$f(2^{2d+1})$ steps for an edge insertion, then we proceed as usual. Otherwise we stop the algorithm before the~$(f(2^{2d+1})+1)$-st step, because we know some core value changed (we do not know, and it does not matter, \emph{which} core values changed),
so~$s^*$ is in the~$3$-core. In this case, we can output that~$u^\top M v = 1$ and continue to the next vector pair. Use the rollback technique, that is, record all changes to the memory of~\A and undo them before the next vector pair, which reverts the memory of~\A to its state right after building~$G_M$. This proves the correctness of the algorithm for~\OuMv.

For each of the~$Q$ vector pairs, the algorithm takes~$\Oh(2N f(2^{2d+1}))$ time, thus $\Oh(QN)$ computation time which, by \Cref{prop:oumv} with~$\epsilon = 1$, falsifies the \OMvConj.
\end{proof}

One can see this proof also disallows algorithms for core maintenance on a graph on~$m$ edges which take some of the graph size in account. That is, if each update with affected set~$V^*$ takes~$\Ohl(m^{\frac12 - \epsilon} f(||V^*||_d))$ time for any~$\epsilon > 0$ and a non-decreasing function~$f$, then the OMv conjecture is still falsified.
    
We can also extend the unboundedness result for the truss maintenance problem, as below, and also for directed and $(2-\epsilon)$-approximate core maintenance, with a similar adaptation.

\begin{problemStatement}{Full truss maintenance (\TrussMaint)~\cite{zhang_unboundedness_2019}}
    Given a graph, process a series of edge insertions/deletions, keeping the truss values of all edges explicitly and up to date.
\end{problemStatement}

\begin{restatable}{theorem}{trussUnboundedOmv} \label{res:truss_unbounded_omv}
    There exists no bounded algorithm in the RAM model for incremental \TrussMaint, unless the \OMvConj is false.
\end{restatable}

\begin{proof}
    We modify the proof of \Cref{res:kcore_unbounded_omv}. Build~$G_M$ from~$C_M$ as in \Cref{res:mcv_to_truss}, with the following changes. The graph of the single~1-gate in~$G_M$ will instead be a single edge, and we will add edges, like a wire, from the output edge~$e^*$ to the edge representing the 1-gate.
    
    With these modifications, while the circuit value of~$C_M$ is 0, only a constant number of edges change their truss value whenever a wire is inserted or removed, because the~$4$-truss does not propagate from the 1-gate unless the circuit value is~1.
    With this property in place, the rest of the proof proceeds as in \Cref{res:kcore_unbounded_omv}.
\end{proof}


\Cref{res:kcore_unbounded_omv,res:truss_unbounded_omv} show that the core and truss maintenance problems are unlikely to be bounded, even exponentially. The current line of research has been attempting to obtain a bounded algorithm for core maintenance, however, this is fated to fail.
The algorithms in~\cite{guo_simplified_2024,sariyuce_incremental_2016,zhang_fast_2017,zhang_maintaining_2024} are useful, and work fast in practice, but they inevitably take~$\tilde{\Omega}(n)$ time in the worst case, even when no vertex changes core value, and our results show this will likely always be the case.

In \Cref{fig:zhang_2017_counterexample}, we show a simple instance that proves the unboundedness of state-of-the-art core maintenance algorithms such as~\cite{zhang_maintaining_2024}. These algorithms keep an explicit order in which vertices with the lowest degree are deleted in the static core decomposition algorithm.
Start such an algorithm with the graph in \Cref{fig:zhang_2017_counterexample}, including edge~$e_1$ but excluding edge~$e_2$. Then, remove~$e_1$ and insert~$e_2$. This causes the deletion order of all~$n-6$ vertices in the middle to reverse from right-to-left to left-to-right, and the algorithm takes~$\Omega(n)$ time, even though no core value changes, proving its unboundedness.

\begin{figure}[h]
    \centering
    \begin{tikzpicture}
	\begin{pgfonlayer}{nodelayer}
		\node [style=node 0] (0) at (-3.5, 1) {};
		\node [style=node 0] (1) at (-3.5, -1) {};
		\node [style=node 0] (2) at (-2.5, 0) {};
		\node [style=node 0] (3) at (-1.5, 0) {};
		\node [style=node 0] (4) at (-0.5, 0) {};
		\node [style=none] (5) at (-2, 0.5) {$e_1$};
		\node [style=node 0] (7) at (3.5, 1) {};
		\node [style=node 0] (8) at (3.5, -1) {};
		\node [style=node 0] (9) at (2.5, 0) {};
		\node [style=node 0] (10) at (1.5, 0) {};
		\node [style=node 0] (11) at (0.5, 0) {};
		\node [style=none] (12) at (2, 0.5) {$e_2$};
		\node [style=none] (13) at (-1.5, -0.5) {$u_1$};
		\node [style=none] (14) at (-0.5, -0.5) {$u_2$};
		\node [style=none] (15) at (0.5, -0.5) {$u_{n-7}$};
		\node [style=none] (16) at (1.5, -0.5) {$u_{n-6}$};
	\end{pgfonlayer}
	\begin{pgfonlayer}{edgelayer}
		\draw (0) to (2);
		\draw (2) to (1);
		\draw (1) to (0);
		\draw (3) to (4);
		\draw [style=dashed edge] (2) to (3);
		\draw (7) to (9);
		\draw (9) to (8);
		\draw (8) to (7);
		\draw (10) to (11);
		\draw [style=dashed edge] (9) to (10);
		\draw [style=cdots arrow] (4) to (11);
	\end{pgfonlayer}
\end{tikzpicture}
    \caption{Example that proves the unboundedness of some core maintenance algorithms.}
    \Description{Two triangles connected by a path of n-6 degree-2 vertices, with edges e1 and e2 at the ends connected to the triangles.}
    \label{fig:zhang_2017_counterexample}
\end{figure}%

\section{\twoCore algorithm} \label{sec:2core_alg}

In~\Cref{sec:mcv_to_3core,sec:all_lower} we showed that a polylogarithmic algorithm for determining whether a vertex is in the~$k$-core of a dynamic graph for~$k \geq 3$ would falsify the \OMvlink and \SETH conjectures, and imply a polylogarithmic algorithm for the incremental variant of any problem in \P. In this section, we show that \emph{there is} a polylogarithmic algorithm for~$k=2$, complementing this result.
Specifically, we describe an algorithm with time complexity~$\dynOh{m \log n}{\log^2 n}{\log n}$\footnote{Here $m$ and $n$ are, as usual, the number of edges and vertices in the graph.} for the following problem:

\begin{problemStatement}{\twoCore}
    Maintain a graph, subject to edge insertion/deletion operations, with query: Given a vertex, is it in the~$2$-core?
\end{problemStatement}


The algorithm described in this section was implemented and extensively tested for correctness. Its source code is available at~\cite{yan_soares_couto_dynamic_2024}.
Related to this, and not superseding the result given here, \citet{sun_fully_2020} described a~$(4+\epsilon)$-approximation for \CoreValue that runs in polylogarithmic time.

\subsection{Using spanning forests to determine the~2-core}

For any (unrooted) maximal spanning forest~$F$ of a graph~$G$, let~$T_u$ be the component of~$F$ which contains~$u$. The \deff{$F$-subtrees of~$u$} are the connected components of~${T_u - u}$, which are trees, and we root them at the (unique) neighbor of~$u$ in them. We say that~$uv \in E(G)$ is a \deff{tree edge} if~$uv \in E(F)$ and an \deff{extra edge} otherwise. Notice that any extra edge connects vertices in the same component of~$F$, since~$F$ is spanning and maximal, and thus always creates a unique cycle when added to~$F$. A vertex is~\deff{$F$-special} if it has an extra edge incident to it. Our algorithm will maintain a maximal spanning forest~$F$ of the dynamic graph~$G$, and use~$F$ to determine which vertices are in the~2-core of~$G$.


\begin{lemma} \label{lemma:2core_iff_F}
    A vertex~$u$ is in the~$2$-core of~$G$ if and only if, for any maximal spanning forest~$F$ of~$G$,~$u$ is $F$-special or two distinct $F$-subtrees of~$u$ have an $F$-special vertex.
\end{lemma}

\begin{proof}
    If~$u$ is $F$-special, then there is an extra edge~$uv$ in~$G$. This edge forms a cycle in~$G$ containing $u$.
    On the other hand, if~$v$ and~$w$ are in different $F$-subtrees of~$u$ and are $F$-special, then both of them are contained in cycles, and either~$u$ is contained in one of these cycles, or~$u$ is contained in the unique path in~$F$ between these two cycles. In all cases, we have a subgraph of~$G$ with minimum degree~2 containing~$u$, thus~$u$ is in the~$2$-core of~$G$.


    Otherwise, if no extra edge is incident to~$u$ and at most one of its $F$-subtrees has an $F$-special vertex, suppose~$u$ is in the~$2$-core of~$G$. Then it has at least two incident edges in the~$2$-core, and thus at least one of its $F$-subtrees without an $F$-special vertex intersects the~$2$-core. However, the intersection of that subtree with the~$2$-core is a forest, and it does not have any incident extra edges, thus it has a leaf with degree~1 in $G$ in the~$2$-core, a contradiction. So~$u$ is not in the~$2$-core of~$G$.
\end{proof}

By maintaining a maximal spanning forest~$F$ of~$G$, and a list of the extra edges, we can use Lemma~\ref{lemma:2core_iff_F} to determine whether vertex~$u$ is in the~$2$-core of~$G$. See \Cref{fig:2-core-lemma}. Determining if~$u$ is $F$-special is trivial. The difficulty is to efficiently determine whether any two distinct $F$-subtrees of~$u$ have an $F$-special vertex. For this, we use the following.

\begin{figure}[h]
    \centering
    \documentclass[tikz]{standalone}
\usepackage{tikz}
\usetikzlibrary{arrows.meta, positioning, calc, shapes.geometric}

\begin{document}

\begin{tikzpicture}[
    font=\sffamily,
    >=Stealth,
    vertex/.style={
        circle, 
        draw=black!80, 
        fill=white, 
        thick, 
        minimum size=4mm,
        inner sep=0pt,
    },
    structural edge/.style={
        draw=gray!40,
        line width=3pt,
    },
    base arc/.style={
        ->,
        thick,
        bend left=15
    },
    base loop/.style={
        ->,
        thick,
        min distance=8mm,
        in=120,
        out=60,
        looseness=5
    },
    order label/.style={
        auto,
        text=#1,
        font=\bfseries\scriptsize,
        inner sep=0.5pt,
        minimum size=3.5mm,
        fill opacity=0.9,
        text opacity=1
    },
    scale=0.6
]

    \coordinate (pos_tree) at (0,0);
    \coordinate (pos_digraph) at (6,0);
    
    \coordinate (pos_tourA) at (0,-5);
    \coordinate (pos_tourB) at (6,-5);

    \begin{scope}[shift={(pos_tree)}, local bounding box=tree]
        \node[vertex] (r) at (0,0) {r};
        \node[vertex] (u) at (-1.5,-2) {u};
        \node[vertex] (v) at (1.5,-2) {v};

        \draw[structural edge] (r) -- (u);
        \draw[structural edge] (r) -- (v);

        \node[anchor=south, font=\bfseries] at (0, 0.9) {Original Tree};
    \end{scope}

    \begin{scope}[shift={(pos_digraph)}, local bounding box=digraph]
        \node[vertex] (r) at (0,0) {r};
        \node[vertex] (u) at (-1.5,-2) {u};
        \node[vertex] (v) at (1.5,-2) {v};

        \draw[base loop, draw=gray!60] (r) to (r);
        \draw[base loop, draw=gray!60, in=135, out=225] (u) to (u); 
        \draw[base loop, draw=gray!60, in=45, out=-45] (v) to (v); 

        \draw[base arc, draw=gray!60] (r) to (u);
        \draw[base arc, draw=gray!60] (u) to (r);
        \draw[base arc, draw=gray!60] (r) to (v);
        \draw[base arc, draw=gray!60] (v) to (r);

        \node[anchor=south, font=\bfseries] at (0, 0.9) {Digraph};
    \end{scope}

    \begin{scope}[shift={(pos_tourA)}, local bounding box=tourA]
        \node[vertex] (r) at (0,0) {r};
        \node[vertex] (u) at (-1.5,-2) {u};
        \node[vertex] (v) at (1.5,-2) {v};
        \def\theme{blue!70!black}

        \draw[base loop, color=\theme] (r) to node[order label=\theme, above] {1} (r);
        
        \draw[base arc, color=\theme] (r) to node[order label=\theme, right, pos=0.6] {2} (u);
        
        \draw[base loop, in=135, out=225, color=\theme] (u) to node[order label=\theme, left] {3} (u);
        
        \draw[base arc, color=\theme] (u) to node[order label=\theme, left, pos=0.4] {4} (r);
        
        \draw[base arc, color=\theme] (r) to node[order label=\theme, right, pos=0.6] {5} (v);
        
        \draw[base loop, in=45, out=-45, color=\theme] (v) to node[order label=\theme, right] {6} (v);
        
        \draw[base arc, color=\theme] (v) to node[order label=\theme, left, pos=0.4] {7} (r);

        \node[anchor=south, font=\bfseries, text=\theme] at (0, 1.3) {Tour A};
    \end{scope}

    \begin{scope}[shift={(pos_tourB)}, local bounding box=tourB]
        \node[vertex] (r) at (0,0) {r};
        \node[vertex] (u) at (-1.5,-2) {u};
        \node[vertex] (v) at (1.5,-2) {v};
        \def\theme{red!70!black}

        \draw[base arc, color=\theme] (r) to node[order label=\theme, right, pos=0.6] {1} (u);
        
        \draw[base loop, in=135, out=225, color=\theme] (u) to node[order label=\theme, left] {2} (u);
        
        \draw[base arc, color=\theme] (u) to node[order label=\theme, left, pos=0.4] {3} (r);
        
        \draw[base arc, color=\theme] (r) to node[order label=\theme, right, pos=0.6] {5} (v);
        
        \draw[base loop, in=45, out=-45, color=\theme] (v) to node[order label=\theme, right] {6} (v);
        
        \draw[base arc, color=\theme] (v) to node[order label=\theme, left, pos=0.4] {7} (r);

        \draw[base loop, color=\theme] (r) to node[order label=\theme, above] {4} (r);

        \node[anchor=south, font=\bfseries, text=\theme] at (0, 1.3) {Tour B};
    \end{scope}

    
    \path ($(pos_tree)!0.5!(pos_digraph)$) coordinate (mid_x_top);
    \path ($(pos_tourA)!0.5!(pos_tourB)$) coordinate (mid_x_bot);
    
    \draw[dashed, gray!30] 
        ($(mid_x_top) + (0, 1.5)$) -- 
        ($(mid_x_bot) + (0, -2.9)$);
        
    \coordinate (mid_y) at ($(tree.south)!0.5!(tourA.north)$);
    
    \draw[dashed, gray!30] 
        ($(pos_tree) + (-2.5, -2.8)$) -- 
        ($(pos_digraph) + (2.5, -2.8)$);

\end{tikzpicture}
\end{document}
    \caption{Example of an Euler tour on a rooted tree. Top: The transformation of the original tree (left) into a digraph (right). Bottom: Two valid Euler tours for the original tree.}
    \Description{Example tree with root r and two children u and v. Two examples of euler tours, rr ru uu ur rv vv vr and ru uu ur rr rv vv vr.}
    \label{fig:euler_tour}
\end{figure}

\newcommand\ET{\mathit{ET}}

An \deff{Euler tour} of a rooted tree is a sequence of arcs made by any Eulerian trail of the digraph created by replacing each edge of the tree by two arcs, one in each direction, plus adding a self-loop arc on each vertex. For simplicity, we say the position of a vertex in an Euler tour is the position of its self-loop. The trail should start at the root of the tree. See \Cref{fig:euler_tour} for an example.

\begin{figure}[h]
    \centering
    \documentclass[tikz, border=5mm]{standalone}
\usepackage{tikz-qtree}
\usetikzlibrary{backgrounds,positioning,fit}

\begin{document}
\begin{tikzpicture}[
    edge from parent path={(\tikzparentnode) -- (\tikzchildnode)},
    sibling distance=20pt,
    region/.style={fill=orange!20, rounded corners, inner sep=3pt},
    ]
    \tikzstyle{vertex}=[circle, draw, inner sep=2pt, minimum size=7pt, font=\small, fill=white]
    \tikzstyle{special}=[vertex, fill=blue!50]
    \tikzstyle{extra_edge}=[-, thick, dashed]
    \tikzstyle{info}=[draw=none, fill=none, font=\small]
    \tikzstyle{queried}=[line width=1pt, draw=red]

    \begin{scope}[xshift=-2.75cm]
        \Tree [.\node[vertex, queried](u_a){u};
            [.\node[vertex](v_a){v}; \node[special](x_a){x}; \node[special](y_a){y}; ]
            [.\node[special](w_a){w}; \node[special](z_a){z}; ]
            [.\node[vertex](p_a){p}; \node[vertex](q_a){q}; ]
        ]
        \draw[extra_edge] (x_a) to[bend right=40] (y_a);
        \draw[extra_edge] (z_a) to[bend left=40] (w_a);

        \begin{scope}[on background layer]
            \node[region, fit=(v_a)(x_a)(y_a)] {};
            \node[region, fit=(w_a)(z_a)] {};
        \end{scope}
        
    \end{scope}

    \begin{scope}[xshift=2.75cm]
        \Tree [.\node[vertex, queried](p_c){p};
            [.\node[vertex](u_c){u};
                [.\node[vertex](v_c){v}; \node[special](x_c){x}; \node[special](y_c){y}; ]
                [.\node[special](w_c){w}; \node[special](z_c){z}; ]
            ]
            \node[vertex](q_c){q};
        ]
        \draw[extra_edge] (x_c) to[bend right=40] (y_c);
        \draw[extra_edge] (z_c) to[bend left=40] (w_c);

        \begin{scope}[on background layer]
            \node[region, fit=(u_c)(v_c)(w_c)(x_c)(y_c)(z_c)] {};
        \end{scope}
        
    \end{scope}

    \begin{scope}[yshift=0cm, xshift=5cm]
        \node[special, label=right:Special vertex] at (0,0) {};
        \node[vertex, queried, label=right:Queried vertex] at (0,-1) {};
        \draw[-] (-0.25, -2) -- (0.2, -2);
        \node[info, right] at (0.2, -2) {Tree edge};
        \draw[densely dashed] (-.25, -3) -- (0.2, -3);
        \node[info, right] at (0.2, -3) {Extra edge};
    \end{scope}

\end{tikzpicture}
\end{document}
    \caption{Using Lemma~\ref{lemma:2core_iff_F} to determine if a vertex is in the 2-core.
    On the left, vertex $u$ is not $F$-special, but is in the 2-core because two of its $F$-subtrees (shaded) contain $F$-special vertices. On the right, vertex $p$ is not in the 2-core, as it is not $F$-special and only one of its $F$-subtrees (shaded) contains $F$-special vertices.}
    \Description{Two examples for the lemma: left vertex u is in the two core because two special subtrees exist, right vertex p is not becaus it only has one special subtree.}
    \label{fig:2-core-lemma}
\end{figure}

\begin{lemma} \label{lemma:et_lca}
    If~$T$ is a tree rooted at a vertex~$u$, $\ET$ is an Euler tour of~$T$, and~${S \subseteq V(T) \setminus \{u\}}$, then two distinct $T$-subtrees of~$u$ have vertices from~$S$ if and only if~$u$ is in the unique path in~$T$ between~$v$ and~$w$, where~$v$ and~$w$ are the first and the last vertices from~$S$ in~$\ET$.
\end{lemma}

\begin{proof}
    Let~$v$ and $w$ be the first and last vertices from~$S$ in~$\ET$, respectively.
    If~$u$ is in the unique path in $T$ between~$v$ and~$w$, then~$v$ and~$w$ are in distinct $T$-subtrees of~$u$, as they are disconnected by the removal of~$u$.

    Analogously, if~$u$ is not in the path between~$v$ and~$w$, then they are in the same $T$-subtree of~$u$. Notice that every sequence of arcs in an Euler tour is a walk in~$T$, visiting vertices and traversing edges, up or down. Furthermore, in an Euler tour, after walking down an edge~$uv$, all arcs in the $T$-subtree of~$u$ rooted at~$v$ must be traversed before walking up the edge in the reverse direction~$vu$, otherwise the sequence would not be an Eulerian trail, as arcs~$uv$ and~$vu$ must appear exactly once in the sequence. Let~$x$ be the root of the $T$-subtree of~$u$ containing~$v$ and~$w$. We conclude that the arc~$ux$ appears before~$v$ and~$xu$ appears after~$w$ in the sequence. Additionally, all arcs between~$ux$ and~$xu$ are in that same $T$-subtree of~$u$. Therefore, all vertices from~$S$, which by choice are between~$v$ and~$w$ in~$\ET$, and thus between~$ux$ and~$xu$, are also in the same $T$-subtree of~$u$.
\end{proof}

To determine whether a given vertex of~$G$ is in its~$2$-core, we can maintain both a maximal spanning forest of~$G$ and Euler tours for each of its connected components, and then use \Cref{lemma:2core_iff_F,lemma:et_lca}. We now describe data structures which allow doing that efficiently.

\subsection{Auxiliary data structures and algorithms}

We will maintain two dynamic tree structures over~$F$: an~\emph{Euler tour tree}, used to find the first and last $F$-special vertices on some Euler tour of a tree~$T_u$, and a \emph{link-cut tree}, used to determine if~$u$ is in the unique path in~$T_u$ between those vertices. For a graph on $n$ vertices, all update and query operations in both structures take amortized~$\Oh(\lg n)$ time.

\newcommand{\op}[1]{\texttt{#1}}

An Euler tour tree (ETT)~\cite{henzinger_randomized_1999,ett_stanford_notes} is a data structure for representing a forest, which supports edge insertions and deletions.
It maintains an Euler tour of each tree of the forest, stored in a balanced binary search tree (BST). When storing a maximal spanning forest $F$ of a graph $G$, to accommodate queries on the first and last $F$-special vertices in the Euler tour, the BST which stores each Euler tour can be augmented to store aggregated data on each BST node, such as the total number of $F$-special vertices in its subtree. Then, we can traverse the BST to find the first and last $F$-special vertices in~$\Oh(\lg n)$ time. The ETT structure also allows making any node the root of its tree, as that can be done by BST split and join operations.

A Link-Cut Tree (LCT)~\cite{sleator_data_1981,lct_demaine_notes} is a data structure for storing a forest, which supports edge insertions and deletions (by combining operations \op{link(v, w)}, \op{cut(v)}, and \op{evert(v)} from~\cite{sleator_data_1981}), and efficiently querying information about paths in the forest. In particular, it allows for queries on the lowest common ancestor (LCA) of two vertices (using~\op{nca(v, w)}). Given $v$ and $w$, if we make~$u$ the root of its tree (by calling \op{evert(u)}), then the LCA of~$v$ and~$w$ is~$u$ if and only if~$u$ is in the unique path between~$v$ and~$w$.

In order to maintain a maximal spanning forest $F$ of~$G$, we will use the HDT algorithm~\cite{holm_poly-logarithmic_1998} for dynamic connectivity, which has amortized time complexity~$\dynOh{m \lg n}{\lg^2 n}{\frac{\lg n}{\lg \lg n}}$, where $m$ and $n$ are the number of edges and vertices in~$G$. For a dynamic graph~$G$, HDT maintains a maximal spanning forest~$F$ of~$G$, in which it answers connectivity queries. The maximal spanning forest~$F$ changes in the following way:

\begin{itemize}
    \item When an edge~$uv$ is inserted in~$G$, if~$u$ and~$v$ are already connected in~$F$, then~$F$ does not change. Otherwise,~$uv$ is inserted in~$F$.
    \item When an edge~$uv$ is deleted from~$G$, if it is not contained in~$F$, then~$F$ does not change. Otherwise,~$uv$ is deleted from~$F$, and a single replacement edge~$wz$ is inserted in~$F$ if~$u$ and~$v$ are still connected in~$G - uv$.
\end{itemize}

The most intricate part of HDT is efficiently finding a replacement edge, if it exists. The details of how that is done are not relevant here, and do not need to be modified, as we only care about the final maximal spanning forest~$F$. Our final algorithm works like this:

\begin{enumerate}
    \item Given an initial graph~$G$, create an empty LCT and ETT.
    \item Preprocess~$G$ and apply each edge insertion/deletion operation using the HDT algorithm.
    \item Whenever the HDT algorithm inserts or removes an edge from~$F$, insert or remove that edge from both the ETT and the LCT. On the ETT, keep track of which vertices are $F$-special by maintaining how many extra edges are incident on each vertex.
    \item On a query operation, to determine whether~$u$ is in the~$2$-core of~$G$:
    \begin{enumerate}
        \item If~$u$ is $F$-special, then $u$ is in the~$2$-core.
        \item Otherwise, make~$u$ the root of its tree in the ETT, and use its BST to find~$v$ and~$w$, the first and last $F$-special vertices in the stored Euler tour of~$T_u$. If there is no $F$-special vertex in $T_u$, then $u$ is not in the 2-core.
        \item Using the LCT, make~$u$ the root of its tree. Then~$u$ will be in the~$2$-core if and only if $u$ is the LCA of~$v$ and~$w$.
    \end{enumerate}
\end{enumerate}

Every edge insertion/deletion operation results in at most one insertion and one deletion on both the ETT\footnotemark{} and the LCT. Thus, the total update time is~$\Oh(\lg^2 n)$, from the HDT algorithm itself. Furthermore, queries take~$\Oh(\lg n)$ since it performs a constant number of operations on the ETT and LCT, each taking~$\Oh(\lg n)$ time, yielding the amortized complexity~$\dynOh{m \lg n}{\lg^2 n}{\lg n}$ for~\twoCore. The correctness follows from~\Cref{lemma:2core_iff_F,lemma:et_lca}.

\footnotetext{It is possible to avoid using ETTs, and solve~\twoCore using only LCTs. ETTs are used to determine the ``first'' and ``last'' descendant of~$u$ that is $F$-special, and the HDT algorithm also uses ETTs for similar reasons. It is possible to augment the link-cut tree structure to allow for descendant queries~\cite{klein_chapter_2021,yufan_you_maintain_2019}, and thus solve these sub-problems. However, that significantly complicates the implementation and is not commonly done in the literature. So we decided to explain and implement the algorithm using ETTs.} 

\documentclass[../main-waoa.tex]{subfiles}


\begin{document}

\section{Planarizing gadgets for~\texorpdfstring{$k$}{k}-core do not exist} \label{sec:no_crossing}

Even though the maximum core value on a planar graph is~5, it is not yet known whether the \kcore problem on planar graphs has an efficient dynamic or parallel algorithm. Consider the parallel setting. In order to prove the hardness\footnotemark{} of 3-core, an \NC reduction from monotone circuit value to~3-core is done, since monotone circuit value is a known hard problem~\cite{greenlaw_limits_1995}. However, while this reduction can maintain planarity, the planar monotone circuit value problem can actually be solved efficiently in parallel~\cite{ramachandran_efficient_1996}, so the existence of the reduction does \emph{not} imply planar~3-core is hard. Moreover, that efficient algorithm for planar monotone circuit value is not known to translate back to the~\kcore problem.
\footnotetext{A parallel problem is hard if it is~\P-complete under~\NC reductions, that is, it is unlikely to be solvable in polylogarithmic parallel time.}

Very few other planar problems are known to be \P-complete (under \NC reductions), and the ones that are have no known reduction to~\kcore. Another common approach used for proving hardness of a planar problem~$\pi$ is to reduce the non-planar version of~$\pi$ to its planar version. This is done by devising a \emph{planarizing gadget}\footnote{It may also be called a crossing gadget or a crossover gadget.}. We start from a drawing of a non-planar instance of~$\pi$ and then, for each crossing, we replace it with a planarizing gadget such that the new instance has one less crossing, while its answer remains the same.
For example, planarizing gadgets are used to prove \P-completeness of planar circuit value~\cite{goldschlager_monotone_1977} and \NP-completeness of planar~3-coloring~\cite{garey_simplified_1976}.

In this section, we prove that a planarizing gadget for~\kcore cannot exist for any~$k \geq 3$. Proving the nonexistence of such gadgets, as in~\cite{gurjar_planarizing_2016,mccoll_planar_1981}, is useful because it shows this approach for proving hardness is a dead-end, providing some indication that the planar problem might not be hard after all. In this case, hardness would probably need to be proved using a reduction from one of the (scarce) hard planar problems. Even though this indicates an efficient parallel or dynamic algorithm for planar~\kcore might exist, one is yet to be found.

\newcommand{\mP}{\mathcal{P}}

\begin{definition}[Planarizing gadget] \label{def:planarc}
    Let $k \geq 3$. A graph $\mP = \mP(a, b, c, d)$ is a \emph{planarizing gadget} for \kcore if:
    \begin{enumerate}
        \item It is a plane graph (an embedded planar graph).
        \item The vertices $a, b, c, d$ lie on its outer face, and in the circular order $a, c, b, d$.
        \item For any graph $G$ containing edges $ab$ and $cd$, let $G(\mP) \coloneqq G - \{ab, cd\} + \mP$ (identifying the vertices $a,b,c,d$ of $G$ with the ones in $\mP$). Then $G$ has a non-empty \kcore if and only if $G(\mP)$ has a non-empty \kcore. 
    \end{enumerate}
\end{definition}

\newcommand{\AXB}{\mathrm{AXB}}
\newcommand{\CXD}{\mathrm{CXD}}

\begin{theorem} \label{res:no_crossing}
    There exists no planarizing gadget for~\kcore.
\end{theorem}
\begin{proof}
    First consider the case when~$k=3$, and suppose a planarizing gadget~$\mP = \mP(a,b,c,d)$ exists. Let~$k(G)$ be the~3-core of a graph~$G$, and~$G[U]$ be the subgraph of~$G$ induced by the vertex set~$U \subseteq V(G)$. See~\Cref{fig:no_crossing} for the helper graphs $G_\varnothing, G_{ab}, G_{cd}, G_a, G_f$ used in this proof.

\begin{figure*}[h]
\centering
\captionsetup[subfigure]{justification=centering}
\begin{subfigure}{0.33\textwidth}
    \centering
    \def\todraw{1}
    \documentclass[tikz, border=5mm]{standalone}
\usepackage[american]{circuitikz}
\usetikzlibrary{positioning, matrix, decorations.pathreplacing, calc, fit,scopes,backgrounds,shapes.geometric}

\begin{document}
\begin{tikzpicture}[
    node 0/.style={circle, draw=black, minimum size=11pt, font=\small, inner sep=1pt},
    node inner/.style={circle, draw=black, minimum size=8pt, inner sep=1.5pt},
]

\newcommand{\drawArrowInline}[2]{
\begin{scope}[xscale=0.5]
    \node [style=node inner] (#22) at  (1, -0.5) {};  
    \node [style=node inner] (#23) at  (1, 0) {};  
    \node [style=node inner] (#24) at  (1, 0.5) {};  
    \node [style=node inner] (#25) at  (2, 0) {};  
    \draw (#1) to (#22); \draw (#23) to (#22); \draw (#23) to (#24); \draw (#1) to (#24);
    \draw (#24) to (#25); \draw (#23) to (#25); \draw (#22) to (#25);
\end{scope}
}

\newcommand{\drawArrow}[1]{
    \node [style=node inner] (#11) at  (0,0) {};      
    \drawArrowInline{#11}{#1}
}

\def\p{0.5}

\newcommand{\drawBase}[4]{
    \node[style=node 0,#1] (a) at (-\p, \p) {a};
    \node[style=node 0,#2] (b) at (\p, -\p) {b};
    \node[style=node 0,#3] (c) at (-\p, -\p) {c};
    \node[style=node 0,#4] (d) at (\p, \p) {d};
    \draw[color=red,thick] (a) -- (b);
    \draw[color=blue,thick] (c) -- (d);
}

\ifdefined\todraw\else
\def\todraw{5}
\fi

\if\todraw0
\drawBase{}{}{}{}
\scoped[shift={(-\p,\p)}, rotate=180]
    \drawArrowInline{a}{aa};
\fi

\if\todraw1
\drawBase{fill=gray!60}{fill=gray!60}{}{}
\scoped[shift={(-\p,\p)}, rotate=180,every node/.style={fill=gray!60}]
    \drawArrowInline{a}{aa};
\scoped[shift={(\p,-\p)}, rotate=0,every node/.style={fill=gray!60}]
    \drawArrowInline{b}{bb};
\fi

\if\todraw2
\drawBase{}{}{fill=gray!60}{fill=gray!60}
\scoped[shift={(-\p,-\p)}, rotate=-180,every node/.style={fill=gray!60}]
    \drawArrowInline{c}{cc};
\scoped[shift={(\p,\p)}, rotate=0,every node/.style={fill=gray!60}]
    \drawArrowInline{d}{dd};
\fi

\if\todraw3
\drawBase{}{}{}{}
\scoped[shift={(-\p*2,\p)}, rotate=180]
    \drawArrow{aa};
\scoped[shift={(\p,-\p)}, rotate=0]
    \drawArrowInline{b}{bb};
\draw (a) -- (aa1);
\draw (a) -- (d);

\fi

\if\todraw4
\drawBase{}{}{}{}
\fi

\if\todraw5
\begin{scope}[scale=0.9]
\def\p{1}
\node[style=node 0] (a) at (-\p, \p) {a};
\node[style=node 0] (b) at (\p, -\p) {b};
\node[style=node 0] (c) at (-\p, -\p) {c};
\node[style=node 0] (d) at (\p, \p) {d};
\begin{scope}[on background layer, every node/.style={draw, rotate fit=45, fill opacity=0.3, rounded corners=8pt, inner sep=7pt, line width=1.2pt}]
\node[fill=red, fit=(a)(b)] {};
\node[fill=blue, fit=(c)(d)] {};
\end{scope}
\node (A) at (-0.65*\p, 0.65*\p) {A};
\node (B) at (0.65*\p, -0.65*\p) {B};
\node (C) at (-0.65*\p, -0.65*\p) {C};
\node (D) at (0.65*\p, 0.65*\p) {D};
\node (X) at (0,0) {X};
\node (Y) at (0,1.5*\p) {Y};
\end{scope}
\fi

\end{tikzpicture}
\end{document}
    \captionsetup{textformat=simple}
    \subcaption{$G_{ab}$ \label{fig:no_crossing:Gab}}
    \Description{Two edges AB and CD, where A and B are both input vertices of two arrows used in the 3-core reduction.}
\end{subfigure}%
\begin{subfigure}{0.33\textwidth}
    \centering
    \def\todraw{5}
    \documentclass[tikz, border=5mm]{standalone}
\usepackage[american]{circuitikz}
\usetikzlibrary{positioning, matrix, decorations.pathreplacing, calc, fit,scopes,backgrounds,shapes.geometric}

\begin{document}
\begin{tikzpicture}[
    node 0/.style={circle, draw=black, minimum size=11pt, font=\small, inner sep=1pt},
    node inner/.style={circle, draw=black, minimum size=8pt, inner sep=1.5pt},
]

\newcommand{\drawArrowInline}[2]{
\begin{scope}[xscale=0.5]
    \node [style=node inner] (#22) at  (1, -0.5) {};  
    \node [style=node inner] (#23) at  (1, 0) {};  
    \node [style=node inner] (#24) at  (1, 0.5) {};  
    \node [style=node inner] (#25) at  (2, 0) {};  
    \draw (#1) to (#22); \draw (#23) to (#22); \draw (#23) to (#24); \draw (#1) to (#24);
    \draw (#24) to (#25); \draw (#23) to (#25); \draw (#22) to (#25);
\end{scope}
}

\newcommand{\drawArrow}[1]{
    \node [style=node inner] (#11) at  (0,0) {};      
    \drawArrowInline{#11}{#1}
}

\def\p{0.5}

\newcommand{\drawBase}[4]{
    \node[style=node 0,#1] (a) at (-\p, \p) {a};
    \node[style=node 0,#2] (b) at (\p, -\p) {b};
    \node[style=node 0,#3] (c) at (-\p, -\p) {c};
    \node[style=node 0,#4] (d) at (\p, \p) {d};
    \draw[color=red,thick] (a) -- (b);
    \draw[color=blue,thick] (c) -- (d);
}

\ifdefined\todraw\else
\def\todraw{5}
\fi

\if\todraw0
\drawBase{}{}{}{}
\scoped[shift={(-\p,\p)}, rotate=180]
    \drawArrowInline{a}{aa};
\fi

\if\todraw1
\drawBase{fill=gray!60}{fill=gray!60}{}{}
\scoped[shift={(-\p,\p)}, rotate=180,every node/.style={fill=gray!60}]
    \drawArrowInline{a}{aa};
\scoped[shift={(\p,-\p)}, rotate=0,every node/.style={fill=gray!60}]
    \drawArrowInline{b}{bb};
\fi

\if\todraw2
\drawBase{}{}{fill=gray!60}{fill=gray!60}
\scoped[shift={(-\p,-\p)}, rotate=-180,every node/.style={fill=gray!60}]
    \drawArrowInline{c}{cc};
\scoped[shift={(\p,\p)}, rotate=0,every node/.style={fill=gray!60}]
    \drawArrowInline{d}{dd};
\fi

\if\todraw3
\drawBase{}{}{}{}
\scoped[shift={(-\p*2,\p)}, rotate=180]
    \drawArrow{aa};
\scoped[shift={(\p,-\p)}, rotate=0]
    \drawArrowInline{b}{bb};
\draw (a) -- (aa1);
\draw (a) -- (d);

\fi

\if\todraw4
\drawBase{}{}{}{}
\fi

\if\todraw5
\begin{scope}[scale=0.9]
\def\p{1}
\node[style=node 0] (a) at (-\p, \p) {a};
\node[style=node 0] (b) at (\p, -\p) {b};
\node[style=node 0] (c) at (-\p, -\p) {c};
\node[style=node 0] (d) at (\p, \p) {d};
\begin{scope}[on background layer, every node/.style={draw, rotate fit=45, fill opacity=0.3, rounded corners=8pt, inner sep=7pt, line width=1.2pt}]
\node[fill=red, fit=(a)(b)] {};
\node[fill=blue, fit=(c)(d)] {};
\end{scope}
\node (A) at (-0.65*\p, 0.65*\p) {A};
\node (B) at (0.65*\p, -0.65*\p) {B};
\node (C) at (-0.65*\p, -0.65*\p) {C};
\node (D) at (0.65*\p, 0.65*\p) {D};
\node (X) at (0,0) {X};
\node (Y) at (0,1.5*\p) {Y};
\end{scope}
\fi

\end{tikzpicture}
\end{document}
    \captionsetup{textformat=simple}
    \subcaption{$\mP = A \cupdot B \cupdot C \cupdot D \cupdot X \cupdot Y$ \label{fig:no-crossing:split}}
    \Description{Planarizing gadget split into sets A, B, C, D, X, and Y.}
\end{subfigure}%
\begin{subfigure}{0.33\textwidth}
    \centering
    \def\todraw{2}
    \documentclass[tikz, border=5mm]{standalone}
\usepackage[american]{circuitikz}
\usetikzlibrary{positioning, matrix, decorations.pathreplacing, calc, fit,scopes,backgrounds,shapes.geometric}

\begin{document}
\begin{tikzpicture}[
    node 0/.style={circle, draw=black, minimum size=11pt, font=\small, inner sep=1pt},
    node inner/.style={circle, draw=black, minimum size=8pt, inner sep=1.5pt},
]

\newcommand{\drawArrowInline}[2]{
\begin{scope}[xscale=0.5]
    \node [style=node inner] (#22) at  (1, -0.5) {};  
    \node [style=node inner] (#23) at  (1, 0) {};  
    \node [style=node inner] (#24) at  (1, 0.5) {};  
    \node [style=node inner] (#25) at  (2, 0) {};  
    \draw (#1) to (#22); \draw (#23) to (#22); \draw (#23) to (#24); \draw (#1) to (#24);
    \draw (#24) to (#25); \draw (#23) to (#25); \draw (#22) to (#25);
\end{scope}
}

\newcommand{\drawArrow}[1]{
    \node [style=node inner] (#11) at  (0,0) {};      
    \drawArrowInline{#11}{#1}
}

\def\p{0.5}

\newcommand{\drawBase}[4]{
    \node[style=node 0,#1] (a) at (-\p, \p) {a};
    \node[style=node 0,#2] (b) at (\p, -\p) {b};
    \node[style=node 0,#3] (c) at (-\p, -\p) {c};
    \node[style=node 0,#4] (d) at (\p, \p) {d};
    \draw[color=red,thick] (a) -- (b);
    \draw[color=blue,thick] (c) -- (d);
}

\ifdefined\todraw\else
\def\todraw{5}
\fi

\if\todraw0
\drawBase{}{}{}{}
\scoped[shift={(-\p,\p)}, rotate=180]
    \drawArrowInline{a}{aa};
\fi

\if\todraw1
\drawBase{fill=gray!60}{fill=gray!60}{}{}
\scoped[shift={(-\p,\p)}, rotate=180,every node/.style={fill=gray!60}]
    \drawArrowInline{a}{aa};
\scoped[shift={(\p,-\p)}, rotate=0,every node/.style={fill=gray!60}]
    \drawArrowInline{b}{bb};
\fi

\if\todraw2
\drawBase{}{}{fill=gray!60}{fill=gray!60}
\scoped[shift={(-\p,-\p)}, rotate=-180,every node/.style={fill=gray!60}]
    \drawArrowInline{c}{cc};
\scoped[shift={(\p,\p)}, rotate=0,every node/.style={fill=gray!60}]
    \drawArrowInline{d}{dd};
\fi

\if\todraw3
\drawBase{}{}{}{}
\scoped[shift={(-\p*2,\p)}, rotate=180]
    \drawArrow{aa};
\scoped[shift={(\p,-\p)}, rotate=0]
    \drawArrowInline{b}{bb};
\draw (a) -- (aa1);
\draw (a) -- (d);

\fi

\if\todraw4
\drawBase{}{}{}{}
\fi

\if\todraw5
\begin{scope}[scale=0.9]
\def\p{1}
\node[style=node 0] (a) at (-\p, \p) {a};
\node[style=node 0] (b) at (\p, -\p) {b};
\node[style=node 0] (c) at (-\p, -\p) {c};
\node[style=node 0] (d) at (\p, \p) {d};
\begin{scope}[on background layer, every node/.style={draw, rotate fit=45, fill opacity=0.3, rounded corners=8pt, inner sep=7pt, line width=1.2pt}]
\node[fill=red, fit=(a)(b)] {};
\node[fill=blue, fit=(c)(d)] {};
\end{scope}
\node (A) at (-0.65*\p, 0.65*\p) {A};
\node (B) at (0.65*\p, -0.65*\p) {B};
\node (C) at (-0.65*\p, -0.65*\p) {C};
\node (D) at (0.65*\p, 0.65*\p) {D};
\node (X) at (0,0) {X};
\node (Y) at (0,1.5*\p) {Y};
\end{scope}
\fi

\end{tikzpicture}
\end{document}
    \captionsetup{textformat=simple}
    \subcaption{$G_{cd}$}
    \Description{Two edges AB and CD, where C and D are both input vertices of two arrows used in the 3-core reduction.}
\end{subfigure}%
\\%
\begin{subfigure}{0.33\textwidth}
    \centering
    \def\todraw{0}
    \documentclass[tikz, border=5mm]{standalone}
\usepackage[american]{circuitikz}
\usetikzlibrary{positioning, matrix, decorations.pathreplacing, calc, fit,scopes,backgrounds,shapes.geometric}

\begin{document}
\begin{tikzpicture}[
    node 0/.style={circle, draw=black, minimum size=11pt, font=\small, inner sep=1pt},
    node inner/.style={circle, draw=black, minimum size=8pt, inner sep=1.5pt},
]

\newcommand{\drawArrowInline}[2]{
\begin{scope}[xscale=0.5]
    \node [style=node inner] (#22) at  (1, -0.5) {};  
    \node [style=node inner] (#23) at  (1, 0) {};  
    \node [style=node inner] (#24) at  (1, 0.5) {};  
    \node [style=node inner] (#25) at  (2, 0) {};  
    \draw (#1) to (#22); \draw (#23) to (#22); \draw (#23) to (#24); \draw (#1) to (#24);
    \draw (#24) to (#25); \draw (#23) to (#25); \draw (#22) to (#25);
\end{scope}
}

\newcommand{\drawArrow}[1]{
    \node [style=node inner] (#11) at  (0,0) {};      
    \drawArrowInline{#11}{#1}
}

\def\p{0.5}

\newcommand{\drawBase}[4]{
    \node[style=node 0,#1] (a) at (-\p, \p) {a};
    \node[style=node 0,#2] (b) at (\p, -\p) {b};
    \node[style=node 0,#3] (c) at (-\p, -\p) {c};
    \node[style=node 0,#4] (d) at (\p, \p) {d};
    \draw[color=red,thick] (a) -- (b);
    \draw[color=blue,thick] (c) -- (d);
}

\ifdefined\todraw\else
\def\todraw{5}
\fi

\if\todraw0
\drawBase{}{}{}{}
\scoped[shift={(-\p,\p)}, rotate=180]
    \drawArrowInline{a}{aa};
\fi

\if\todraw1
\drawBase{fill=gray!60}{fill=gray!60}{}{}
\scoped[shift={(-\p,\p)}, rotate=180,every node/.style={fill=gray!60}]
    \drawArrowInline{a}{aa};
\scoped[shift={(\p,-\p)}, rotate=0,every node/.style={fill=gray!60}]
    \drawArrowInline{b}{bb};
\fi

\if\todraw2
\drawBase{}{}{fill=gray!60}{fill=gray!60}
\scoped[shift={(-\p,-\p)}, rotate=-180,every node/.style={fill=gray!60}]
    \drawArrowInline{c}{cc};
\scoped[shift={(\p,\p)}, rotate=0,every node/.style={fill=gray!60}]
    \drawArrowInline{d}{dd};
\fi

\if\todraw3
\drawBase{}{}{}{}
\scoped[shift={(-\p*2,\p)}, rotate=180]
    \drawArrow{aa};
\scoped[shift={(\p,-\p)}, rotate=0]
    \drawArrowInline{b}{bb};
\draw (a) -- (aa1);
\draw (a) -- (d);

\fi

\if\todraw4
\drawBase{}{}{}{}
\fi

\if\todraw5
\begin{scope}[scale=0.9]
\def\p{1}
\node[style=node 0] (a) at (-\p, \p) {a};
\node[style=node 0] (b) at (\p, -\p) {b};
\node[style=node 0] (c) at (-\p, -\p) {c};
\node[style=node 0] (d) at (\p, \p) {d};
\begin{scope}[on background layer, every node/.style={draw, rotate fit=45, fill opacity=0.3, rounded corners=8pt, inner sep=7pt, line width=1.2pt}]
\node[fill=red, fit=(a)(b)] {};
\node[fill=blue, fit=(c)(d)] {};
\end{scope}
\node (A) at (-0.65*\p, 0.65*\p) {A};
\node (B) at (0.65*\p, -0.65*\p) {B};
\node (C) at (-0.65*\p, -0.65*\p) {C};
\node (D) at (0.65*\p, 0.65*\p) {D};
\node (X) at (0,0) {X};
\node (Y) at (0,1.5*\p) {Y};
\end{scope}
\fi

\end{tikzpicture}
\end{document}
    \captionsetup{textformat=simple}
    \subcaption{$G_a$}
    \Description{Two edges AB and CD, where A is an input vertex of one arrow used in the 3-core reduction.}
\end{subfigure}%
\begin{subfigure}{0.33\textwidth}
    \centering
    \def\todraw{4}
    \documentclass[tikz, border=5mm]{standalone}
\usepackage[american]{circuitikz}
\usetikzlibrary{positioning, matrix, decorations.pathreplacing, calc, fit,scopes,backgrounds,shapes.geometric}

\begin{document}
\begin{tikzpicture}[
    node 0/.style={circle, draw=black, minimum size=11pt, font=\small, inner sep=1pt},
    node inner/.style={circle, draw=black, minimum size=8pt, inner sep=1.5pt},
]

\newcommand{\drawArrowInline}[2]{
\begin{scope}[xscale=0.5]
    \node [style=node inner] (#22) at  (1, -0.5) {};  
    \node [style=node inner] (#23) at  (1, 0) {};  
    \node [style=node inner] (#24) at  (1, 0.5) {};  
    \node [style=node inner] (#25) at  (2, 0) {};  
    \draw (#1) to (#22); \draw (#23) to (#22); \draw (#23) to (#24); \draw (#1) to (#24);
    \draw (#24) to (#25); \draw (#23) to (#25); \draw (#22) to (#25);
\end{scope}
}

\newcommand{\drawArrow}[1]{
    \node [style=node inner] (#11) at  (0,0) {};      
    \drawArrowInline{#11}{#1}
}

\def\p{0.5}

\newcommand{\drawBase}[4]{
    \node[style=node 0,#1] (a) at (-\p, \p) {a};
    \node[style=node 0,#2] (b) at (\p, -\p) {b};
    \node[style=node 0,#3] (c) at (-\p, -\p) {c};
    \node[style=node 0,#4] (d) at (\p, \p) {d};
    \draw[color=red,thick] (a) -- (b);
    \draw[color=blue,thick] (c) -- (d);
}

\ifdefined\todraw\else
\def\todraw{5}
\fi

\if\todraw0
\drawBase{}{}{}{}
\scoped[shift={(-\p,\p)}, rotate=180]
    \drawArrowInline{a}{aa};
\fi

\if\todraw1
\drawBase{fill=gray!60}{fill=gray!60}{}{}
\scoped[shift={(-\p,\p)}, rotate=180,every node/.style={fill=gray!60}]
    \drawArrowInline{a}{aa};
\scoped[shift={(\p,-\p)}, rotate=0,every node/.style={fill=gray!60}]
    \drawArrowInline{b}{bb};
\fi

\if\todraw2
\drawBase{}{}{fill=gray!60}{fill=gray!60}
\scoped[shift={(-\p,-\p)}, rotate=-180,every node/.style={fill=gray!60}]
    \drawArrowInline{c}{cc};
\scoped[shift={(\p,\p)}, rotate=0,every node/.style={fill=gray!60}]
    \drawArrowInline{d}{dd};
\fi

\if\todraw3
\drawBase{}{}{}{}
\scoped[shift={(-\p*2,\p)}, rotate=180]
    \drawArrow{aa};
\scoped[shift={(\p,-\p)}, rotate=0]
    \drawArrowInline{b}{bb};
\draw (a) -- (aa1);
\draw (a) -- (d);

\fi

\if\todraw4
\drawBase{}{}{}{}
\fi

\if\todraw5
\begin{scope}[scale=0.9]
\def\p{1}
\node[style=node 0] (a) at (-\p, \p) {a};
\node[style=node 0] (b) at (\p, -\p) {b};
\node[style=node 0] (c) at (-\p, -\p) {c};
\node[style=node 0] (d) at (\p, \p) {d};
\begin{scope}[on background layer, every node/.style={draw, rotate fit=45, fill opacity=0.3, rounded corners=8pt, inner sep=7pt, line width=1.2pt}]
\node[fill=red, fit=(a)(b)] {};
\node[fill=blue, fit=(c)(d)] {};
\end{scope}
\node (A) at (-0.65*\p, 0.65*\p) {A};
\node (B) at (0.65*\p, -0.65*\p) {B};
\node (C) at (-0.65*\p, -0.65*\p) {C};
\node (D) at (0.65*\p, 0.65*\p) {D};
\node (X) at (0,0) {X};
\node (Y) at (0,1.5*\p) {Y};
\end{scope}
\fi

\end{tikzpicture}
\end{document}
    \captionsetup{textformat=simple}
    \subcaption{$G_\varnothing$}
    \Description{Two edges AB and CD.}
\end{subfigure}%
\begin{subfigure}{0.33\textwidth}
    \centering
    \def\todraw{3}
    \documentclass[tikz, border=5mm]{standalone}
\usepackage[american]{circuitikz}
\usetikzlibrary{positioning, matrix, decorations.pathreplacing, calc, fit,scopes,backgrounds,shapes.geometric}

\begin{document}
\begin{tikzpicture}[
    node 0/.style={circle, draw=black, minimum size=11pt, font=\small, inner sep=1pt},
    node inner/.style={circle, draw=black, minimum size=8pt, inner sep=1.5pt},
]

\newcommand{\drawArrowInline}[2]{
\begin{scope}[xscale=0.5]
    \node [style=node inner] (#22) at  (1, -0.5) {};  
    \node [style=node inner] (#23) at  (1, 0) {};  
    \node [style=node inner] (#24) at  (1, 0.5) {};  
    \node [style=node inner] (#25) at  (2, 0) {};  
    \draw (#1) to (#22); \draw (#23) to (#22); \draw (#23) to (#24); \draw (#1) to (#24);
    \draw (#24) to (#25); \draw (#23) to (#25); \draw (#22) to (#25);
\end{scope}
}

\newcommand{\drawArrow}[1]{
    \node [style=node inner] (#11) at  (0,0) {};      
    \drawArrowInline{#11}{#1}
}

\def\p{0.5}

\newcommand{\drawBase}[4]{
    \node[style=node 0,#1] (a) at (-\p, \p) {a};
    \node[style=node 0,#2] (b) at (\p, -\p) {b};
    \node[style=node 0,#3] (c) at (-\p, -\p) {c};
    \node[style=node 0,#4] (d) at (\p, \p) {d};
    \draw[color=red,thick] (a) -- (b);
    \draw[color=blue,thick] (c) -- (d);
}

\ifdefined\todraw\else
\def\todraw{5}
\fi

\if\todraw0
\drawBase{}{}{}{}
\scoped[shift={(-\p,\p)}, rotate=180]
    \drawArrowInline{a}{aa};
\fi

\if\todraw1
\drawBase{fill=gray!60}{fill=gray!60}{}{}
\scoped[shift={(-\p,\p)}, rotate=180,every node/.style={fill=gray!60}]
    \drawArrowInline{a}{aa};
\scoped[shift={(\p,-\p)}, rotate=0,every node/.style={fill=gray!60}]
    \drawArrowInline{b}{bb};
\fi

\if\todraw2
\drawBase{}{}{fill=gray!60}{fill=gray!60}
\scoped[shift={(-\p,-\p)}, rotate=-180,every node/.style={fill=gray!60}]
    \drawArrowInline{c}{cc};
\scoped[shift={(\p,\p)}, rotate=0,every node/.style={fill=gray!60}]
    \drawArrowInline{d}{dd};
\fi

\if\todraw3
\drawBase{}{}{}{}
\scoped[shift={(-\p*2,\p)}, rotate=180]
    \drawArrow{aa};
\scoped[shift={(\p,-\p)}, rotate=0]
    \drawArrowInline{b}{bb};
\draw (a) -- (aa1);
\draw (a) -- (d);

\fi

\if\todraw4
\drawBase{}{}{}{}
\fi

\if\todraw5
\begin{scope}[scale=0.9]
\def\p{1}
\node[style=node 0] (a) at (-\p, \p) {a};
\node[style=node 0] (b) at (\p, -\p) {b};
\node[style=node 0] (c) at (-\p, -\p) {c};
\node[style=node 0] (d) at (\p, \p) {d};
\begin{scope}[on background layer, every node/.style={draw, rotate fit=45, fill opacity=0.3, rounded corners=8pt, inner sep=7pt, line width=1.2pt}]
\node[fill=red, fit=(a)(b)] {};
\node[fill=blue, fit=(c)(d)] {};
\end{scope}
\node (A) at (-0.65*\p, 0.65*\p) {A};
\node (B) at (0.65*\p, -0.65*\p) {B};
\node (C) at (-0.65*\p, -0.65*\p) {C};
\node (D) at (0.65*\p, 0.65*\p) {D};
\node (X) at (0,0) {X};
\node (Y) at (0,1.5*\p) {Y};
\end{scope}
\fi

\end{tikzpicture}
\end{document}
    \captionsetup{textformat=simple}
    \subcaption{$G_f$ \label{fig:no_crossing:Gf}}
    \Description{Two edges AB and CD, with an extra edge AD connecting them, where B is the input vertex of one arrow used in the 3-core reduction, and A is connected by an edge to the input vertex of another arrow.}
\end{subfigure}%
\caption{Graphs for the proof of~\Cref{res:no_crossing}. Vertices in gray are in the~3-core. \label{fig:no_crossing}}
\end{figure*}

    First, note that $G_\varnothing(\mP) = \mP$ thus~$k(\mP) = k(G_\varnothing(\mP)) = k(G_\varnothing) = \varnothing$.
    Now, consider the graph~$G_{ab}$ from~\Cref{fig:no_crossing:Gab}, and let~$\mathrm{AXB} = \mP \cap k(G_{ab}(\mP))$ be the part of~$\mP$ that intersects with the~3-core of~$G_{ab}(\mP)$, which is non-empty. Then~$\{a,b\} \subseteq V(\AXB)$. Analogously, let~$\CXD = \mP \cap k(G_{cd}(\mP))$ and $\{c,d\} \subseteq V(\CXD)$. There exists a path from~$a$ to~$b$ in~$\AXB$, otherwise, take~$C_a$ to be the connected component of~$a$ in~$\AXB$, and note that~$G_a(C_a)$ would have a non-empty~3-core while~$G_a$ does not. The same follows for~$c$ and~$d$ in~$\CXD$.
    
    Since~$\mP$ is a plane graph, this partitions the vertices of~$\mP$ into sets~$A$,~$B$,~$C$,~$D$,~$X$ and~$Y$, with~$(a, b, c, d) \in A \x B \x C \x D$ and~$X = V(\AXB) \cap V(\CXD)$, like in~\Cref{fig:no-crossing:split}. There are no edges from~$C$ to~$D$, since otherwise there would be a crossing between those edges and some edge from the~$ab$-path in~$\AXB$.

    For every vertex~$u \in V(\AXB)$, we have that (1)~$\delta_{\AXB}(u)$, the degree of~$u$ in the subgraph~$\AXB$, is at least 3 if~$u \notin \{a,b\}$ and at least~1 otherwise, since~$u$ is part of~$k(G_{ab}(\mP))$. The same follows for~$\CXD$, and furthermore, since there are no edges from~$C$ to~$D$, it also follows that (2)~$\delta_{\mathrm{DX}}(v)$, for any~$v \in D$, is at least~3 if~$v \neq d$ and at least~1 otherwise, where~$\mathrm{DX} = \mP[D\cupdot X]$.

    Now, consider~$G_f$ from~\Cref{fig:no_crossing:Gf}, which has an empty~3-core. Consider~$G_f(\mP[A\cupdot X \cupdot  B\cupdot D])$ and note that, because of the edge $ad$, all its vertices have degree at least 3, the vertices in~$A \cupdot X \cupdot B$ due to~(1) and the vertices in~$D$ due to~(2). This implies~$k(G_f(\mP)) \neq \varnothing$, a contradiction to the existence of~$\mP$.
\end{proof}

For the dynamic planar~$k$-core problem, the nonexistence of the planarizing gadget implies we cannot prove its hardness directly from the OMv or SETH techniques used in~\Cref{sec:all_lower}. In the literature, there has been some work on finding lower bounds for planar dynamic problems via alternative planar reductions from~OMv~\cite{abboud_popular_2016}, but no results for core decomposition exist so far.


@article{zhang_maintaining_2024,
  title    = {Maintaining top-\textit{t} cores in dynamic graphs},
  volume   = {36},
  issn     = {1558-2191},
  doi      = {10.1109/TKDE.2023.3332638},
  number   = {9},
  journal  = {IEEE Transactions on Knowledge and Data Engineering},
  author   = {Zhang, Yikai and Yu, Jeffrey Xu and Zhang, Ying and Qin, Lu},
  optmonth = sep,
  year     = {2024},
  pages    = {4766--4780}
}

@article{seidman_network_1983,
  title    = {Network structure and minimum degree},
  volume   = {5},
  issn     = {0378-8733},
  doi      = {10.1016/0378-8733(83)90028-X},
  number   = {3},
  journal  = {Social Networks},
  author   = {Seidman, Stephen B.},
  optmonth = sep,
  year     = {1983},
  pages    = {269--287}
}

@misc{batagelj_m_2003,
  title        = {An \textit{𝒪}(\textit{m}) algorithm for cores decomposition of networks},
  doi          = {10.48550/arXiv.cs/0310049},
  optpublisher = {arXiv},
  author       = {Batagelj, Vladimir and Zaveršnik, Matjaž},
  optmonth     = oct,
  year         = {2003}
}

@book{huang_community_2019,
  optaddress   = {Cham},
  series       = {Synthesis {Lectures} on {Data} {Management}},
  title        = {Community {Search} over {Big} {Graphs}},
  copyright    = {https://www.springer.com/tdm},
  isbn         = {978-3-031-01874-9},
  language     = {en},
  optpublisher = {Springer International Publishing},
  author       = {Huang, Xin and Lakshmanan, Laks V. S. and Xu, Jianliang},
  year         = {2019},
  doi          = {10.1007/978-3-031-01874-9}
}

@article{hanauer_recent_2022,
  title    = {Recent advances in fully dynamic graph algorithms – a quick reference guide},
  volume   = {27},
  issn     = {1084-6654, 1084-6654},
  doi      = {10.1145/3555806},
  language = {en},
  journal  = {ACM J. Exp. Algorithmics},
  author   = {Hanauer, Kathrin and Henzinger, Monika R. and Schulz, Christian},
  optmonth = dec,
  year     = {2022},
  pages    = {1--45}
}

@inproceedings{zhang_unboundedness_2019,
  optaddress   = {Amsterdam Netherlands},
  title        = {Unboundedness and efficiency of truss maintenance in evolving graphs},
  isbn         = {978-1-4503-5643-5},
  doi          = {10.1145/3299869.3300082},
  language     = {en},
  booktitle    = {2019 {International} {Conference} on {Management} of {Data} ({SIGMOD})},
  optpublisher = {ACM},
  author       = {Zhang, Yikai and Yu, Jeffrey Xu},
  optmonth     = jun,
  year         = {2019},
  pages        = {1024--1041}
}

@article{sun_fully_2020,
  title    = {Fully dynamic approximate \textit{k}-core decomposition in hypergraphs},
  volume   = {14},
  issn     = {1556-4681},
  doi      = {10.1145/3385416},
  number   = {4},
  journal  = {ACM Trans. Knowl. Discov. Data},
  author   = {Sun, Bintao and Chan, T.-H. Hubert and Sozio, Mauro},
  optmonth = may,
  year     = {2020},
  pages    = {39:1--39:21}
}

@article{malliaros_core_2020,
  title      = {The core decomposition of networks: theory, algorithms and applications},
  volume     = {29},
  issn       = {1066-8888, 0949-877X},
  shorttitle = {The core decomposition of networks},
  doi        = {10.1007/s00778-019-00587-4},
  language   = {en},
  number     = {1},
  journal    = {The VLDB Journal},
  author     = {Malliaros, Fragkiskos D. and Giatsidis, Christos and Papadopoulos, Apostolos N. and Vazirgiannis, Michalis},
  optmonth   = jan,
  year       = {2020},
  pages      = {61--92}
}

@article{hastad_clique_1999,
  title    = {Clique is hard to approximate within $n^{1-\epsilon}$},
  volume   = {182},
  issn     = {0001-5962},
  doi      = {10.1007/BF02392825},
  language = {en},
  number   = {1},
  journal  = {Acta Math.},
  author   = {Håstad, Johan},
  year     = {1999},
  pages    = {105--142}
}

@inproceedings{saito_extracting_2006,
  optaddress   = {Hong Kong, China},
  title        = {Extracting communities from complex networks by the \textit{k}-dense method},
  isbn         = {978-0-7695-2702-4},
  doi          = {10.1109/ICDMW.2006.76},
  booktitle    = {Sixth {IEEE} {International} {Conference} on {Data} {Mining} - {Workshops} ({ICDMW}'06)},
  optpublisher = {IEEE},
  author       = {Saito, Kazumi and Yamada, Takeshi},
  year         = {2006},
  pages        = {300--304}
}

@inproceedings{abboud_popular_2014,
  optaddress   = {Philadelphia, PA, USA},
  title        = {Popular conjectures imply strong lower bounds for dynamic problems},
  isbn         = {978-1-4799-6517-5},
  doi          = {10.1109/FOCS.2014.53},
  booktitle    = {{IEEE} 55th {Annual} {Symposium} on {Foundations} of {Computer} {Science} ({FOCS})},
  optpublisher = {IEEE},
  author       = {Abboud, Amir and Williams, Virginia Vassilevska},
  optmonth     = oct,
  year         = {2014},
  pages        = {434--443}
}

@misc{abboud_popular_2014_full,
  title        = {Popular conjectures imply strong lower bounds for dynamic problems},
  doi          = {10.48550/arXiv.1402.0054},
  optpublisher = {arXiv},
  author       = {Abboud, Amir and Williams, Virginia Vassilevska},
  optmonth     = feb,
  year         = {2014}
}

@article{li_efficient_2014,
  title     = {Efficient core maintenance in large dynamic graphs},
  volume    = {26},
  copyright = {https://ieeexplore.ieee.org/Xplorehelp/downloads/license-information/IEEE.html},
  issn      = {1041-4347},
  doi       = {10.1109/TKDE.2013.158},
  number    = {10},
  journal   = {IEEE Trans. Knowl. Data Eng.},
  author    = {Li, Rong-Hua and Yu, Jeffrey Xu and Mao, Rui},
  optmonth  = oct,
  year      = {2014},
  pages     = {2453--2465}
}

@article{sariyuce_incremental_2016,
  title      = {Incremental \textit{k}-core decomposition: algorithms and evaluation},
  volume     = {25},
  issn       = {1066-8888, 0949-877X},
  shorttitle = {Incremental k-core decomposition},
  doi        = {10.1007/s00778-016-0423-8},
  language   = {en},
  number     = {3},
  journal    = {The VLDB Journal},
  author     = {Sarıyüce, Ahmet Erdem and Gedik, Buğra and Jacques-Silva, Gabriela and Wu, Kun-Lung and Çatalyürek, Umit V.},
  optmonth   = jun,
  year       = {2016},
  pages      = {425--447}
}

@inproceedings{guo_parallel_2023,
  optaddress   = {Salt Lake City UT USA},
  title        = {Parallel order-based core maintenance in dynamic graphs},
  isbn         = {979-8-4007-0843-5},
  doi          = {10.1145/3605573.3605597},
  language     = {en},
  booktitle    = {52nd {International} {Conference} on {Parallel} {Processing} ({ICPP})},
  optpublisher = {ACM},
  author       = {Guo, Bin and Sekerinski, Emil},
  optmonth     = aug,
  year         = {2023},
  pages        = {122--131}
}

@incollection{du_core_2020,
  optaddress   = {Cham},
  title        = {Core decomposition, maintenance and applications},
  volume       = {12000},
  isbn         = {978-3-030-41671-3},
  language     = {en},
  booktitle    = {Complexity and {Approximation}},
  optpublisher = {Springer International Publishing},
  author       = {Zhang, Feiteng and Liu, Bin and Fang, Qizhi},
  editor       = {Du, Ding-Zhu and Wang, Jie},
  year         = {2020},
  doi          = {10.1007/978-3-030-41672-0_12},
  pages        = {205--218}
}

@inproceedings{hu_maintaining_2017,
  optaddress   = {Singapore Singapore},
  title        = {Maintaining densest subsets efficiently in evolving hypergraphs},
  isbn         = {978-1-4503-4918-5},
  doi          = {10.1145/3132847.3132907},
  language     = {en},
  booktitle    = {2017 {ACM} on {Conference} on {Information} and {Knowledge} {Management} ({CIKM})},
  optpublisher = {ACM},
  author       = {Hu, Shuguang and Wu, Xiaowei and Chan, T.-H. Hubert},
  optmonth     = nov,
  year         = {2017},
  pages        = {929--938}
}

@incollection{cormen_ram_2022,
  optaddress   = {Cambridge, Massachusetts London, England},
  edition      = {Fourth},
  title        = {{RAM} model},
  isbn         = {978-0-262-04630-5},
  language     = {eng},
  booktitle    = {Introduction to algorithms},
  optpublisher = {The MIT Press},
  author       = {Cormen, Thomas H. and Leiserson, Charles E. and Rivest, Ronald L. and Stein, Clifford},
  year         = {2022},
  pages        = {25--26}
}

@inproceedings{henzinger_unifying_2015,
  optaddress   = {Portland Oregon USA},
  title        = {Unifying and Strengthening Hardness for Dynamic Problems via the {Online} {Matrix}-{Vector} {Multiplication} {Conjecture}},
  isbn         = {978-1-4503-3536-2},
  doi          = {10.1145/2746539.2746609},
  language     = {en},
  booktitle    = {47th {Annual} {ACM} {Symposium} on {Theory} of {Computing} ({STOC})},
  optpublisher = {ACM},
  author       = {Henzinger, Monika R. and Krinninger, Sebastian and Nanongkai, Danupon and Saranurak, Thatchaphol},
  optmonth     = jun,
  year         = {2015},
  pages        = {21--30}
}

@article{guo_simplified_2024,
  title    = {Simplified algorithms for order-based core maintenance},
  volume   = {80},
  issn     = {0920-8542, 1573-0484},
  doi      = {10.1007/s11227-024-06190-x},
  language = {en},
  number   = {13},
  journal  = {J. Supercomputing},
  author   = {Guo, Bin and Sekerinski, Emil},
  optmonth = sep,
  year     = {2024},
  pages    = {19592--19623}
}

@inproceedings{zhang_fast_2017,
  optaddress   = {San Diego, CA, USA},
  title        = {A fast order-based approach for core maintenance},
  isbn         = {978-1-5090-6543-1},
  doi          = {10.1109/ICDE.2017.93},
  booktitle    = {{IEEE} 33rd {International} {Conference} on {Data} {Engineering} ({ICDE})},
  optpublisher = {IEEE},
  author       = {Zhang, Yikai and Yu, Jeffrey Xu and Zhang, Ying and Qin, Lu},
  optmonth     = apr,
  year         = {2017},
  pages        = {337--348}
}

@article{cohen_trusses_2008,
  title    = {Trusses: {Cohesive} subgraphs for social network analysis},
  volume   = {16},
  number   = {3.1},
  journal  = {National security agency technical report},
  author   = {Cohen, Jonathan},
  year     = {2008},
  OPTnote     = {Publisher: Citeseer},
  pages    = {1--29}
}

@article{wang_truss_2012,
  title    = {Truss decomposition in massive networks},
  volume   = {5},
  issn     = {2150-8097},
  doi      = {10.14778/2311906.2311909},
  language = {en},
  number   = {9},
  journal  = {Proc. VLDB Endow.},
  author   = {Wang, Jia and Cheng, James},
  optmonth = may,
  year     = {2012},
  pages    = {812--823}
}

@inproceedings{gabert_batch_2022,
  optaddress   = {Philadelphia Pennsylvania},
  title        = {Batch dynamic algorithm to find \textit{k}-core hierarchies},
  isbn         = {978-1-4503-9384-3},
  doi          = {10.1145/3534540.3534694},
  language     = {en},
  booktitle    = {5th {ACM} {SIGMOD} {Joint} {International} {Workshop} on {Graph} {Data} {Management} {Experiences} \& {Systems} ({GRADES}) and {Network} {Data} {Analytics} ({NDA})},
  optpublisher = {ACM},
  author       = {Gabert, Kasimir and Pinar, Ali and Çatalyürek, Umit V.},
  optmonth     = jun,
  year         = {2022},
  pages        = {1--10}
}

@inproceedings{giatsidis_d-cores_2011,
  optaddress   = {Vancouver, BC, Canada},
  title        = {D-cores: measuring collaboration of directed graphs based on degeneracy},
  isbn         = {978-1-4577-2075-8},
  shorttitle   = {D-cores},
  doi          = {10.1109/ICDM.2011.46},
  booktitle    = {{IEEE} 11th {International} {Conference} on {Data} {Mining} ({ICDM})},
  optpublisher = {IEEE},
  author       = {Giatsidis, Christos and Thilikos, Dimitrios M. and Vazirgiannis, Michalis},
  optmonth     = dec,
  year         = {2011},
  pages        = {201--210}
}

@article{fang_effective_2019,
  title     = {Effective and efficient community search over large directed graphs},
  volume    = {31},
  copyright = {https://ieeexplore.ieee.org/Xplorehelp/downloads/license-information/IEEE.html},
  issn      = {1041-4347, 1558-2191, 2326-3865},
  doi       = {10.1109/TKDE.2018.2872982},
  number    = {11},
  journal   = {IEEE Trans. Knowl. Data Eng.},
  author    = {Fang, Yixiang and Wang, Zhongran and Cheng, Reynold and Wang, Hongzhi and Hu, Jiafeng},
  optmonth  = nov,
  year      = {2019},
  pages     = {2093--2107}
}

@article{tian_maximal_2023,
  title    = {Maximal {D}-truss search in dynamic directed graphs},
  volume   = {16},
  issn     = {2150-8097},
  doi      = {10.14778/3598581.3598592},
  language = {en},
  number   = {9},
  journal  = {Proc. VLDB Endow.},
  author   = {Tian, Anxin and Zhou, Alexander and Wang, Yue and Chen, Lei},
  optmonth = may,
  year     = {2023},
  pages    = {2199--2211}
}

@inproceedings{larsen_super-logarithmic_2023,
  optaddress   = {Santa Cruz, CA, USA},
  title        = {Super-logarithmic lower bounds for dynamic graph problems},
  copyright    = {https://doi.org/10.15223/policy-029},
  isbn         = {979-8-3503-1894-4},
  doi          = {10.1109/FOCS57990.2023.00096},
  booktitle    = {{IEEE} 64th {Annual} {Symposium} on {Foundations} of {Computer} {Science} ({FOCS})},
  optpublisher = {IEEE},
  author       = {Larsen, Kasper Green and Yu, Huacheng},
  optmonth     = nov,
  year         = {2023},
  pages        = {1589--1604}
}

@misc{ett_stanford_notes,
  optaddress = {Stanford},
  title      = {Lecture {Notes}: {Euler} {Tour} {Trees}},
  url        = {https://web.stanford.edu/class/cs166/lectures/15/Small15.pdf},
  author     = {Schwarz, Keith},
  year       = {2023},
  note       = {{Stanford}, CS166, Class 15}
}

@misc{lct_demaine_notes,
  title  = {Lecture {Notes}: {Link}-{Cut} {Trees}},
  url    = {https://courses.csail.mit.edu/6.851/spring12/scribe/L19.pdf},
  author = {Demaine, Erik},
  year   = {2012},
  note   = {{MIT}, Advanced Data Structures (Spring'12)}
}

@article{henzinger_randomized_1999,
  title    = {Randomized fully dynamic graph algorithms with polylogarithmic time per operation},
  volume   = {46},
  issn     = {0004-5411, 1557-735X},
  doi      = {10.1145/320211.320215},
  language = {en},
  number   = {4},
  journal  = {J. ACM},
  author   = {Henzinger, Monika R. and King, Valerie},
  optmonth = jul,
  year     = {1999},
  pages    = {502--516}
}

@incollection{klein_chapter_2021,
  optaddress   = {USA},
  title        = {Chapter 18: {Splay} trees and link-cut trees},
  url          = {https://planarity.org/Klein_splay_trees_and_link-cut_trees.pdf},
  booktitle    = {Optimization {Algorithms} for {Planar} {Graphs}},
  optpublisher = {Not published},
  author       = {Klein, Philip and Shay, Mozes},
  optmonth     = dec,
  year         = {2021},
  pages        = {249--280}
}

@inproceedings{sleator_data_1981,
  optaddress   = {Milwaukee, Wisconsin, United States},
  title        = {A data structure for dynamic trees},
  doi          = {10.1145/800076.802464},
  language     = {en},
  booktitle    = {13th {Annual} {ACM} {Symposium} on {Theory} of {Computing} ({STOC})},
  optpublisher = {ACM Press},
  author       = {Sleator, Daniel D. and Tarjan, Robert E.},
  year         = {1981},
  pages        = {114--122}
}

@inproceedings{holm_poly-logarithmic_1998,
  optaddress   = {Dallas, Texas, United States},
  title        = {Poly-logarithmic deterministic fully-dynamic algorithms for connectivity, minimum spanning tree, 2-edge, and biconnectivity},
  isbn         = {978-0-89791-962-3},
  doi          = {10.1145/276698.276715},
  language     = {en},
  booktitle    = {30th {Annual} {ACM} {Symposium} on {Theory} of {Computing}  ({STOC})},
  optpublisher = {ACM Press},
  author       = {Holm, Jacob and De Lichtenberg, Kristian and Thorup, Mikkel},
  year         = {1998},
  pages        = {79--89}
}

@misc{yan_soares_couto_dynamic_2024,
  title  = {Dynamic 2-core},
  url    = {https://github.com/yancouto/phd/tree/main/dynamic_2core},
  author = {Yan Soares Couto},
  year   = {2024}
}

@misc{yan_p_completeness_compendium_2026,
  title  = {P-completeness Compendium},
  url    = {https://yancouto.github.io/p-completeness-compendium},
  author = {Yan Soares Couto},
  year   = {2026}
}

@misc{yufan_you_maintain_2019,
  title    = {Maintain subtree information using link/cut trees},
  url      = {https://codeforces.com/blog/entry/67637},
  author   = {Yufan You},
  optmonth = jun,
  year     = {2019}
}

@article{impagliazzo_complexity_2001,
  title     = {On the complexity of \textit{k}-{SAT}},
  volume    = {62},
  copyright = {https://www.elsevier.com/tdm/userlicense/1.0/},
  issn      = {00220000},
  doi       = {10.1006/jcss.2000.1727},
  language  = {en},
  number    = {2},
  journal   = {Journal of Computer and System Sciences},
  author    = {Impagliazzo, Russell and Paturi, Ramamohan},
  optmonth  = mar,
  year      = {2001},
  pages     = {367--375}
}

@article{impagliazzo_which_2001,
  title     = {Which problems have strongly exponential complexity?},
  volume    = {63},
  copyright = {https://www.elsevier.com/tdm/userlicense/1.0/},
  issn      = {00220000},
  doi       = {10.1006/jcss.2001.1774},
  language  = {en},
  number    = {4},
  journal   = {Journal of Computer and System Sciences},
  author    = {Impagliazzo, Russell and Paturi, Ramamohan and Zane, Francis},
  optmonth  = dec,
  year      = {2001},
  pages     = {512--530}
}

@article{cygan_problems_2016,
  title    = {On problems as hard as {CNF}-{SAT}},
  volume   = {12},
  issn     = {1549-6325, 1549-6333},
  doi      = {10.1145/2925416},
  language = {en},
  number   = {3},
  journal  = {ACM Trans. Algorithms},
  author   = {Cygan, Marek and Dell, Holger and Lokshtanov, Daniel and Marx, Dániel and Nederlof, Jesper and Okamoto, Yoshio and Paturi, Ramamohan and Saurabh, Saket and Wahlström, Magnus},
  optmonth = jun,
  year     = {2016},
  pages    = {1--24}
}

@inproceedings{patrascu_possibility_2010,
  optaddress   = {USA},
  title        = {On the possibility of faster {SAT} algorithms},
  isbn         = {978-0-89871-701-3 978-1-61197-307-5},
  doi          = {10.1137/1.9781611973075.86},
  language     = {en},
  booktitle    = {21st {Annual} {ACM}-{SIAM} {Symposium} on {Discrete} {Algorithms} ({SODA})},
  optpublisher = {Society for Industrial and Applied Mathematics},
  author       = {Pătraşcu, Mihai and Williams, Ryan},
  optmonth     = jan,
  year         = {2010},
  pages        = {1065--1075}
}

@inproceedings{liu_parallel_2022,
  optaddress   = {Philadelphia PA USA},
  title        = {Parallel batch-dynamic algorithms for \textit{k}-core decomposition and related graph problems},
  isbn         = {978-1-4503-9146-7},
  doi          = {10.1145/3490148.3538569},
  language     = {en},
  booktitle    = {34th {ACM} {Symposium} on {Parallelism} in {Algorithms} and {Architectures} ({SPAA})},
  optpublisher = {ACM},
  author       = {Liu, Quanquan C. and Shi, Jessica and Yu, Shangdi and Dhulipala, Laxman and Shun, Julian},
  optmonth     = jul,
  year         = {2022},
  pages        = {191--204}
}

@article{zhang_order_2023,
  title    = {Order based algorithms for the core maintenance problem on edge-weighted graphs},
  volume   = {941},
  issn     = {03043975},
  doi      = {10.1016/j.tcs.2022.11.008},
  language = {en},
  journal  = {Theoretical Computer Science},
  author   = {Zhang, Feiteng and Liu, Bin and Liu, Zhenming and Fang, Qizhi},
  optmonth = jan,
  year     = {2023},
  pages    = {140--155}
}

@book{greenlaw_limits_1995,
  address   = {New York},
  title        = {Limits to {Parallel} {Computation}: {P}-{Completeness} {Theory}},
  isbn         = {978-0-19-508591-4},
  shorttitle   = {Limits to {Parallel} {Computation}},
  publisher = {Oxford University Press},
  author       = {Greenlaw, Raymond and Hoover, H. James and Ruzzo, Walter L.},
  year         = {1995}
}

@techreport{anderson_p-complete_1984,
  optaddress  = {Stanford, CA, USA},
  title       = {A {P}-complete problem and approximations to it},
  institution = {Stanford University},
  author      = {Anderson, Richard and Mayr, Ernst W.},
  year        = {1984}
}

@article{miltersen_complexity_1994,
  title     = {Complexity models for incremental computation},
  volume    = {130},
  copyright = {https://www.elsevier.com/tdm/userlicense/1.0/},
  issn      = {03043975},
  doi       = {10.1016/0304-3975(94)90159-7},
  language  = {en},
  number    = {1},
  journal   = {Theoretical Computer Science},
  author    = {Miltersen, Peter Bro and Subramanian, Sairam and Vitter, Jeffrey Scott and Tamassia, Roberto},
  optmonth  = aug,
  year      = {1994},
  pages     = {203--236}
}

@inproceedings{fredman_cell_1989,
  optaddress   = {Seattle, Washington, United States},
  title        = {The cell probe complexity of dynamic data structures},
  isbn         = {978-0-89791-307-2},
  doi          = {10.1145/73007.73040},
  language     = {en},
  booktitle    = {21st Annual {ACM} Symposium on {Theory} of Computing (STOC)},
  optpublisher = {ACM Press},
  author       = {Fredman, M. and Saks, M.},
  year         = {1989},
  pages        = {345--354}
}

@article{goldschlager_monotone_1977,
  title    = {The monotone and planar circuit value problems are log space complete for {P}},
  volume   = {9},
  issn     = {0163-5700},
  doi      = {10.1145/1008354.1008356},
  language = {en},
  number   = {2},
  journal  = {SIGACT News},
  author   = {Goldschlager, Leslie M.},
  optmonth = jul,
  year     = {1977},
  pages    = {25--29}
}

@phdthesis{patrascu_lower_2008,
  optaddress = {USA},
  type       = {{PhD} {Thesis}},
  title      = {Lower bound techniques for data structures},
  school     = {Massachusetts Institute of Technology},
  author     = {Pătraşcu, Mihai},
  year       = {2008}
}

@article{yao_should_1981,
  title    = {Should Tables Be Sorted?},
  volume   = {28},
  issn     = {0004-5411, 1557-735X},
  doi      = {10.1145/322261.322274},
  language = {en},
  number   = {3},
  journal  = {J. ACM},
  author   = {Yao, Andrew Chi-Chih},
  optmonth = jul,
  year     = {1981},
  OPTnote     = {Publisher: Association for Computing Machinery (ACM)},
  pages    = {615--628}
}

@inproceedings{clifford_new_2015,
author = {Gronlund, Raphael Clifford. Allan and Larsen, Kasper Green},
title = {New Unconditional Hardness Results for Dynamic and Online Problems},
year = {2015},
isbn = {9781467381918},
publisher = {IEEE Computer Society},
address = {USA},
doi = {10.1109/FOCS.2015.71},
booktitle = {Proceedings of the 2015 IEEE 56th Annual Symposium on Foundations of Computer Science (FOCS)},
pages = {1089–1107},
numpages = {19},
series = {FOCS '15}
}

@inproceedings{ahmed_visualisation_2007,
	address = {Sydney, NSW},
	title = {Visualisation and analysis of the internet movie database},
	isbn = {978-1-4244-0808-5},
	doi = {10.1109/APVIS.2007.329304},
	booktitle = {2007 6th {International} {Asia}-{Pacific} {Symposium} on {Visualization}},
	publisher = {IEEE},
	author = {Ahmed, A. and Batagelj, V. and Fu, X. and Hong, S.-H. and Merrick, D. and Mrvar, A.},
	month = feb,
	year = {2007},
	pages = {17--24},
}

@article{gurjar_planarizing_2016,
	title = {Planarizing {Gadgets} for {Perfect} {Matching} {Do} {Not} {Exist}},
	volume = {8},
	issn = {1942-3454, 1942-3462},
	doi = {10.1145/2934310},
	language = {en},
	number = {4},
	journal = {ACM Transactions on Computation Theory},
	author = {Gurjar, Rohit and Korwar, Arpita and Messner, Jochen and Straub, Simon and Thierauf, Thomas},
	month = jul,
	year = {2016},
	pages = {1--15},
}

@article{ramachandran_efficient_1996,
	title = {An {Efficient} {Parallel} {Algorithm} for the {General} {Planar} {Monotone} {Circuit} {Value} {Problem}},
	volume = {25},
	issn = {0097-5397, 1095-7111},
	doi = {10.1137/S0097539793260775},
	language = {en},
	number = {2},
	journal = {SIAM Journal on Computing},
	author = {Ramachandran, Vijaya and Yang, Honghua},
	month = apr,
	year = {1996},
	pages = {312--339},
}

@inproceedings{abboud_popular_2016,
	address = {New Brunswick, NJ, USA},
	title = {Popular {Conjectures} as a {Barrier} for {Dynamic} {Planar} {Graph} {Algorithms}},
	isbn = {978-1-5090-3933-3},
	doi = {10.1109/FOCS.2016.58},
	booktitle = {2016 {IEEE} 57th {Annual} {Symposium} on {Foundations} of {Computer} {Science} ({FOCS})},
	publisher = {IEEE},
	author = {Abboud, Amir and Dahlgaard, Soren},
	month = oct,
	year = {2016},
	pages = {477--486},
}

@article{garey_simplified_1976,
	title = {Some simplified {NP}-complete graph problems},
	volume = {1},
	copyright = {https://www.elsevier.com/tdm/userlicense/1.0/},
	issn = {03043975},
	doi = {10.1016/0304-3975(76)90059-1},
	language = {en},
	number = {3},
	journal = {Theoretical Computer Science},
	author = {Garey, M.R. and Johnson, D.S. and Stockmeyer, L.},
	month = feb,
	year = {1976},
	pages = {237--267},
}

@inproceedings{couto_hardness_2025,
    author = {Couto, Yan S. and Fernandes, Cristina G.},
    title = {Hardness of Dynamic Core and Truss Decompositions},
    year = {2025},
    isbn = {978-3-032-06705-0},
    publisher = {Springer-Verlag},
    address = {Berlin, Heidelberg},
    doi = {10.1007/978-3-032-06706-7_5},
    booktitle = {Proceedings of the 23rd International Workshop on Approximation and Online Algorithms (WAOA)},
    pages = {64--80},
    numpages = {17},
    location = {Warsaw, Poland}
}

@article{mccoll_planar_1981,
	title = {Planar {Crossovers}},
	volume = {C-30},
	copyright = {https://ieeexplore.ieee.org/Xplorehelp/downloads/license-information/IEEE.html},
	issn = {0018-9340, 1557-9956, 2326-3814},
	doi = {10.1109/TC.1981.1675758},
	number = {3},
	journal = {IEEE Transactions on Computers},
	author = {W. F. McColl},
	month = mar,
	year = {1981},
	pages = {223--225},
}

@inproceedings{goldberg_finding_1984,
	title = {Finding a {Maximum} {Density} {Subgraph}},
	url = {https://digicoll.lib.berkeley.edu/record/136696/files/CSD-84-171.pdf},
	author = {Goldberg, Andrew V.},
	year = {1984},
}

@inproceedings{sawlani_near-optimal_2020,
	address = {Chicago IL USA},
	title = {Near-optimal fully dynamic densest subgraph},
	isbn = {978-1-4503-6979-4},
	doi = {10.1145/3357713.3384327},
	language = {en},
	booktitle = {Proceedings of the 52nd {Annual} {ACM} {SIGACT} {Symposium} on {Theory} of {Computing}},
	publisher = {ACM},
	author = {Sawlani, Saurabh and Wang, Junxing},
	month = jun,
	year = {2020},
	pages = {181--193},
}

@article{henzinger_fine-grained_2022,
	title = {Fine-{Grained} {Complexity} {Lower} {Bounds} for {Families} of {Dynamic} {Graphs}},
	volume = {244},
	copyright = {Creative Commons Attribution 4.0 International license, info:eu-repo/semantics/openAccess},
	issn = {1868-8969},
	doi = {10.4230/LIPICS.ESA.2022.65},
	language = {en},
	journal = {LIPIcs, Volume 244, ESA 2022},
	publisher = {Schloss Dagstuhl – Leibniz-Zentrum für Informatik},
	author = {Henzinger, Monika and Paz, Ami and Sricharan, A. R.},
	editor = {Chechik, Shiri and Navarro, Gonzalo and Rotenberg, Eva and Herman, Grzegorz},
	year = {2022},
	note = {Artwork Size: 14 pages, 957027 bytes
ISBN: 9783959772471
Medium: application/pdf},
	pages = {65:1--65:14},
}
\end{document}